\documentclass[reqno, 11pt]{amsart}

\usepackage{amsfonts,amssymb,amsbsy,amsmath,amsthm,enumerate}
\usepackage{txfonts}
\usepackage{eucal}
\usepackage{dsfont}
\usepackage{mathrsfs}
\usepackage{nicefrac}

\usepackage{stackrel}

\usepackage[small,nohug,
]{diagrams}
\diagramstyle[labelstyle=\scriptstyle]
\newarrow{Corresponds}<--->

\usepackage[dvips]{graphicx}
\usepackage{color}

\newtheorem{theorem}{Theorem}[section]
\newtheorem{proposition}[theorem]{Proposition}
\newtheorem{lemma}[theorem]{Lemma}
\newtheorem{corollary}[theorem]{Corollary}
\newtheorem{assumption}[theorem]{Assumption}
\newtheorem{definition}[theorem]{Definition}
\newtheorem{remark}[theorem]{Remark}

\numberwithin{equation}{section}

\numberwithin{figure}{section}
\numberwithin{table}{section}

\newcommand\beq{\begin{equation}}
\newcommand{\bea}{\begin{eqnarray}}
\newcommand{\eea}{\end{eqnarray}}
\newcommand{\beas}{\begin{eqnarray*}}
\newcommand{\eeas}{\end{eqnarray*}}
\newcommand{\beql}{\begin{equation} \label}
\newcommand{\eeq}{\end{equation}}

\newcommand{\R}{\mathbb R}
\newcommand{\N}{\mathbb N}
\newcommand{\C}{\mathbb C}                           

\newcommand{\Z}{\mathbb Z}

\newcommand{\s}[1]{\CMcal{#1}}
\newcommand{\f}[1]{\mathcal{#1}}                  
\newcommand{\bb}[1]{\mathscr{#1}}
\newcommand{\rr}[1]{\mathfrak{#1}}
\newcommand{\n}[1]{\mathds {#1}}

\newcommand{\expo}[1]{{\rm e}^{#1}}                 
\newcommand{\dd}{\,{\rm d}}
\newcommand{\ii}{\,{\rm i}\,}
\newcommand{\ncint}{\mathrel{{\ooalign{$\int$\cr\kern+.07em\raise.15ex\hbox{$\pmb{\scriptstyle-}$}\cr}}}}           \newcommand{\ncpartial}{\mathrel{{\ooalign{$\partial$\cr\kern+.29em\raise.79ex\hbox{$\pmb{\scriptstyle-}$}\cr}}}}

\newcommand{\virg}[1]{\lq\lq#1\rq\rq}                
\newcommand{\ie}{{\sl i.\,e.\,}}
\newcommand{\eg}{{\sl e.\,g.\,}}
\newcommand{\cf}{{\sl cf.\,}}

\begin{document}

\title[On the K-theoretic classification of dynamically stable systems]{
On the K-theoretic classification of dynamically stable systems}


\author[G. De~Nittis]{Giuseppe De Nittis}
\address[De~Nittis]{
Pontificia Universidad Cat\'olica de Chile,
Facultad de Matem\'aticas \& Instituto de F\'{\i}sica,
Santiago, Chile}
\email{gidenittis@mat.uc.cl}

\author[K. Gomi]{Kiyonori Gomi}
\address[Gomi]{Department of Mathematical Sciences, Shinshu University,  Nagano, Japan}
\email{kgomi@math.shinshu-u.ac.jp}

\thanks{{\bf MSC2010}
Primary: 16E20; 
Secondary:  	81T99, 46L80, 	14D21.
}

\thanks{{\bf Keywords.}
$K$-theory, topological phases, Krein spaces,
$\s{PT}$ and $\s{C}$ symmetries.
\vspace{2mm}}


\begin{abstract}
\vspace{-4mm}
This paper deals with the construction of a suitable topological $K$-theory capable of classifying topological phases of dynamically stable systems described by gapped  
$\eta$-self-adjoint operators on a Krein space with indefinite metric $\eta$.
\end{abstract}


\maketitle

\vspace{-5mm}
\tableofcontents

\section{Introduction}
\label{sec:introduction}

This paper deals with the construction of a suitable $K$-theory capable of classifying topological phases of dynamically stable systems described by  
$\eta$-self-adjoint operators on a Krein space with indefinite metric (or fundamental symmetry) $\eta$. The motivations, the  methodology and the main results are described below.

\subsection{Motivations}
\label{sec:introduction_motivation}
In the standard formulation  of  Quantum Mechanics (QM) the physics of a (closed) system is described by 
a separable Hilbert space $\s{H}$ endowed with a \emph{scalar product}\footnote{\label{note:intro_01}In this paper the name  \emph{scalar product} is used as synonyms of \emph{positive definite sesquilinear form}. This is  the same thing as requiring that the associated quadratic form $\f{Q}(\varphi):=\langle\varphi;\varphi\rangle>0$ is strictly positive for non-zero vectors $\varphi\neq 0$.} $\langle\;,\;\rangle$ and by a (time-independent)
self-adjoint operator $H$, called  Hamiltonian. The vectors in $\s{H}$ describe the possible \emph{states} of the system, the scalar product $\langle\;,\;\rangle$ fixes the \emph{probabilistic interpretation} of the theory and the Hamiltonian $H$  generates  the dynamics of the system through
 the Schr\"odinger equation.
The role of the self-adjoint condition $H=H^*$ is twofold; (i) it assures the reality of the spectrum of $H$ which is interpreted as the set of possible energies of the system; (ii) it assures that 
the dynamics
generated by $H$ is implemented by a unitary group  $t\mapsto V_t:=\expo{-\ii t H}$, called the  \emph{time propagator},  which  preserves the geometry of $\s{H}$  and therefore the
probabilistic interpretations.
Bender and coworkers, in a long series of papers (see \eg \cite{bender-boettcher-98,bender-brody-jones-03,bender-07}),   argued that conditions (i) and (ii) are the only crucial  properties which characterize physical (quantum) systems which are  \emph{dynamically stable}.
With this locution one refers to (somehow unphysical) systems which are not subjected to dissipations or losses (of matter, or energy, or information) and   that are not destined to disappear in a remote future or past (\cf \cite{peano-schulz-baldes-17}).   Although the self-adjointness  of the Hamiltonian is a sufficient condition for (i) and (ii)  it is not strictly necessary (a fact already recognized at the early days of QM \cite{dirac-42,pauli-43}) and can be relaxed somehow. A possibility  to formulate rigorously the notion of dynamical stability is as follows:
\begin{definition}[Dynamical stability]\label{def:dyn_stab}
A Hamiltonian $H$ acting on the Hilbert space $\s{H}$ describes a
  \emph{dynamically stable} system  if  there exists a bounded operator $G$, with bounded inverse $G^{-1}$, such that $\tilde{H}:=GHG^{-1}$ is self-adjoint  on $\s{H}$.
      \end{definition}
 
 \medskip
 
 \noindent
Said differently, $H$ is  dynamically stable whenever it is
\emph{similar} to a self-adjoint operator on $\s{H}$.
If this is the case condition (i) about the reality of the spectrum follows by the fact that similar operators have the same spectrum and so
$\sigma(H)=\sigma(\tilde{H})\subseteq\R$. The dynamical condition (ii) follows by the observation that 
the group $t\mapsto V_t:=G^{-1}\expo{-\ii t \tilde{H}}G$ integrates the Schr\"odinger equation
$$
\ii\frac{\partial}{\partial t}\psi\;=\;H\psi\;,\qquad\quad\psi\in\s{D}({H})\;.
$$
and is unitary with respect to the modified scalar product $
\langle \phi, \varphi \rangle_G: = \langle \phi, G^*G \varphi \rangle$ for all $\phi, \varphi\in\s{H}
$. In particular, $H$ turns out to be self-adjoint with respect to the  
scalar product $
\langle \;, \; \rangle_G$. For more details we refer to
Corollary 
\ref{col:func-calc}.

\medskip

Definition \ref{def:dyn_stab} is quite strong and in principle it could be possible to set up weaker and workable reformulations of the notion of dynamical stability. However, Definition 1.1 has been inspired by  physical and mathematical motivations that we will try to retrace below.

\medskip

Bender and coworkers at the end of the 90's observed that 
 certain non-self-adjoint Hamiltonians subjected to a fundamental space-time reflection symmetry (\emph{$\s{PT}$-symmetry})  seem to obey  the two crucial requirements (i) and (ii) that characterize systems with a stable dynamics.
  This observation   opened a new and fruitful line of research on \emph{$\s{PT}$-symmetric Quantum Mechanics} ($\s{PT}$-QM). As the literature on $\s{PT}$-QM is extremely large, and this is not the central focus of the present work, we refer to the recent review papers \cite{mostafazadeh-10,bender-07,bender-15,albeverio-kushel-15}, and references therein, for a complete overview on the subject. Nevertheless, it is worth to mention that intriguing
 toy examples of (one-dimensional) $\s{PT}$-symmetric systems are provided by the family of operators
\begin{equation}\label{eq:mod1}
H_\epsilon\;:=\;-\frac{\dd^2}{\dd x^2}\;-\; \big(\ii f(x)\big)^\epsilon\;,\qquad\quad \epsilon\in\R
\end{equation}
where $f(-x)=-f(x)$ is an odd real-valued function. The special case $f(x)=x$ has been
extensively studied in \cite{bender-07} from a classical and a quantum point of view. A second family of examples is given by the Hamiltonians
\begin{equation}\label{eq:mod2}
\rr{m}_w\;:=\;w(x)\;\frac{\dd^2}{\dd x^2}\;,
\end{equation}
with $w$ a real-valued \emph{weight} function such that $w(x)^2=1$. The special case $w(x)={\rm sgn}(x):=x/|x|$
has been 
studied for the first time in \cite{curgus-najman-95} and recently in \cite{kuzhel-sudilovskaya-17}. Moreover, models  of $\s{PT}$-QM  have recently found application in a wide variety of areas like in \emph{optics} \cite{el-ganainy-makris-christodoulides-musslimani-07} (see also \cite{bender-16} and references therein), in the study of the \emph{atomic diffusion} in inhomogeneous magnetic fields \cite{zhao-schaden-wu-10}, in the theory of \emph{superconducting wires} \cite{rubinstein-sternberg-ma-07,chtchelkatchev-golubov-baturina-vinokur-12} 
and \emph{electronic circuits} \cite{schindler-li-zheng-ellis-kottos-11}, in \emph{quantum information} \cite{croke-15},
in the study of \emph{nonlinear waves} \cite{konotop-yang-zezyulin-16}, \emph{magnon} systems \cite{lee-kottos-shapiro-15,galda-vinokur-16} and \emph{metamaterials} \cite{mostafazadeh-16} ... and probably others.

\medskip

Just after the first works on $\s{PT}$-QM  various authors (like
Mostafazadeh
\cite{mostafazadeh-02,mostafazadeh-03,mostafazadeh-03-2}
or Albeverio and Kuzhel \cite{albeverio-kushel-04,albeverio-kushel-15} just to mention few of them)
recognized that  the concept of $\s{PT}$-QM can be placed in a more general and natural mathematical framework by using the theory of
 \emph{Krein spaces} \cite{azizov-iokhvidov-67,bognar-74} (see also Section \ref{sec:indefinite}). 
In its simplest incarnation
a Krein space can be thought as a separable Hilbert space $\s{H}$ along with a self-adjoint involution $\eta=\eta^*=\eta^{-1}$, called \emph{fundamental symmetry} (or \emph{metric} operator), such that the dimensions of the eigenspaces related to the eigenvalues $\pm1$ are equal. The fundamental symmetry $\eta$ endows 
$\s{H}$ with the non-degenerate,  indefinite inner product $
\langle\langle \phi, \varphi \rangle\rangle_\eta: = \langle \phi, \eta \varphi \rangle$ for all $\phi, \varphi\in\s{H}
$. The \emph{$\eta$-adjoint} \cite{tomita-80,gheondea-88} of a linear operator $H$ 
acting on $\s{H}$ is then defined by 
$$
H^\sharp\;:=\;\eta H^*\eta\;.
$$
Operators which fulfill the conditions $H^\sharp=H$ are called \emph{pseudo-Hermitian} 
\cite{mostafazadeh-02,mostafazadeh-03,mostafazadeh-06,mostafazadeh-10,jones-05,albeverio-kushel-04,sato-hasebe-esaki-kohmoto-12,ghosh-12}, and  sometimes \emph{quasi-Hermitia}  \cite{scholtz-geyer-hahne-92}, in the physical jargon. We prefer to refer to them as
\emph{$\eta$-self-adjoint} in agreement with a more mathematical tradition. It is not hard to see that both operators $H_\epsilon$ in \eqref{eq:mod1} and $\rr{m}_w$ in \eqref{eq:mod2} are $\eta$-self-adjoint with respect to a suitable choice of the fundamental symmetry $\eta$.
For the operators $H_\epsilon$ the appropriate fundamental symmetry is given by $\eta\equiv \s{P}$
where $\s{P}$ is the \emph{parity} operator $(\s{P}\phi)(x):=\phi(-x)$.
In the case of the operator $\rr{m}_w$ the right fundamental symmetry is provided by  $\eta\equiv w(x)$. Further examples of $\eta$-self-adjoint operators, adapted for certain physical applications, will be described in Appendix \ref{sec:maxwell-meta}.

\medskip

In general,  $\eta$-self-adjointness does not suffice to assure the dynamical stability of Definition 
\ref{def:dyn_stab}. For that, something more is required.
One says that an  $\eta$-self-adjoint operator $H$
admits a \emph{$\s{C}$-symmetry} if there is an $\eta$-self-adjoint involution $\Xi$ such that the product $\eta \Xi>0$ is strictly positive and $H\Xi=\Xi H$ (see \eg \cite{kuzhel-09,albeverio-kushel-15,kuzhel-sudilovskaya-17}
 and references therein\footnote{In the standard literature a $\s{C}$-symmetry is usually denoted with the letter $\s{C}$. We prefer to use the greek letter $\Xi$ to  avoid confusion with the complex conjugation denote here with a $C$.}). The theory of $\s{C}$-symmetric $\eta$-self-adjoint operators will be briefly reviewed in Section \ref{sec:C-symm}. Here, we need only to recall a fundamental fact \cite[Theorem 6.3.4]{albeverio-kushel-15}:
 The existence of a $\s{C}$-symmetry for the $\eta$-self-adjoint operator $H$ is equivalent to the dynamical stability according to Definition 
\ref{def:dyn_stab} 
 (see Section \ref{sec:op-C-sym} for  more details). 
Consequently, on the basis of  the  considerations above, we feel justified to make the following:

\medskip

\noindent 
{\bf Working hypothesis.}
\emph{$\eta$-self-adjoint operators that possess a $\s{C}$-symmetry are good (and interesting) candidates for the description of the dynamics of certain {physical} (quantum) systems.}

\medskip

Before continuing, it is fare to mention that  to prove that a given $\eta$-self-adjoint operator admits a $\s{C}$-symmetry is   not a trivial fact in general. Various conditions for the existence and the construction of a $\s{C}$-symmetry have been studied in \cite{kuzhel-09,albeverio-kushel-15,kuzhel-sudilovskaya-17}. 
For instance, it turns out that the operator 
$\rr{m}_w$ of \eqref{eq:mod2} for the case $w(x)=x/|x|$ has  a $\s{C}$-symmetry described
in \cite[Proposition 3.5]{kuzhel-sudilovskaya-17}. On the other hand
the construction of a $\s{C}$-symmetry for the family of operators $H_\epsilon$ in \eqref{eq:mod1}
 is a complicated (still open) problem. For instance, for the special case $f(x)=x$, it has been shown that $\sigma(H_\epsilon)\subseteq [0,+\infty)$ for $2\leqslant \epsilon<4$ \cite[Appendix B]{dorey-dunning-tateo-01} but 
  the form of the corresponding operator implementing  the $\s{C}$-symmetry (if any) is not known.
Because of the complexity of the problem the majority of the available techniques for the construction of $\s{C}$-symmetries are  based on perturbative approximations 
\cite{bender-07}.

\medskip

In the last decade since the groundbreaking discovery of topologically protected states induced by strong spin-orbit interactions, tremendous progress has been made in the study and the understanding
of topological properties of (quantum) matter.
The investigation  of  \emph{topological insulators} (we refer to  \cite{hasan-kane-10,ando-fu-15,chiu-teo-schnyder-ryu-16} for a modern overview and an updated bibliography of the subject) is doubtless
 a current
 \virg{hot topic} in mathematical physics and condensed matter. This evidence justifies the following natural question: 
\emph{Is it possible to generalize 
 the set of results and theories available for the study and the classification of \virg{standard} topological insulators, described by self-adjoint operators,  to more general systems which are dynamically stable in the sense of Definition \ref{def:dyn_stab}?} Of course the idea of investigating the topological properties of specific non-self-adjoint systems is not new in literature. 
The first result in this direction established the  absence of topological phases for certain $\s{PT}$-symmetric Dirac Hamiltonians \cite{hu-hughes-11}. Soon after, the existence of   non-trivial topological phases for non-self-adjoint generalizations of the Kane-Mele model has been proved in \cite{esaki-sato-hasebe-kohmoto-11}. However, these phases possess the unpleasant feature of having complex energies and this is an evidence of the fact that the corresponding Hamiltonian cannot be dynamically stable.
Other examples of non-self-adjoint models with \emph{unstable} (\ie complex energy) topological states have been proposed in \cite{schomerus-13} as a model for one-dimensional photonic crystals and in
\cite{zhu-lu-chen-14} via a non-self-adjoint modification of the Su-Schrieffer-Heeger model.
The first examples of non-self-adjoint models with \emph{stable} (\ie real energy) topological states
have been builded in  \cite{ghosh-12} by using  Dirac-type Hamiltonians and in \cite{yuce-15}
by a non-Hermitian generalization of the Aubry-Andre model. A new non-self-adjoint modification of the 
Su-Schrieffer-Heeger model which shows \emph{stable}  topological states has been recently proposed in \cite{lieu-18}. Finally, the study of \emph{topological edge modes} in $\s{PT}$-symmetric models has gained interest in various physical contexts \cite{leykam-bliokh-huang-chong-nori-17,ke-wang-long-wang-lu-17,peano-schulz-baldes-17}, just to mention some of the most recent contributions.

\medskip

An aspect common to all the works mentioned above  is that they concern with very specific models. In contrast, one of the bigger success, perhaps the greatest, in the study of \virg{standard} topological insulators was the discovery of a unified framework where to organize and classify all the possible topological phases according to the spatial dimensions. This is the so-called \emph{periodic table} for topological insulators. After some initial partial results toward a unified classification of topological insulators 
\cite{altland-zirnbauer-97,schnyder-ryu-furusaki-ludwig-08,qi-hughes-zhang-08}, the mechanism underlying the periodic organization of the topological phases was finally revealed by Kitaev in the seminal paper \cite{kitaev-09}. Here,  Kitaev showed that topological insulators can be organized and classified by using the \emph{$K$-theory} \cite{atiyah-67,karoubi-78} and the  periodic recurrence of the topological phases with respect to the spatial dimension was interpreted as a manifestation of the \emph{Bott periodicity}. Starting from this moment, $K$-Theory became the principal tool in the classification of topological insulators described by self-adjoint models and the initial Kitaev's idea has been generalized and extended in several directions \cite{freed-moore-13,thiang-16,kellendonk-15,kubota-17} producing a great advancement in the understanding of   topological phases of matter.
With these premises
it seems natural to believe that  
there must be 
a more general classification scheme capable of including also non-self-adjoint dynamically stable systems. 
The present work has been  inspired by this idea and its main achievement is:

\medskip

\noindent 
{\bf Main goal.}
\emph{To construct a suitable $K$-theory capable of classifying topological phases of systems described by  
$\eta$-self-adjoint operators that possess a $\s{C}$-symmetry.}

\medskip

At a first glance, based on \cite[Theorem 6.3.4]{albeverio-kushel-15} which proves the equivalence between $\eta$-self-adjoint operators with a $\s{C}$-symmetry and
dynamically stable systems in the sense of of Definition 
\ref{def:dyn_stab},
one can naively think that the usual $K$-theory used for the classification of self-adjoint systems should be 
enough also to classify the type of systems described by the working hypothesis above. Indeed, this is not the case! In fact, as discussed below,  
systems described by $\s{C}$-symmetric
$\eta$-self-adjoint operators posses more structure than simple self-adjoint systems. A very  elementary example that shows this enrichment  is discussed in Remark \ref{rk:more_structure}.

\subsection
{Methodological extension of the CAZ classification scheme}
\label{subsec:metod_CAZ}

In the classification scheme of \virg{standard} topological insulators \cite{altland-zirnbauer-97,schnyder-ryu-furusaki-ludwig-08,qi-hughes-zhang-08} enter two crucial ingredients: The spatial dimension of the system and  a bunch of fundamental \emph{involutive (pseudo-)symmetries}. Let us focus on the role  of the latter. 

\medskip

The nature of the \emph{quantum symmetries} has been described by  Wigner in the seminal paper \cite{wigner-59} (see Section \ref{sect:hilb_wigner} for more details). The \emph{Wigner's theorem} states that the  physical symmetries 
of a quantum systems represented in a Hilbert space $\s{H}$ can be implemented only by
 {unitary} or {anti-unitary} operators\footnote{\label{footnote:A}An operator $A$ on the Hilbert space $\s{H}$ is {anti}-linear when $A(a\phi+b\varphi)=\overline{a}(A\phi)+\overline{b}(A\varphi)$ for all $a,b\in\C$ and $\phi, \varphi\in\s{D}(A)$.
For an {anti}-linear operator the definition of the adjoint needs to be adjusted in order to compensate for the complex conjugation. The adjoint operator of the anti-linear operator $A$  is the anti-linear operator $A^*$ which fulfills $\langle\phi,A\varphi\rangle=\langle\varphi,A^*\phi\rangle$ for all $\varphi\in \s{D}(A)$ and $\phi\in \s{D}(A^*)$.}. The set of all quantum symmetries on $\s{H}$ has a group 
structure and is denoted with
${\tt QS}(\s{H})$. 
The following notion is central in the classification of topological insulators.
\begin{definition}[Involutive quantum symmetry]\label{def:int-quan-inv}
An \emph{involutive quantum symmetry} on the Hilbert space $\s{H}$ is a bounded $\R$-linear map $U:\s{H}\to\s{H}$  which meets:
\begin{itemize}
\item[(i)] $U^{-1}=U^*$,
\vspace{1mm}
\item[(ii)] $U(\ii\n{1})= \varpi \ii\; U$,
\vspace{1mm}
\item[(iii)] $U=\varepsilon^{\frac{1-\varpi}{2}}U^{-1}$,
\end{itemize}
with $\varpi,\varepsilon\in\{-1,+1\}$.
\end{definition}

\medskip

\noindent
Properties (i) and (ii) in Definition \ref{def:int-quan-inv} are simply a way of expressing the fact that $U\in {\tt QS}(\s{H})$, namely that $U$ is a unitary ($\varpi=+ 1$) or anti-unitary  ($\varpi=- 1$)  operator.  
It is fair to observe that the specification about the boundedness of $U$ at the beginning of
Definition \ref{def:int-quan-inv} is  unnecessarily since it is implied by property (i). However, this is a point of distinction with respect to the Krein space case and for this reason we prefer to insist on it.  
Item (iii)  deserves some clarifications.
Usually a symmetry $U$ is  called \emph{involutive} if $U^2=\varepsilon\n{1}$ where $\varepsilon = \pm 1$ and is a \emph{proper} involution if $\varepsilon=+1$ and an \emph{anti}-involution when $\varepsilon=-1$. A simple calculation shows that $(\ii U)^2=\ii^{1+\varpi} \varepsilon\n{1}$. Therefore, the transformation
$U\mapsto \ii U$ induces the changes  $\varepsilon\mapsto -\varepsilon$ in the linear case while it leaves unchanged
the values of $\varepsilon$  in the anti-linear case.  This observation suggests that  in the liner case an 
{anti}-involution can be  always turned into a {proper} involution while the difference between  
{anti}-involution and {proper} involution is persistent in the anti-linear case. As a consequence one can 
use the  convention $U^2=\varepsilon^{\frac{1-\varpi}{2}}\n{1}$ for the signs of the involutive quantum symmetry. The latter relation  justifies property (iii). Summing up,  in the linear case ($\varpi=+1$) an {involutive} quantum symmetry is always a self-adjoint unitary
$U^{-1}=U^*=U$, independently of the  second parameter  $\varepsilon$. Conversely, in the anti-linear case ($\varpi=-1$) one has  \emph{even} symmetries ($\varepsilon=+1$) which are still specified by $U^{-1}=U^*=U$
and \emph{odd} symmetries ($\varepsilon=-1$) which are 
skew-adjoint unitaries $U^{-1}=U^*=-U$.

 \begin{table}[h]
 \centering
 \begin{tabular}{|c|c||c|c|c|c|}
 \hline
{\bf Symmetry} & {\bf Symbol}   & $\quad \boldsymbol{\varpi}\quad $ & $\quad \boldsymbol{\varepsilon} \quad$ & $\quad \boldsymbol{c}\quad$ \\
\hline
 \hline
 \rule[-3mm]{0mm}{9mm}
Proper Linear &  & $+1$   & irr. & $+1$  \\
\hline
\hline
 \rule[-3mm]{0mm}{9mm}
Chiral  & $\chi$  & $+1$ & irr.
& $-1$\\
\hline
 \rule[-3mm]{0mm}{9mm}
Even Time Reversal & $T\;(+)$  & $-1$ & $+1$  
& $+1$\\
\hline
 \rule[-3mm]{0mm}{9mm}
Odd Time Reversal & $T\;(-)$  & $-1$ & $-1$  
& $+1$\\
\hline
 \rule[-3mm]{0mm}{9mm}
Even Particle-Hole & $P\;(+)$  & $-1$ & $+1$  
& $-1$\\
\hline
 \rule[-3mm]{0mm}{9mm}
Odd Particle-Hole & $P\;(-)$  & $-1$ & $-1$  
& $-1$\\
\hline
\end{tabular}\vspace{2mm}
 \caption{\footnotesize The six types of involutive quantum symmetries for a given self-adjoint operator $H=H^*$. In the case of the first two linear symmetries the sign $\varepsilon$ is irrelevant (irr.) as suggested by property (iii) in Definition \eqref{def:int-quan-inv}. The various names come from canonical examples of fundamental symmetries for electronic systems \cite{altland-zirnbauer-97,schnyder-ryu-furusaki-ludwig-08,qi-hughes-zhang-08}.
 }
 \label{tab:int-1}
 \end{table}

The interplay between quantum symmetries and a self-adjoint Hamiltonian $H=H^*$ can be expressed by the 
relation
$$
UH\;=\;c\;HU \;,\qquad c=\pm 1\;
\qquad \text{(commuting vs. anti-commuting dichotomy)}\;.
$$
We say that a $U\in {\tt QS}(\s{H})$ is a \emph{proper} ($c=+1$) quantum symmetry for $H$ if $UH=HU$ and is a \emph{improper} ($c=-1$) quantum symmetry if $UH=-HU$.

\medskip

According to the different signs of the parameters $\varpi,\varepsilon$ and $c$ the involutive quantum symmetries are classified in six families as listed in Table \ref{tab:int-1}. The latter are  the building blocks for the classification of \virg{standard} topological insulators. In the \emph{Cartan-Altland-Zirnbauer (CAZ) scheme} \cite{altland-zirnbauer-97,schnyder-ryu-furusaki-ludwig-08,qi-hughes-zhang-08} gapped electronic 
systems are arranged in ten classes according to the fact that the (self-adjoint) Hamiltonian $H$ that describes the system possesses or violates the \emph{chiral} symmetry $\chi$, the \emph{time reversal} symmetry $T$ or the  \emph{particle-hole} symmetry $P$. Linear proper symmetries do not enter in the topological classification since the commutation relation $[H,U]=0$ assures that $H$ can be decomposed along the eigenspaces of $U$ providing reduced systems which can be independently classified. The ten classes of topological insulators in the CAZ classification are described  in Table \ref{tab:int-2}.
%
 \begin{table}[h]
 \centering
 \begin{tabular}{|c||c|c||c|c|c|c||c|c|c|c|}
 \hline
 & {\bf A}   & {\bf AIII} & {\bf AI} & {\bf AII}& {\bf D} & {\bf C}& {\bf BDI} &{\bf CI}& {\bf DIII} &{\bf CII} \\
\hline
 \hline
 \rule[-3mm]{0mm}{9mm}
$\boldsymbol{\chi}$   & \phantom{0}0\phantom{0} & $1$ & \phantom{0}0\phantom{0}
& \phantom{0}0\phantom{0}& \phantom{0}0\phantom{0}
& \phantom{0}0\phantom{0}& \phantom{0}(1)\phantom{0}
& \phantom{0}(1)\phantom{0}& \phantom{0}(1)\phantom{0}
& \phantom{0}(1)\phantom{0}\\
\hline
 \rule[-3mm]{0mm}{9mm}
$\boldsymbol{T}$ & 0  & \phantom{0}0\phantom{0} & $+1$  
& $-1$& \phantom{0}0\phantom{0}
& \phantom{0}0\phantom{0}&$+1$&$+1$&$-1$&$-1$\\
\hline
 \rule[-3mm]{0mm}{9mm}
$\boldsymbol{P}$ & 0  & 0 & \phantom{0}0\phantom{0}  
& \phantom{0}0\phantom{0}& $+1 $
& $-1$ &$+1$&$-1$&$+1$&$-1$\\
\hline
 \end{tabular}\vspace{2mm}
 \caption{\footnotesize The 
 ten classes of the \emph{Cartan-Altland-Zirnbauer (CAZ) scheme} \cite{altland-zirnbauer-97,schnyder-ryu-furusaki-ludwig-08,qi-hughes-zhang-08} for the topological classification of self-adjoint Hamiltonians. For the case of the chiral symmetry $\chi$ the $1$ means that the symmetry is present and the $0$ that the symmetry is broken. A similar convention is used for the time reversal symmetry $T$ and the  particle-hole symmetry $P$ with the only difference that $+1$ or $-1$ means that the related symmetry is  even or odd, respectively.  
 The first two classes A and AIII share the characteristic of not having anti-linear symmetries and form the subfamily of the \emph{complex} classes. The last eight classes have at least one anti-linear symmetry and provide the subfamily of the \emph{real} classes. The last four classes  possess simultaneously the two anti-linear symmetries $T$ and $P$. In this case the commutation condition $PT=TP$  is tacitly assumed. In particular, the product $PT$  gives rise to a linear symmetry of type $\chi$ and this fact is annotated in the table by the symbol $(1)$.
}
\label{tab:int-2}
 \end{table}

The natural question we want to raise and analyze in this work is  \emph{how the CAZ classification scheme change in the case of 
$\eta$-self-adjoint operators that possess a $\s{C}$-symmetry.} The starting point of the analysis concerns the fate of the {Wigner's theorem} in Krein spaces (see Section \ref{sec:krein_space-Qsim}). 
This question has been investigated by 
Bracci,  Morchio and  Strocchi in \cite{bracci-morchio-strocchi}  where they proved an extended version of the Wigner's theorem adapted to spaces with an indefinite metric. The first important difference with respect to the usual Hilbert space case  is that  quantum symmetries in a Krein space can be implemented also by unbounded operators   (see Remark \ref{rk:unbound_eta-unit_math}). Moreover, this eventuality survives even if one focuses the attention to involutive symmetries (see Remark \ref{rk:gen-C-sym}). For this reason, if one decides to work only with bounded quantum symmetries, such a   choice has to be specified.
Let ${\tt QS}_\eta(\s{H})$ be the set of the bounded \emph{$\eta$-quantum symmetries} acting in the Krein space $\s{H}$ with metric $\eta$.
The second difference, is that $\eta$-quantum symmetries are subjected to two types of 
dichotomies: The usual linear vs. anti-lienar dichotomy expressed by  property (ii) in Definition \ref{def:int-quan-inv} and a new dichotomy, absent in the Hilbert space case, described by the equation
 $U^{-1}=\wp\;U^\sharp=\wp\;\eta U^*\eta$ where $\wp$ is a sign
which distinguishes between \emph{$\eta$-isometries} ($\wp=+1$) and \emph{$\eta$-pseudo-isometries} ($\wp=-1$). 
 The combinations of the different values of $\varpi$ and  $\wp$ lead to four types of
$\eta$-quantum symmetries as listed in Table \ref{tab:B1}.

 \begin{table}[h]
 \centering
 \begin{tabular}{|c||c|c|c|c|c|}
 \hline
$U\in {\tt QS}_\eta(\s{H})$  & $\boldsymbol{\eta}${\bf-unitary}   & $\boldsymbol{\eta}${\bf-anti-unitary} & $\boldsymbol{\eta}${\bf-pseudo-unitary} & $\boldsymbol{\eta}${\bf-pseudo-anti-unitary} \\
\hline
 \hline
 \rule[-3mm]{0mm}{9mm}
$\boldsymbol{\wp}$ & $+1$  & $+1$   & $-1$ & $-1$  \\
\hline
 \rule[-3mm]{0mm}{9mm}
$\boldsymbol{\varpi}$ & $+1$  & $-1$ & $+1$  
& $-1$\\
\hline
\end{tabular}\vspace{2mm}
\caption{\footnotesize The four types of $\eta$-quantum symmetries according to the values of $\varpi$ and $\wp$.
}
\label{tab:B1}
 \end{table}

\medskip

By combining the $\eta$-isometric condition  with the involutive requirement  $U=\varepsilon U^{-1}$ one obtains $U=\wp\varepsilon U^\sharp$.
Again the change $U\mapsto \ii U$
can be used to change the sign of $\varepsilon$ in the linear case but it is important to notice that it 
preserves the sign of $\wp$ independently 
of the linearity or anti-linearity of $U$. The last observations justify the following:
\begin{definition}[Involutive $\eta$-quantum symmetry]\label{def:int-eta-quan-inv}
An  \emph{involutive $\eta$-quantum symmetry} on the Krein space $\s{H}$ with metric $\eta$ is a bounded $\R$-linear map $U:\s{H}\to\s{H}$  which meets:
\begin{itemize}
\item[(i)] $U^{-1}=\wp\; U^\sharp$,
\vspace{1mm}
\item[(ii)] $U(\ii\n{1})= \varpi \ii\; U$,
\vspace{1mm}
\item[(iii)] $U=\varepsilon^{\frac{1-\varpi}{2}}U^{-1}$,
\end{itemize}
with $\wp,\varpi,\varepsilon\in\{-1,+1\}$.
\end{definition}

\medskip

\noindent
A comparison between Definition \ref{def:int-quan-inv} and Definition \ref{def:int-eta-quan-inv} shows that the main  difference is  in the  condition  (i) where the extra sign $\wp$ appears. This sign  is responsible for enriching the zoology of the possible  symmetries. Again, the dichotomy even vs. odd involution is relevant only in the anti-linear case while in the linear case one conventionally assumes that involutions are always of even type. 

\medskip

In Krein spaces the \virg{physical} Hamiltonians are $\eta$-self adjoint $H=H^\sharp$. 
The interplay 
between Hamiltonians and $\eta$-quantum symmetries is again expressed by $UH=c\;HU$
where $c=\pm 1$ discriminates between 
proper (commuting) and improper (anti-commuting) symmetries.

\medskip

According to the different signs of the parameters $\varpi,\wp,\varepsilon,c$ the involutive $\eta$-quantum symmetries are grouped in 12 families which are described  in  Table \ref{tab:int-4}.
A comparison between Table \ref{tab:int-1} and Table \ref{tab:int-4}
shows that in the Krein space case the number of possible symmetries doubles.
Along with the classical symmetries $\chi,T,P$ one has the new \emph{$\eta$-reflecting symmetry} $R$
and the reflecting versions $\chi_R,T_R,P_R$ of the classical symmetries. The set of these seven symmetries, along with the dichotomy even vs. odd  in the anti-linear cases, are the
 the building blocks for the construction of the extended classification scheme for $\eta$-self-adjoint operators.  Again linear proper symmetries do not enter in the  classification scheme since the commutation relations $[H,U]=0$
 and $[\eta,U]=0$ assure that $H$ and $\eta$ can be decomposed along the eigenspaces of $U$ providing reduced systems which can be independently classified. 
 
 \medskip
 
More symmetries imply more symmetry classes. In the Krein space case the  number of possible CAZ classes increases considerably. For instance the two complex classes A and AII described in the first two columns of Table \ref{tab:int-2} correspond to five classes in the Krein space case as showed in Table \ref{tab:int-5}.
The two time-reversal classes AI and AII described in the second and third columns of Table \ref{tab:int-2} split in six new classes in the Krein space case as displayed in Table \ref{tab:int-6}. A similar counting exercise can also be performed for the remaining classes. As a result the total number of (possible) symmetry classesfor $\eta$-self-adjoint operators is much bigger than ten (a complete counting is beyond the scope of this work). Therefore, the so called \emph{ten-fold way} of Table \ref{tab:int-2} which organizes the symmetry classes of standard self-adjoint systems appears to be insufficient in the case of dynamically stable systems described by $\eta$-self-adjoint operators with a $\s{C}$-symmetry.
%
 \begin{table}[h]
 \centering
 \begin{tabular}{|c|c||c|c|c|c|c|}
 \hline
{\bf Symmetry} & {\bf Symbol}   & $\quad \boldsymbol{\varpi}\quad $ & $\quad \boldsymbol{\wp}\quad $ &  $\quad \boldsymbol{\varepsilon} \quad$ & $\quad \boldsymbol{c}\quad$ \\
\hline
 \hline
 \rule[-3mm]{0mm}{9mm}
Proper Linear &  & $+1$ & $+1$ & irr. & $+1$  \\
\hline
\hline
 \rule[-3mm]{0mm}{9mm}
Pure Reflecting & $R$ & $+1$ & $-1$ & irr. & $+1$  \\
\hline
 \rule[-3mm]{0mm}{9mm}
Chiral  & $\chi$  & $+1$ & $+1$ &irr.
& $-1$\\
\hline
 \rule[-3mm]{0mm}{9mm}
Reflecting Chiral  & $\chi_R$  & $+1$ & $-1$ &irr.
& $-1$\\
\hline
 \rule[-3mm]{0mm}{9mm}
Even Time Reversal & $T\;(+)$  & $-1$ & $+1$ & $+1$
& $+1$\\
\hline
\rule[-3mm]{0mm}{9mm}
Even Reflecting Time Reversal & $T_R\;(+)$  & $-1$ & $-1$ & $+1$ 
& $+1$\\
\hline
 \rule[-3mm]{0mm}{9mm}
Odd Time Reversal & $T\;(-)$  & $-1$ & $+1$ & $-1$ 
& $+1$\\
\hline
  \rule[-3mm]{0mm}{9mm}
Odd Reflecting Time Reversal & $T_R\;(-)$  & $-1$ & $-1$ & $-1$
& $+1$\\
\hline
\rule[-3mm]{0mm}{9mm}
Even Particle-Hole & $P\;(+)$  & $-1$ & $+1$ & $+1$  
& $-1$\\
\hline
\rule[-3mm]{0mm}{9mm}
Even Reflecting Particle-Hole & $P_R\;(+)$  & $-1$ & $-1$ & $+1$  
& $-1$\\
\hline
 \rule[-3mm]{0mm}{9mm}
Odd Particle-Hole & $P\;(-)$  & $-1$ &$+1$& $-1$  
& $-1$\\
\hline
\rule[-3mm]{0mm}{9mm}
Odd Reflecting Particle-Hole & $P_R\;(-)$  & $-1$ & $-1$ & $-1$  
& $-1$\\
\hline
\end{tabular}\vspace{2mm}
 \caption{\footnotesize The 12 types of strong involutive $\eta$-quantum symmetries for a given $\eta$-self-adjoint operator $H=H^\sharp$. In the case of the first four linear symmetries the sign $\varepsilon$ is irrelevant (irr.) as suggested by property (iii) in Definition \ref{def:int-eta-quan-inv}. The various names 
 are chosen  in such a way to keep the connection with the standard nomenclature of Table \ref{tab:int-1}.}
 \label{tab:int-4}
 \end{table}
%

 \begin{table}[h]\label{tab:int-5-A}
 \centering
 \begin{tabular}{|c||c|c|c|c|c|c|c|c|c|c|}
 \hline
 & $\boldsymbol{R}$    & $\boldsymbol{\chi}$ & $\boldsymbol{\chi_R}$ & $\boldsymbol{T}$ & $\boldsymbol{T_R}$ & $\boldsymbol{P}$& $\boldsymbol{P}_R$ \\
\hline
 \hline
 \rule[-3mm]{0mm}{9mm}
{\bf A}  & \phantom{0}0\phantom{0} & $0$ & \phantom{0}0\phantom{0}
& \phantom{0}0\phantom{0}& \phantom{0}0\phantom{0}
& \phantom{0}0\phantom{0}& \phantom{0}0\phantom{0}
\\
\hline
\rule[-3mm]{0mm}{9mm}
{\bf R}   & \phantom{0}1\phantom{0} & $0$ & \phantom{0}0\phantom{0}
& \phantom{0}0\phantom{0}& \phantom{0}0\phantom{0}
& \phantom{0}0\phantom{0}& \phantom{0}0\phantom{0}
 \\
\hline
 \rule[-3mm]{0mm}{9mm}
{\bf AIII}   & \phantom{0}0\phantom{0} & $1$ & \phantom{0}0\phantom{0}
& \phantom{0}0\phantom{0}& \phantom{0}0\phantom{0}
& \phantom{0}0\phantom{0}& \phantom{0}0\phantom{0}
\\
\hline
 \rule[-3mm]{0mm}{9mm}
{\bf AIII$_R$} & 0  & \phantom{0}0\phantom{0} & $1$  
& $0$& \phantom{0}0\phantom{0}
& \phantom{0}0\phantom{0}&$0$\\
\hline
 \rule[-3mm]{0mm}{9mm}
{\bf RIII} & 1  & \phantom{0}1\phantom{0} & $(1)$  
& $0$& \phantom{0}0\phantom{0}
& \phantom{0}0\phantom{0}&$0$\\
\hline
 \end{tabular}\vspace{2mm}
 \caption{\footnotesize The five classes which extend the two complex classes (A and AIII) of the standard CAZ scheme of Table 1.2 in the case of $\eta$-self-adjoint Hamiltonians. These are the classes defined by \emph{linear}   
  involutive $\eta$-quantum symmetries. As usual $1$ means that the symmetry is present and  $0$ that the symmetry is broken. In the case of the class RIII the commutation condition $R\chi=\chi R$  is tacitly assumed. 
}
\label{tab:int-5}
 \end{table}

 \begin{table}[h]\label{tab:int-5-C}
 \centering
 \begin{tabular}{|c||c|c|c|c|c|c|c|c|c|c|}
 \hline
 & $\boldsymbol{R}$    & $\boldsymbol{\chi}$ & $\boldsymbol{\chi_R}$ & $\boldsymbol{T}$ & $\boldsymbol{T_R}$ & $\boldsymbol{P}$& $\boldsymbol{P}_R$ \\
\hline
 \hline
\rule[-3mm]{0mm}{9mm}
{\bf AI} & \phantom{0}0\phantom{0}  & \phantom{0}0\phantom{0} & $0$  
& $+1$& \phantom{0}0\phantom{0}
& \phantom{0}0\phantom{0}&$0$\\
\hline
 \rule[-3mm]{0mm}{9mm}
{\bf AI$_R$} & 0  & 0 & \phantom{0}0\phantom{0}  
&$0$& $+1$&0&0\\
\hline
 \rule[-3mm]{0mm}{9mm}
{\bf RI} & 1  & 0 & \phantom{0}0\phantom{0}  
&$+1$& $(+1)$&0&0\\
\hline
 \rule[-3mm]{0mm}{9mm}
{\bf AII} & 0  & 0 & \phantom{0}0\phantom{0}  
& $-1$& $0$
& $0$ &$0$\\
\hline
\rule[-3mm]{0mm}{9mm}
{\bf AII$_R$} & 0  & 0 & \phantom{0}0\phantom{0}  
& $0$& $-1$
& $0$ &$0$\\
\hline
\rule[-3mm]{0mm}{9mm}
{\bf RII} & $1$  & $0$ & \phantom{0}0\phantom{0}  
& $-1$& $(-1) $
& $0$ &$0$\\
\hline
 \end{tabular}\vspace{2mm}
 \caption{\footnotesize The six classes which extend the two time reversal classes (AI and AII) of the standard CAZ scheme of Table 1.2 for the case of $\eta$-self-adjoint Hamiltonians. These are the classes defined by  
  involutive $\eta$-quantum symmetries with $c=+1$. As usual $1$ means that the symmetry is present and  $0$ that the symmetry is broken. In the case of the anti-linear symmetries $T$ and $T_R$ the signs
 $+1$ or $-1$ means that the related symmetry is  even or odd, respectively. 
  In the case of the classes RI and RII the commutation condition $RT=T R$  is tacitly assumed.}
  \label{tab:int-6}
 \end{table}

\subsection
{Methodological extension of the $K$-theory classification scheme}
Each symmetry class in the CAZ scheme of Table \ref{tab:int-2} supports a certain number of possible topological phases. The enumeration of these phases can be performed by different incarnations of the $K$-theory, according to the specificities and the generality of the situation under consideration   \cite{freed-moore-13,thiang-16,kellendonk-15,kubota-17}. In the case of the description of systems described by self-adjoint gapped Hamiltonians the relevant $K$-theory has been defined and studied by  Freed and Moore \cite{freed-moore-13}.

\medskip

There are two key points in the construction of the $K$-theory classification proposed by Freed and Moore:
 One key point is that the self-adjoint Hamiltonian, describing a gapped quantum systems, plays the role of a 
 \emph{gradation}  (namely a self-adjoint involutions)  in the sense of the  Karoubi's construction of $K$-theory \cite{karoubi-78}. The second key point is that one can incorporate various  quantum symmetries of the systems into a modified version of the Karoubi's $K$-theory. For example, the incorporation of the time reversal symmetry $T$ leads to the Atiyah's $KR$-theory \cite{atiyah-66} in the even case and to the Dupont's $KQ$-theory \cite{dupont-69}
 in the odd case. The \emph{Freed-Moore $K$-theory}   \cite{freed-moore-13}  unifies and generalizes  the $KR$-theory and the $KQ$-theory to an arbitrary number of quantum symmetries.

\medskip

Let us be a bit more precise on the description of the crucial ideas for the construction of  the $K$-theory of Freed-Moore. First of all the gradation. Given a self-adjoint Hamiltonian $H$ with a \emph{gap} around the energy $\lambda\notin\sigma(H)$ one can construct via spectral calculus the operator
\begin{equation}\label{eq:intro_grad}
\Gamma_H\;:=\;\rm{sgn}(H)\;=\;\frac{H-\lambda\n{1}}{|H-\lambda\n{1}|}\;=\;\frac{H-\lambda\n{1}}{\sqrt{(H-\lambda\n{1})^2}}\;.
\end{equation}
By construction $\Gamma_H$ is a gradation of $\s{H}$ since $\Gamma_H=\Gamma_H^*=\Gamma_H^{-1}$ (\cf Remark \ref{rk:grad-gap}). Now the symmetries. Let $T,P\in{\tt QS}(\s{H})$ be respectively a time reversal symmetry and a  particle-hole symmetry for $H$, and hence for $\Gamma_H$, in the sense of Table 1.1. The operator $T$ and $P$ are anti-linear and 
by general arguments it is not necessary to consider other anti-linear symmetries \cite[Lemma 6.1]{thiang-16}. 
Moreover, one can also assume the commutativity of $T$ and $P$ \cite[Proposition 6.2]{thiang-16}.  This fact allows to define also the chiral symmetry $\chi:=TP=PT$. By following \cite[Section 6.1]{thiang-16} let us call  \emph{$PT$-group} the subgroup of ${\tt QS}(\s{H})$ given by $\{\n{1},T,P,\chi\}\simeq \Z_2\times\Z_2$ where the isomorphism is generated by the signs $\varpi$ and $c$.  
This subgroup implements different representations of  appropriate Clifford algebras $C\ell^{r,s}$ in function of the signs  of $T^2$ and $P^2$. Moreover, these representations can be graded or ungraded with respect to the  gradation $\Gamma_H$. The possible Clifford actions representable by the $PT$-group 
are ten and are in one-to-one correspondence with the  ten classes of the CAZ  scheme described in Table 1.2 (for the details see \cite[Section 6.1]{thiang-16}). A $K$-theory capable of classifying at the same time all the ten CAZ classes must be endowed with a graded structure coming from the gapped Hamiltonian $H$ and must be compatible with (graded or ungraded) Clifford actions implemented by unitary or anti-unitary operators associated with the $PT$-group (the so called \emph{PUA-representatios} \cite{parthasarathy-69,thiang-16}). This is exactly what the $K$-theory of Freed and Moore does  \cite{freed-moore-13}. A more  precise presentation of the Freed-Moore $K$-theory is postponed in Section \ref{sec:karubi_vs_freed}.

\medskip

How the picture sketched above changes in presence of a Krein space structure given by an indefinite metric $\eta$? In this case the role of the gapped self-adjoint operators is played by
 gapped $\eta$-self-adjoint operators with a  
$\s{C}$-symmetry $\Xi$. However, in order to extend in a straightforward way the $K$-theory picture sketched above we will make use of the following representation:
\begin{theorem}\label{theo_fund_rep_res}
Let  $(\s{H},\langle\langle\;,\;\rangle\rangle_\eta)$ be the Krein space associated to a fundamental symmetry  $\eta^*=\eta=\eta^{-1}$. Let $H$ be an $\eta$-self-adjoint operator  which admits a $\s{C}$-symmetry $\Xi$ and $U\in {\tt QS}_\eta(\s{H})$  a bounded $\eta$-quantum symmetry of type $(\varpi,\wp)$. Then:
\begin{itemize}
\item[(1)] There is a  bounded $\eta$-unitary operator $G_\Xi\in {\tt QS}_\eta(\s{H})$ such that  the transformed operator
$$
\tilde{H}\;:=\;G_\Xi H G_\Xi^{-1}
$$
is self-adjoint and $\eta$-self adjoint, \ie $\tilde{H}^*=\tilde{H}=\tilde{H}^\sharp$.
\vspace{1mm}
\item[(2)] The transformed  symmetry
$$
\tilde{U}\;:=\;G_\Xi U G_\Xi^{-1}\;\in\;{\tt QS}_\eta(\s{H})
$$
is still of type $(\varpi,\wp)$. Moreover
\begin{equation}\label{eq:triv_one}
UH\;=\;cHU\qquad\Leftrightarrow\qquad\tilde{U}\tilde{H}\;=\;c\tilde{H}\tilde{U}\;.
\end{equation}
\item[(3)] 
Assume in addition  that $U$ is 
an involutive $\eta$-quantum symmetry (in the sense of Definition \ref{def:int-eta-quan-inv})  characterized by the signs $(\varpi,\wp,\varepsilon)$ which is \emph{$\Xi$-compatible} in the sense that
\begin{equation}\label{eq:triv_one_two}
U\Xi\;=\;\wp\; \Xi U\;.
\end{equation}
Then  $\tilde{U}$ is a
an involutive $\eta$-quantum symmetry of  type $(\varpi,\wp,\varepsilon)$
which fulfills
$$
\tilde{U}^{-1}\;=\;\tilde{U}^*\;,\qquad\quad \tilde{U}\eta=\wp\; \eta \tilde{U}\;.
$$
In particular  $\tilde{U}$ is a unitary ($\wp=+1$) or an anti-unitary ($\wp=-1$) operator on $\s{H}$.
\end{itemize}
\end{theorem}
\proof[{Proof} (sketch of)]
Item (1)  is proved in  Proposition \ref{prop:H-Xi} and implies that $\tilde{H}$ is a self-adjoint operator such that $\tilde{H}\eta=\eta\tilde{H}$. Item (2) is quite immediate. In fact $U$ and $G_\Xi$ are both elements of the group 
${\tt QS}_\eta(\s{H})$ and so also their composition $\tilde{U}$ belongs to ${\tt QS}_\eta(\s{H})$. Since $G_\Xi$ is linear, it follows that the linear or anti-linear nature  of $\tilde{U}$  depends only on 
${U}$. Similar argument for the sign $\wp$. A direct computation provides $\tilde{U}^\sharp=G_\Xi U^\sharp G_\Xi^{-1}$ and so $\tilde{U}^\sharp=\wp\tilde{U}^{-1}$ if and only if ${U}^\sharp=\wp{U}^{-1}$. Finally equation \eqref{eq:triv_one} is a direct consequence of the definition of $\tilde{U}$ and $\tilde{H}$. 
A general version of  item (3) is proved in Proposition \ref{prop:reduction_symmetry}. Anyway, in view of item (2) one has that $\tilde{U}$ is an $\eta$-quantum symmetry of type $(\varpi,\wp)$. Moreover a direct computation $\tilde{U}^2=G_\Xi U^2 G_\Xi^{-1}=\varepsilon\n{1}$ shows that $\tilde{U}$ is an involution of parity $\varepsilon$. Finally
 let us observe that   
 $\tilde{U}^{-1}=G_\Xi U^{-1} G_\Xi^{-1}$ and $\tilde{U}^{*}=G_\Xi^{-1} U^* G_\Xi$ where the condition $G_\Xi=G_\Xi^*$ has been used. Then the condition $\tilde{U}^{-1}=\tilde{U}^{*}$ is equivalent to $U^*G_\Xi^2=G_\Xi^2U^{-1}$. Since $G_\Xi^2=\eta\Xi$ (\cf Lemma \ref{lemma_reduc_Xi}) the last equality is guaranteed by $G_\Xi^2U^{-1}=\wp \eta U^{-1}\Xi$
which follows from \eqref{eq:triv_one_two} along with 
$U^*G_\Xi^2= \eta U^\sharp\Xi=\wp \eta U^{-1}\Xi$
which is ensured by Definition \ref{def:int-eta-quan-inv}. The condition $\tilde{U}^{-1}=\wp \eta \tilde{U}^*\eta$ along with the unitarity $\tilde{U}^{-1}=\tilde{U}^*$ implies the last relation $\tilde{U}\eta=\wp \eta \tilde{U}$.
\qed

\medskip

The representation described in Theorem \ref{theo_fund_rep_res} allows to reconsider the problem of the construction of a $K$-theory for $\eta$-self-adjoint  operators with a  
$\s{C}$-symmetry in the following terms: First of all, instead of  \emph{single} gapped Hamiltonians, we need to consider \emph{pairs} $(\Gamma_H,\eta)$ given by a  gradation $\Gamma_H$ (associated with a gapped Hamiltonian $H$ according to \eqref{eq:intro_grad}) and   an indefinite metric $\eta$ constrained by the condition $\Gamma_H\eta=\eta\Gamma_H$. Secondly, the fundamental symmetries are described by the \emph{$PTR$-group} which is the subgroup of ${\tt QS}_\eta(\s{H})\cap {\tt QS}(\s{H})$ given by $\{\n{1},T,P,\chi,R,T_R,P_R,\chi_R\}\simeq \Z_2\times\Z_2\times\Z_2$ where the isomorphism is generated by the signs $\varpi,\wp$ and $c$. The \emph{new} symmetries specified by $\wp=-1$ can be defined in terms of  the fundamental  {reflecting symmetry} $R$, the  time reversal symmetry $T$ and a particle-hole symmetry $P$ according to $T_R:=RT=TR$, $P_R:=PR=RP$ and $\chi_R=R\chi=\chi R$.
In this framework the extension of the  Freed-Moore $K$-theory is straightforward, but not trivial. The extended $K$-theory must be capable of classifying pairs $(\Gamma_H,\eta)$ and must be compatible with (graded or ungraded) Clifford actions implemented by unitary or anti-unitary operators associated with the $PTR$-group.
The details of this construction, which is the main goal of this work, are described in Section \ref{sec:eta-K-theory}.

\subsection
{Some open question}
This work contains
a preliminary and foundational investigation around the problem of the $K$-theory classification of topological phases for $\eta$-self-adjoint operators which are  dynamically stable. In the formulation developed here some strong assumptions have been formulated and there are several open questions which deserve future investigation. Among the latter, let us point out on the following questions: \emph{What is the fate of the {ten-fold way} in the extended framework presented in this work?
It is possible to consider the problem of the classification of topological phases induced by involutive $\eta$-quantum symmetries which are not necessarily bounded?} One  last open question is presented at the end of Appendix \ref{sec:maxwell-meta} and concerns with the proof that certain $\eta$-self-adjoint operators admit a $\s{C}$-symmetry.

\subsection
{Structure of the paper}
The paper is organized as follows: In {\bf Section \ref{sec:indefinite}} the basic facts of the theory of Krein spaces are reviewed. In particular the notion of $\eta$-self-adjointness and the different types of $\eta$-isometries are presented.
{\bf Section \ref{sec:Q-sym}} is devoted to the presentation of the extended version of the Wigner's theorem on Krein spaces. In {\bf Section \ref{sec:C-symm}} the notion of  $\s{C}$-symmetry is introduced and the most relevant properties of $\s{C}$-symmetric $\eta$-self-adjoint operators are proved.
In {\bf Section \ref{sec:freed_moore_K}} the Freed-Moore $K$-theory is firstly reformulated  in terms of Karoubi's gradations and  then adapted for the classification of $\s{C}$-symmetric $\eta$-self-adjoint operators. In {\bf Appendix \ref{sec:maxwell-meta}} 
a class of $\s{C}$-symmetric $\eta$-self-adjoint operators, relevant for physical applications, is described.
{\bf Appendix \ref{app:explicit-2}} is devoted to the explicit classification of  $\s{C}$-symmetric $\eta$-self-adjoint matrix on $\C^2$.

\medskip
\noindent
{\bf Acknowledgements.} 
GD's research is supported
 by
the  grant \emph{Iniciaci\'{o}n en Investigaci\'{o}n 2015} - $\text{N}^{\text{o}}$ 11150143 funded  by FONDECYT.	 KG's research is supported by 
the JSPS KAKENHI Grant Number 15K04871.
\medskip


\section{Elements of the Krein  space theory}
\label{sec:indefinite}

From a mathematical point of view the study of indefinite metric spaces in finite dimension started by G. Frobenius while the
  interest for the infinite dimension case 
has been pioneered by 
the russian mathematicians \cite{pontryagin-44,iokhvidov-krein-56,iokhvidov-krein-59,ginzburg-iokhvidov-62,krein-langer-63,krein-65}.
We refer to their papers (cited before), as well as to the two general monographs \cite{azizov-iokhvidov-67,bognar-74},
 for a general  exposition of the results obtained in that field.


\subsection{Indefinite metric spaces}
In this paper, the pair $(\s{H},\langle\;,\;\rangle)$
 is always assumed to describe a \emph{separable complex} Hilbert space $\s{H}$ with a 
 \emph{scalar product}
 $\langle\;,\;\rangle$ (\cf Note \ref{note:intro_01}). We will use the symbol $\bb{B}(\s{H})$ to denote the set of bounded operators on $\s{H}$, namely 
 the set of the linear or \emph{anti}-linear maps $A:\s{H}\to\s{H}$ (\cf Note \ref{footnote:A}) such that
\begin{equation}\label{eq:op_norm}
\|A\|\;:=\;\sup_{\psi\in\s{H}\setminus\{0\}}\left(\frac{\langle A\psi,A\psi\rangle}{\langle \psi,\psi\rangle}\right)^{\frac{1}{2}}\;<\;\infty\;.
\end{equation}
 We point out that the definition of boundedness is not affected by the dichotomy between linearity and anti-linearity. 
  The subset of the \emph{linear} bounded operators will be denoted by the symbol $\bb{B}_{\rm lin}(\s{H})$.
 
 \medskip

The notion of  indefinite metric is introduced in the following way:

\begin{definition}[Indefinite metric space]\label{def:001}
An indefinite metric space $(\s{H},\langle\langle\;,\;\rangle\rangle_\eta)$ is the datum of:
\begin{itemize}
\item[(i)] An  underlying  Hilbert space $(\s{H},\langle\;,\;\rangle)$.\vspace{1mm}

\item[(ii)]
A non-degenerate,  possibly indefinite, inner product $\langle\langle\;,\;\rangle\rangle_\eta$ associated by the prescription
$$
\langle\langle \phi, \varphi \rangle\rangle_\eta\;: =\; \langle \phi, \eta \varphi \rangle\;,\qquad\quad \forall \phi, \varphi\in\s{H}
$$
to a bounded linear operator $\eta\in\bb{B}_{\rm lin}(\s{H})$ which is self-adjoint $\eta=\eta^*$, and which admits   a bounded inverse $\eta^{-1}\in \bb{B}_{\rm lin}(\s{H})$.
\end{itemize}
\end{definition}

\medskip

\noindent
The self-adjoint condition for $\eta$ assures that the quadratic form 
$$
\f{Q}_\eta(\varphi)\;:=\;\langle\langle \varphi, \varphi \rangle\rangle_\eta\;=\;\langle\varphi,\eta\,\varphi\rangle\;,\qquad\quad   \varphi\in\s{H}
$$ 
is real-valued. The invertibility of $\eta$, which is equivalent to the fact that $\eta\varphi=0$ if and only if $\varphi=0$, assures that $\langle\langle\;,\;\rangle\rangle_\eta$ is non-degenerate. 

\medskip

\begin{remark}[Fundamental symmetry]\label{rk:fund_sym}{\upshape
Without loss of generality an indefinite metric space can be assumed to have a \emph{metric operator} which meets the conditions 
\begin{equation}\label{eq:fund_sym}
\eta^*\;=\;\eta\;=\;\eta^{-1}\;.
\end{equation}
To see this, let $\eta'$ be a generic metric operator and  build via functional calculus the operators $w:=\sqrt{|\eta'|}$ and $\eta:=\frac{\eta'}{|\eta'|}$.
The map $U_w:\varphi\mapsto w\varphi$ is a unitary equivalence of the
 Hilbert spaces $(\s{H},\langle\;,\;\rangle)$  and $(\s{H},\langle\;,\;\rangle_{w^{-2}})$. The metric operator $\eta'$ on  $(\s{H},\langle\;,\;\rangle_{w^{-2}})$ gives rise to an isometry
$$
\langle\langle\phi,\varphi\rangle\rangle_{\eta'}\;=\;\langle\phi,\eta'\varphi\rangle_{w^{-2}}
\;=\;\langle \varphi,\eta\,\varphi\rangle\;=\;\langle\langle\phi,\varphi\rangle\rangle_{\eta}\;,\qquad\quad \forall \phi, \varphi\in\s{H}
$$
and by construction $\eta$ fulfills the conditions \eqref{eq:fund_sym}.
Operators which verify \eqref{eq:fund_sym} are usually called \emph{fundamental symmetries}. 
}\hfill $\blacktriangleleft$
\end{remark}

\medskip 

In  view of Remark \ref{rk:fund_sym} we will always assume henceforth that the indefinite metric is generated by a  fundamental symmetry $\eta^*=\eta=\eta^{-1}$. This fact allows us to introduce the corresponding self-adjoint projections
$$
P_+\;:=\;\frac{\n{1}+\eta}{2}\;,\qquad\quad P_-\;:=\;\frac{\n{1}-\eta}{2}
$$
that determine the \emph{fundamental} orthogonal decomposition (with respect to the scalar product) of $\s{H}$
$$
\s{H}\;=\;\s{H}_+\;\oplus\;\s{H}_-\;,\qquad\quad \s{H}_\pm\;:=\;{\rm Ran}\;P_\pm\;=\;P_\pm\s{H}\;.
$$
It is not difficult to prove  that the subspaces $\s{H}_+$ and $\s{H}_-$ are also $\eta$-orthogonal. To see this, let $\phi_\pm\in\s{H}_\pm$ and compute
$$
\langle\langle\phi_+,\phi_-\rangle\rangle_{\eta}\;=\;\langle P_+\phi_+,\eta\, P_-\phi_-\rangle\;=\;\langle \phi_+,P_+\eta P_-\,\phi_-\rangle\;=\;0
$$ 
where we used the fact that $P_\pm$ are self-adjoint projections, $P_+\eta P_-=\pm P_+P_-=0$. Then we can also write the fundamental  decomposition as
$$
\s{H}\;=\;\s{H}_+\;[\oplus]_\eta\;\s{H}_-
$$
where the symbol $[\oplus]_\eta$ denotes the orthogonality with respect to the inner product $\langle\langle\;,\;\rangle\rangle_\eta$.

\medskip

Let $\kappa_\pm:={\rm dim} \s{H}_\pm$. The  \emph{rank of indefiniteness} is defined by 
$$
\kappa\;:=\;\min\;\{\kappa_+,\kappa_-\}\;\in\; \N\;\cup\;\{\infty\}\;.
$$
The case $\kappa=0$ coincides with $\eta=\n{1}$ or $\eta=-\n{1}$. This case is trivial since the resulting metric $\langle\langle\;,\;\rangle\rangle_\eta$ coincides, at worst
up to a sign,  with the initial scalar product $\langle\;,\;\rangle$. For this reason  we will always assume henceforth that $\kappa>0$.
\begin{definition}[Krein space]\label{def:krein}
The indefinite metric space $(\s{H},\langle\langle\;,\;\rangle\rangle_\eta)$ associated to a fundamental symmetry $\eta^*=\eta=\eta^{-1}$ is called:
\begin{itemize}
\item  \emph{Pontryagin space} if $0<\kappa<\infty$;
\vspace{1mm}
\item  \emph{Krein space} if $\kappa=\infty$.
\end{itemize}
\end{definition}
\noindent 
The possibility to define Krein spaces as in Definition \ref{def:krein}, namely as  indefinite metric spaces with rank $\kappa=\infty$ is based on a crucial result by  
H. Langer \cite{langer-62}.  In this work we will use the name Krein space also for  spaces of \emph{finite} even dimension such that $\kappa_+=\kappa_-$.

\medskip

The notion of self-adjointness  has a natural adaptation in the case of indefinite metric spaces. First of all
let us recall that  according to the existing literature (see \eg \cite{tomita-80,gheondea-88}) the \emph{$\eta$-adjoint} of a densely defined (possible unbounded) operator $A$ is by definition
$$
A^\sharp\;:=\;\eta A^* \eta\;.
$$

\begin{definition}[$\eta$-self-adjoint operators] \label{dfn:eta_self_adjoint}
Let  $(\s{H},\langle\langle\;,\;\rangle\rangle_\eta)$ be the indefinite metric space associated to a fundamental symmetry $\eta^*=\eta=\eta^{-1}$. A densely defined (possible unbounded) linear operator $H:\s{H}\to\s{H}$ is called 
$\eta$-symmetric if 
\begin{equation}\label{eq:eta-symm_00-f01}
 H\; \subseteq\; \;H^\sharp
\end{equation}
and
\emph{$\eta$-self-adjoint} if 
\begin{equation}\label{eq:eta-self01}
 H\; =\; \;H^\sharp\;.
\end{equation}
 The notation
$$
{\tt H}_\eta(\s{H})\;:=\;\left\{H\in\bb{B}_{\rm lin}(\s{H})\ |\  H = H^\sharp\right\}
$$
 is used to denote the set of bounded $\eta$-self-adjoint  operators on 
$\s{H}$.
\end{definition}

\medskip
From its very definition it follows that $H$ is $\eta$-self-adjoint if and only if
$$
\langle\langle\phi,H\varphi\rangle\rangle_\eta\;=\;
\langle\phi,\eta H\varphi\rangle\;=\;\langle H\phi, \eta\varphi\rangle\;=\;
\langle\langle H\phi,\varphi\rangle\rangle_\eta,\qquad\forall\; \phi,\varphi\in\s{D}(H)\;
$$
which shows that $H$ is \virg{self-adjoint} with respect to the inner product $\langle\langle\;,\;\rangle\rangle_\eta$. Moreover, the relation $H^*=\eta H \eta$ says that  $H$ and  $H^*$ are unitarily equivalent and, in turn, that $\eta$ is a bijection between the dense domains  $\s{D}(H)$ and $\s{D}(H^*)$. Finally, from the equality 
$$
\langle\phi,\eta H\varphi\rangle\;=\;\langle H\phi, \eta\varphi\rangle\;=\;\langle \eta H\phi, \varphi\rangle,\qquad\forall\; \phi,\varphi\in\s{D}(H)
$$
one also infers the relation
\begin{equation}\label{eq:eta-self02}
\eta H\; =\; (\eta H)^*\;.
\end{equation}
%


\subsection{Unitray vs. $\eta$-unitary operators}
\label{sect:anti-lin}
A densely defined linear operator $U$ acting on the Hilbert $(\s{H},\langle\;,\;\rangle)$ is unitary if
$$
\langle\phi,\varphi\rangle\;=\; \langle U\phi,U\varphi\rangle\;,\qquad\quad \forall \phi, \varphi\in\s{D}(U)\;.
$$
This condition immediately allows to conclude that $U$ is bounded in norm by 1 
in its domain $\s{D}(U)$ and so can be extended to a bounded operator on $\s{H}$. Moreover, $U$ turns out to be 
 invertible with inverse $U^{-1}=U^*$. An operator $U$ on  $(\s{H},\langle\;,\;\rangle)$ is \emph{anti}-unitary if it is \emph{anti}-linear
 and 
$$
\langle\phi,\varphi\rangle\;=\; \langle U\varphi,U\phi\rangle\;,\qquad\quad \forall \phi, \varphi\in\s{H}\;.
$$
Every anti-linear operator $U$ is necessarily bounded and invertible. In fact, given    a \emph{complex conjugation} $C$ on $\s{H}$, one has that the anti-linear operator $U$ can be represented as the product  $U=CW_U$ of the complex conjugation $C$  times the  unitary $W_U:=CU$. It is worth to recall that a complex conjugation
is any \emph{anti}-linear map which admits a \emph{real} (orthonormal) basis of vectors in the sense that there exists a basis $\{\psi_j\}_{j\in\N}\subset\s{H}$ such that $C\psi_j=\psi_j$. As a consequence one has that  $C^*=C=C^{-1}$. 

\medskip

Things become less simple in the case of indefinite metric space $(\s{H},\langle\langle\;,\;\rangle\rangle_\eta)$. First of all it should be noted that the requirement
\begin{equation}
\langle\langle\phi,\varphi\rangle\rangle_\eta\;=\;\langle\langle U\phi,U\varphi\rangle\rangle_\eta\;,\qquad\quad \forall \phi, \varphi\in\s{D}(U)\;.
\end{equation}
does \emph{not}
allow to conclude, as in the positive metric case, that $U$ is bounded  when $\s{H}$ is infinte dimensional \cite{azizov-iokhvidov-67,shmulyan-74,tomita-80,gheondea-88}.  This is a fundamental difference which makes the notion of unitarity  more difficult in the case of indefinite metric spaces. For these reasons, in the mathematical literature on indefinite metric spaces  the boundedness property is part of the definition of the unitarity of an operator (see \eg \cite[Chapter 2, Definition 5.1]{azizov-iokhvidov-67}). However, since this property is in general not justificable on the basis of purely physical considerations and it is in general not
shared by operators occurring in the physical applications (\cf Remark \ref{rk:unbound_eta-unit_math}) we prefer to omit it, at least in this general preliminary presentation. The following definition  is borrowed from \cite[Definitions 2 \& 3]{bracci-morchio-strocchi}:
\begin{definition}[$\eta$-unitary and $\eta$-anti-unitaryoperators]\label{def:eta-unit-op}
Let  $(\s{H},\langle\langle\;,\;\rangle\rangle_\eta)$ be the indefinite metric space associate to a fundamental symmetry $\eta^*=\eta=\eta^{-1}$. Let $U:\s{H}\to \s{H}$ be an operator (not linear a priori) with domain $\s{D}(U)\subseteq \s{H}$ and range $\s{R}(U)\subseteq \s{H}$. Assume that:
\begin{itemize}
\item[(i)] Domain and range are dense sets, \ie $\overline{\s{D}(U)}= \s{H}=\overline{\s{R}(U)}$;
\end{itemize}
The operator $U$ is called \emph{$\eta$-unitary} if:
\begin{itemize}
\item[(ii)]
The condition 
\begin{equation}\label{eq:eta-unit-001}
\langle\langle\phi,\varphi\rangle\rangle_\eta\;=\;\langle\langle U\phi, U\varphi,\rangle\rangle_\eta\;,\qquad\quad \forall \phi, \varphi\in\s{D}(U)\;
\end{equation}
holds true.
\end{itemize}
It is called \emph{$\eta$-anti-unitary} if:
\begin{itemize}
\item[(ii')]
The condition 
\begin{equation}
\langle\langle\phi,\varphi\rangle\rangle_\eta\;=\;\langle\langle U\varphi, U\phi,\rangle\rangle_\eta\;,\qquad\quad \forall \phi, \varphi\in\s{D}(U)\;
\end{equation}
holds true.
\end{itemize}
\end{definition}

\begin{remark}[Comparison with the mathematical nomenclature]\label{rk:nomenc_eta-unit_math}{\upshape
Possible unbounded operators which meet condition (ii) of Definition \ref{def:eta-unit-op} are usually called \emph{$\eta$-isometric} (see \eg \cite[Chapter 2, Definition 5.1]{azizov-iokhvidov-67}). However, the notion of $\eta$-unitarity is slightly stronger in view of the condition (i) which assures the density of the domain  and  the range of $U$. Let us observe that the density of $\s{D}(U)$ and condition (ii) are enough to prove that $U^{-1}\subseteq\eta U^*\eta=U^\sharp$ \cite[Proposition 2]{bracci-morchio-strocchi}. Sometime, the stronger condition 
$U^{-1}=U^\sharp$ is used as a definition of $\eta$-unitary operator (see \eg \cite{tomita-80,gheondea-88}). Clearly, whenever $U$ is a bounded operator all these definitions coincide.
}\hfill $\blacktriangleleft$
\end{remark}

\begin{remark}[Unbounded $\eta$-unitarity operators]\label{rk:unbound_eta-unit_math}{\upshape
The construction of explicit examples of unbounded $\eta$-unitarity operators is in general  a non trivial task. For instance the examples presented in \cite[Chapter 2, Example 5.5]{azizov-iokhvidov-67} or \cite[Remark 5.8 (c)]{gheondea-88} are quite implicit. On the other hand is not difficult to show that $\eta$-unitarity operators can have arbitrarily large norm (\cf Lemma \ref{lwmma:appB-iso}). 
A relevant physical examples of unbounded $\eta$-unitarity operators
is provided by the  \emph{Gupta formulation} of quantum electrodynamics. In this theory the Lorentz transformations are described by unbounded and densely defined operators (see \eg the discussion in \cite[Section 7.4]{prugovecki-95} and references therein). 
}\hfill $\blacktriangleleft$
\end{remark}

\noindent 
Let us observe that, as for the usual Hilbert space case,  one can express every $\eta$-anti-unitary operator $U$  as the product $U=C W_U$
where $C$ is a complex conjugation which meets 
\begin{equation}\label{eq:eta-compl_cpnj}
C\eta\;=\;\eta C
\end{equation}
and $W_U:=CU$ is an  $\eta$-unitary operator\footnote{\label{note:eta-compl}
A complex conjugation $C$ which  meets condition \eqref{eq:eta-compl_cpnj}
is sometimes called an \emph{$\eta$-complex conjugation}. It is not difficult to construct  $\eta$-complex conjugations $C$. It is sufficient to select a basis of
$\{\psi_j\}_{j\in\N}\subset\s{H}$ which diagonalizes $\eta$ and to define $C$ as the anti-lienar operator fixing the basis, namely $C\psi_j=\psi_j$.
}.
\medskip

Since we allow $\eta$-unitary and $\eta$-anti-unitary operators to be unbounded, it becomes important to investigate the closure properties of these operators. The following result holds true.
\begin{proposition}[{\cite[Propositions 2 \&  3]{bracci-morchio-strocchi}}]\label{prob:lin-eta-uni}
Every $\eta$-unitary (resp.  $\eta$-anti-unitary) operator is a linear (resp. anti-linear) and closable operator. Moreover, it 
 has an inverse which
is in turn
an $\eta$-unitary (resp. $\eta$-anti-unitary) operator.
\end{proposition}

\medskip

According to the content of Remark \ref{rk:nomenc_eta-unit_math} one can characterize the \emph{group} of bounded $\eta$-unitary operators by
\begin{equation}\label{eq:rep-group_eta-unit}
{\tt U}_\eta(\s{H})\;:=\;\left\{U\in {\tt GL}(\s{H})\ |\ U^\sharp=U^{-1}\right\}
\end{equation}
where ${\tt GL}(\s{H})\subset\bb{B}_{\rm lin}(\s{H})$ is
the group of invertible bounded linear operators on $\s{H}$ and $U^\sharp$ denotes the $\eta$-adjoint of $U$. 
We will use the notation ${\tt AU}_\eta(\s{H})$ for the set of $\eta$-anti-unitary \emph{bounded} operators. This set can be characterized in terms of the group ${\tt U}_\eta(\s{H})$ once a complex conjugation $C$ which meets \eqref{eq:eta-compl_cpnj} has been chosen. In this case one has
\begin{equation}\label{eq:rep-group_eta-ant-unit}
{\tt AU}_\eta(\s{H})\;:=\;\left\{U=CW\ |\ W\in {\tt U}_\eta(\s{H})\right\}\;.
\end{equation}
The following result will be used several times.
\begin{lemma}\label{lemma:usef_ident_01}
Each operator  $U\in {\tt U}_\eta(\s{H})\sqcup {\tt AU}_\eta(\s{H})$ meets  $
\eta U^*=U^{-1}\eta$.
\end{lemma}
\proof
For $\eta$-unitary  operators the relation follows by multiplying on the left the equality $U^\sharp=U^{-1}$ with $\eta$. For the case of $\eta$-anti-unitary operators it is first necessary to observe that the relation $(AB)^*=B^*A^*$ holds true for any pair of bounded  operators, regardless whether they are linear or anti-linear, 
provided that the adjoint for anti-linear operators is defined as in the Note \ref{footnote:A}. Then, for any $U=CW\in {\tt AU}_\eta(\s{H})$ one has
$$
\eta U^*\;=\;\eta (CW)^*\;=\;\eta W^*C\;=\;W^{-1}\eta C\;=\;(W^{-1}C)\eta\;=\;U^{-1}\eta
$$
where the equality $C^*=C=C^{-1}$ and the commutativity $\eta C=C\eta$ have been used.
\qed

\medskip

\noindent
For sake of clarity let us emphasize that  the symbol  $\sqcup$ used in the claim of Lemma \ref{lemma:usef_ident_01}  stands for the \emph{disjoint union} and is justified by the evidence ${\tt U}_\eta(\s{H})\cap {\tt AU}_\eta(\s{H})=\emptyset$.


\subsection{Pseudo-unitarity in indefinite metric spaces}

Just as the concept of anti-unitary operators arise when one uses the  complex structure of the Hilbert  
spaces, a new class of operators  naturally occurs when the inner product is allowed to be indefinite.

\begin{definition}[$\eta$-pseudo-unitary and $\eta$-pseudo-anti-unitary  operators]\label{def:eta-pseudo-anti-unit-op}
Let  $(\s{H},\langle\langle\;,\;\rangle\rangle_\eta)$ be the indefinite metric space associate to a fundamental symmetry $\eta^*=\eta=\eta^{-1}$. Let $U:\s{H}\to \s{H}$ be an operator (not linear a priori) with domain $\s{D}(U)\subseteq \s{H}$ and range $\s{R}(U)\subseteq \s{H}$. Assume that:
\begin{itemize}
\item[(0)] The rank of indefiniteness meets $\kappa=\kappa_+=\kappa_-$;\vspace{1mm}
\item[(i)] Domain and range are dense sets, \ie $\overline{\s{D}(U)}= \s{H}=\overline{\s{R}(U)}$.
\end{itemize}
The operator $U$ is called \emph{$\eta$-pseudo-unitary} if:
\begin{itemize}
\item[(iii)]
The condition 
\begin{equation}\langle\langle\phi,\varphi\rangle\rangle_\eta\;=\;-\;\langle\langle U\phi, U\varphi\rangle\rangle_\eta\;,\qquad\quad \forall \phi, \varphi\in\s{D}(U)\;
\end{equation}
holds true.
\end{itemize}
It is called \emph{$\eta$-pseudo-anti-unitary} if:
\begin{itemize}
\item[(iii')]
The condition 
\begin{equation}\langle\langle\phi,\varphi\rangle\rangle_\eta\;=\;-\;\langle\langle U\varphi, U\phi\rangle\rangle_\eta\;,\qquad\quad \forall \phi, \varphi\in\s{D}(U)\;
\end{equation}
holds true.
\end{itemize}
\end{definition}
\noindent
First of all, let us remark that the occurrence of  pseudo-unitary operators 
in Definition \ref{def:eta-pseudo-anti-unit-op}
is tightly bound to the indefinite
metric. In fact in a Hilbert space the condition $\langle\phi,\phi\rangle =- \langle U\phi, U\phi,\rangle$ is untenable when $\phi\neq 0$. 
Secondly, let us observe that  
an $\eta$-pseudo-unitary or $\eta$-pseudo-anti-unitary operator may exist only in indefinite metric spaces in which the the eigenvalues $+1$ and $-1$ of the metric operator $\eta$ have the same multiplicity and this is the reason for the condition (0) of Definition \ref{def:eta-pseudo-anti-unit-op}. This is related to the fact that 
an $\eta$-pseudo-unitary operator is necessarily invertible and its inverse is still  an $\eta$-pseudo-unitary operator which satisfies $U^{-1}\subseteq -\eta U^*\eta$
which can be proved exactly as in the proof of \cite[Proposition 2]{bracci-morchio-strocchi}. The last inclusion is the key ingredient to prove the \emph{Cartan decomposition} which states that $U$ can be expressed as the product $U=RV_U$ of a (positive) $\eta$-unitary operator  $V_U$ and  a unitary and $\eta$-pseudo-unitary operator $R^*=R^{-1}=-R^\sharp$ (the proof of this fact is  identical to that of \cite[Theorem 2.1.1]{tomita-80}). The last equality implies that 
\begin{equation}\label{eq:eta-rev01}
R \eta\; =\; -\;\eta  R\,,
\end{equation}
 which can be read as the fact that $\eta$ and $-\eta$ are unitarily equivalent. 
However this is possible only if the the condition (0) of Definition \ref{def:eta-pseudo-anti-unit-op} is verified.

\begin{remark}{\upshape
One of the major interests of this work is to study the consequences of the existence of pseudo-unitary symmetries on the classification of topological phases. For this reason henceforth we will  assume 
the validity of condition (0) in Definition \ref{def:eta-pseudo-anti-unit-op}. Said differently, we will always work with finite or infinite dimensional  Krein spaces (\cf Definition \ref{def:krein}).}\hfill $\blacktriangleleft$
\end{remark}

 \begin{remark}[Comparison with the mathematical nomenclature]\label{rk:nomenc_eta-unit_math-beta-vers}{\upshape
In the mathematical literature the notion of $\eta$-pseudo-unitary operators appears in different forms and with different names. For instance in \cite{tomita-80,kissin-shulman-97}   this type  of operators are defined through
 the stronger condition $U^{-1}= -  U^\sharp$ and are called \emph{$\sharp$-immaginary}
in \cite{tomita-80} and \emph{J-skew-unitary} in \cite{kissin-shulman-97}.
In this work we use definitions and terminology introduced in
\cite{bracci-morchio-strocchi}.}\hfill $\blacktriangleleft$
\end{remark}

 A special subclass of unitary operator which meet \eqref{eq:eta-rev01}
 are those  which in addition are self-adjoint,  $R^*=R=R^{-1}$. This operators are called \emph{$\eta$-reflecting} operators\footnote{\label{note:rqref-reflex}
 It is clear from the definition that the role of an $\eta$-reflecting operator $R$ is to interchange an orthonormal basis of $\s{H}_+$ with an orthonormal basis of $\s{H}_-$ and vice versa. Moreover, if one defines an $\eta$-complex structure $C$ which preserves the bases  of  $\s{H}_+$ and $\s{H}_-$ then $C$ and $R$  automatically commute.}. 
Once an  $\eta$-reflecting operator $R$ is given one can represent
any $\eta$-pseudo-unitary operator $U$
according to the Cartan representation $U = RV_U$ where 
$V_U:=RU$ is automatically an $\eta$-unitary operator. A similar conclusion holds true for $\eta$-pseudo-anti-unitary operators. In fact, given an $\eta$-complex conjugation $C$ (\cf Note \ref{note:eta-compl}) and an $\eta$-reflecting operators $R$ such that $CR=RC$ one has that any $\eta$-pseudo-anti-unitary operator $U$ can be represented as $U=CRT_U$ where $T_U:=RCU$ is  $\eta$-unitary.
By combining this last observation with the content of Proposition \ref{prob:lin-eta-uni} and Proposition \ref{prob:lin-eta-anti-uni} one obtains the following description for $\eta$-pseudo-unitary and $\eta$-pseudo-anti-unitary operators:
\begin{proposition}[{\cite[Proposition 4]{bracci-morchio-strocchi}}]\label{prob:lin-eta-anti-uni} 
Every $\eta$-pseudo-unitary (resp. $\eta$-pseudo-anti-unitary)  operator is a linear and closable operator. Moreover, it 
 has an inverse which
is in turn
an $\eta$-pseudo-unitary (resp. $\eta$-pseudo-anti-unitary)  operator.
\end{proposition}

 \medskip

The sets of $\eta$-pseudo-unitary bounded operators and $\eta$-pseudo-anti-unitary bounded operators can be characterized in terms of the group ${\tt U}_\eta(\s{H})$ by
\begin{equation}\label{eq:rep-group_pseud-eta-unit}
{\tt PU}_\eta(\s{H})\;:=\;\left\{U=RV\ |\ V\in {\tt U}_\eta(\s{H})\right\}
\end{equation}
and 
\begin{equation}\label{eq:rep-group_pseud-eta-ant-unit}
{\tt PAU}_\eta(\s{H})\;:=\;\left\{U=CRT\ |\ T\in {\tt U}_\eta(\s{H})\right\}
\end{equation}
once a commuting pair $C$ and $R$ has been chosen.
The following result extends Lemma \ref{lemma:usef_ident_01}.
\begin{lemma}\label{lemma:usef_ident_00}
Each operator  $U\in {\tt PU}_\eta(\s{H})\sqcup {\tt PAU}_\eta(\s{H})$ meets  $
\eta U^*=-U^{-1}\eta$.
\end{lemma}
\proof
Let $U=RV$ be an $\eta$-pseudo-unitary operator. It follows that
$$
\eta U^*\;=\;\eta (R V)^*\;=\;(\eta V^*)R\;=\;(V^{-1}\eta)R\;=\;-(V^{-1}R) \eta\;=\;-U^{-1}\eta\;
$$
where Lemma \ref{lemma:usef_ident_01} has been used. The last relation also shows that $U\in {\tt PU}_\eta(\s{H})$ if and only if $U$ is a linear operator which fulfills 
$U^\sharp=-U^{-1}$. In the case $U=CRT=CW$ is $\eta$-pseudo-anti-unitary one has that $W=RT$ is $\eta$-pseudo-unitary, therefore
$$
\eta U^*\;=\;\eta (CW)^*\;=\;\eta W^*C\;=\;-W^{-1}\eta C\;=\;-(W^{-1}C)\eta\;=\;-U^{-1}\eta
$$
just as in the final part of the proof of Lemma \ref{lemma:usef_ident_01}.
\qed


\section{Quantum symmetries: The Wigner's theorem}
\label{sec:Q-sym}

In this Section the classical Wigner's theorem \cite{wigner-59} is revisited in a geometric language  according to the presentation provided in \cite[Section 1.1]{freed-moore-13}. The same description is then adapted for the extended version of the Wigner's theorem in Krein spaces proved in \cite{bracci-morchio-strocchi}.


\subsection{The Hilbert space case}
\label{sect:hilb_wigner}
According to  the ordinary formulation of  Quantum Mechanics 
the \emph{state}   of a quantum system is described by an element of the  projective space $\n{P}\s{H}$ of a complex separable Hilbert space $(\s{H},\langle\;,\;\rangle)$. In other words, a state   is a line, or \emph{ray}, of vectors. The transition probability between two states $[\phi],[\varphi]\in \n{P}\s{H}$ is given by 
\begin{equation}\label{eq:prob_hilb}
{\rm Prob}\big([\phi],[\varphi]\big)\;:=\;\frac{\big|\langle \phi,\varphi\rangle\big|^2}{\langle \phi,\phi\rangle\; \langle \varphi,\varphi\rangle}\;,
\end{equation}
where $\phi\in[\phi]$ and $\varphi\in[\varphi]$ are suitable representatives. Equation \eqref{eq:prob_hilb}  defines a symmetric function
$$
{\rm Prob}\;:\;\n{P}\s{H}\;\times\; \n{P}\s{H}\;\longrightarrow\; [0,1]\;.
$$
Since the transition probabilities  are interpreted as the possible outcomes of physical measurements,  one is led to call  \emph{quantum symmetry}  any invertible transformation (an automorphism) of the 
state space $S:\n{P}\s{H}\to \n{P}\s{H}$ preserving the symmetric function ${\rm Prob}$. 
The set of these transformations will be denoted by
${\tt QS}
(\n{P}\s{H})$.
A priori a quantum symmetry $S$ is not required to have other additional structural properties. However a fundamental theorem by Wigner \cite{wigner-59}
 states that every  {quantum symmetry} is implemented by means of a linear map $U_S:\s{H}\to \s{H}$
 which can be either  \emph{unitary} or \emph{anti-unitary}.
 
 \medskip
 
 Following \cite[Section 1.1]{freed-moore-13} (see also \cite{thiang-16,kubota-17})
 one can  recast the Wigner's theorem in the following terms. Let 
 $$
 {\tt QS}(\s{H})\;:=\;{\tt U}(\s{H})\;\sqcup\; {\tt AU}(\s{H})
 $$
 be the set of all unitary and
anti-unitary transformations of $\s{H}$.
 It is a group since the composition of two anti-unitary transformations is unitary and we refer to 
${\tt QS}(\s{H})$ as the group of  
 \emph{linear} quantum symmetries. This group fits in a group extension
\begin{equation}\label{eq:ext01}
1\;\longrightarrow\;{\tt U}(\s{H})\;\longrightarrow\;{\tt QS}(\s{H})\;\stackrel{\varpi}{\longrightarrow}\;\Z_2\;\longrightarrow\;1
\end{equation}
where $\Z_2:=\{\pm 1\}$ is the cyclic group of order 2 and the  kernel of the homomorphism $\varpi$ is the group of unitary operators. Said differently, 
one has that $\varpi(U)=+1$ if  $U$ is a unitary operator and $\varpi(U)=-1$ if  $U$ is anti-unitary. The
Wigner  theorem asserts that there is a group extension\footnote{For basic facts about the notion of group extension we refer to \cite[Chapter 4]{brown-82} or \cite{tuynman-wiegerinck-87} or \cite[Appendix A]{freed-moore-13}.}
\begin{equation}\label{eq:ext02}
1\;\longrightarrow\;\n{S}^1\;\stackrel{\imath}{\longrightarrow}\;{\tt QS}(\s{H})\;\stackrel{\pi}{\longrightarrow}\;{\tt QS}(\n{P}\s{H})\;\longrightarrow\;1
\end{equation}
where $\n{S}^1:=\{z\in\C\;|\;|z|=1\}$ denotes the group of complex numbers of modulus one and the map $\imath$ associates to  each $\lambda\in \n{S}^1$ the \emph{scaling} operator $\imath(\lambda):=\lambda\n{1}$. The projection map $\pi$ sends the  {linear} quantum symmetry $U$  into the quantum symmetry $\pi(U)$ which acts on the rays according to  $\pi(U)[\psi]:=[U\psi]$. 
Clearly $\imath(\n{S}^1)\subseteq{\tt QS}(\s{H})$  agrees with the kernel of $\pi$ and every quantum symmetry  $S\in{\tt QS}(\n{P}\s{H})$ lifts to a linear quantum symmetry $U_S:=\pi^{-1}(S)\in {\tt QS}(\s{H})$ 
which is unique up to the composition with 
a scaling operator $\lambda\n{1}\in \imath(\n{S}^1)$. Even though $\n{S}^1$ is commutative,
 the extension \eqref{eq:ext02}
is \emph{not} central since  $\imath(\n{S}^1)$ is not in the center of ${\tt QS}(\s{H})$. In fact, one has that
\begin{equation}\label{eq:ext03}
U\;(\lambda\n{1})\;U^{-1}\;=\;\lambda^{\varpi(U)}\;\n{1}\;=\;
\left\{
\begin{aligned}
&\lambda\;\n{1}&\qquad&\text{if}\quad \varpi(U)=+1\\
&{\lambda}^{-1}\;\n{1}&\qquad&\text{if}\quad \varpi(U)=-1\;.\\
\end{aligned}
\right.
\end{equation}
Hence \eqref{eq:ext02} is not a central extension but rather
 a \emph{twisted} (central) extension 
according to \cite[Definition 1.7]{freed-moore-13} where the twisting is expressed by the map $\varpi:{\tt QS}(\s{H})\to\Z_2$ according to  \eqref{eq:ext03}.

\begin{remark}[Choice of the topology]\label{rk:top_strong}{\upshape
The group ${\tt QS}(\s{H})$ and its subgroup ${\tt U}(\s{H})$ can be endowed with different topologies. In the uniform  topology given by the operator norm \eqref{eq:op_norm}
${\tt U}(\s{H})$ is a closed and connected space with the structure of a topological \emph{Banach-Lie} group. Moreover,  if $\s{H}$ is infinite dimensional and separable  then ${\tt U}(\s{H})$ is contractible to the identity \cite{kuiper-65}. The bijection $C: 
{\tt U}(\s{H})\to {\tt AU}(\s{H})$, induced by any complex conjugation, implies that also ${\tt AU}(\s{H})$ is  connected (to $C$) and contractible in the infinite dimensional case. In summary ${\tt QS}(\s{H})$  becomes
 a topological group with two disconnected  components when endowed with the uniform  topology. 
 Unfortunately the uniform topology is too strong for many purposes. For instance the one-parameter \emph{time evolution} group  $\R\ni t\mapsto \expo{\ii t H}\in {\tt U}(\s{H})$
fails to be uniformly continuous when the self-adjoint operator $H$ is unbounded. To solve this problem one needs to weaken the topology on ${\tt QS}(\s{H})$ without compromising the structure of topologic group. 
One possible candidate which makes the time evolution group continuous is the \emph{strong topology} generated by the semi-norms $s_\psi(A):=\|A\psi\|$ as $\psi$ varies in $\s{H}$.
Even though the operator product is not continuous in general with respect the strong topology, it is continuous when restricted to set of equi-bounded operators. This is the basic observation that allows to show that ${\tt U}(\s{H})$ is a topological \emph{Polish group}\footnote{A Polish group
is a completely metrizable topological group.
To this respect, it is interesting to observe that when $\s{H}$ is infinite dimensional then  ${\tt U}(\s{H})$ is \emph{not} closed with respect to the strong topology (as erroneously claimed below  \cite[Proposition 3]{schottenloher-13}) even though  ${\tt U}(\s{H})$ is closed with respect to the metric which generates the topology equivalent to the strong topology. This  fact, although  looking surprising at first sight, is not absurd. For instance, as proved in \cite[Chapter II, Proposition 4.9]{takesaki-02},  ${\tt U}(\s{H})$ is in fact
closed with respect to the $\ast$-strong topology 
even though the strong and the $\ast$-strong topology coincide. As commented in \cite[Chapter II, Remark 4.10]{takesaki-02} this phenomenon
 is a manifestation of the fact that the uniforme structures induced by the strong and the $\ast$-strong topologies are distinct.} when endowed with the strong topology 
\cite[Proposition II.1]{neeb-97} or
\cite[Theorem 1.2]{espinoza-uribe-14} (see also
 \cite[Proposition 1]{schottenloher-13}).
In one of the foundational papers for twisted equivariant K-theory \cite{atiyah-segal-04}, Atiyah and Segal proposed a different way to topologize  
${\tt U}(\s{H})$. In order to deal  with equivariant Hilbert bundles and its relation with its associated unitary principal equivariant bundles, they claimed that one is obliged to consider the
\emph{compact-open topology}\footnote{For more details about the compact-open topology on ${\tt U}(\s{H})$ we refer to \cite[Appendix 1]{atiyah-segal-04} and \cite[Appendix D]{freed-moore-13}.} on the structural group ${\tt U}(\s{H})$.
A fact that seems to be remarkable is that the compact-open topology and the strong topology agree on ${\tt U}(\s{H})$ (see \cite[Lemma 1.8.]{espinoza-uribe-14}
or \cite[Proposition 2]{schottenloher-13}). This immediately implies that ${\tt U}(\s{H})$ is a topological group also with respect to the compact-open topology. This last  fact has been erroneously denied
in  \cite[pg. 321]{atiyah-segal-04}. The equality between compact-open and strong topology on ${\tt U}(\s{H})$ along with \cite[Proposition A2.1.]{atiyah-segal-04} implies that ${\tt U}(\s{H})$ is a contractible space also with respect to the strong topology if $\s{H}$ is infinite dimensional and separable. Clearly, the strong topology can be 
defined also on  ${\tt AU}(\s{H})$ and the bijection implemented by $C$ becomes a homeomorphism. In summary ${\tt QS}(\s{H})$ endowed with the strong topology turns out to be
 a topological Polish group with two disconnected  components which are contractible when $\s{H}$ is infinite dimensional and separable.
The topology on  the projective space ${\tt U}(\n{P}\s{H})$ is induced as the quotient topology
 associated with the short exact sequence
\begin{equation}\label{eq:ext02-solo_U}
1\;\longrightarrow\;\n{S}^1\;\stackrel{\imath}{\longrightarrow}\;{\tt U}(\s{H})\;\stackrel{\pi}{\longrightarrow}\;{\tt U}(\n{P}\s{H})\;\longrightarrow\;1\;.
\end{equation}
Interestingly, the strong topology on ${\tt U}(\s{H})$ descends to the strong  topology on ${\tt U}(\n{P}\s{H})$
and  ${\tt U}(\s{H})$ can be seen as a principal bundle over ${\tt U}(\n{P}\s{H})$ with typical fiber $\n{S}^1$ \cite[Theorem 1]{simms-70}. The homotopy exact sequence induced by \eqref{eq:ext02-solo_U} provides that $\pi_k({\tt U}(\n{P}\s{H}))=\pi_{k-1}(\n{S}^1)$ showing that ${\tt U}(\n{P}\s{H})$ is a model for the Eilenberg-MacLane space $K(\Z, 2)$. In much the same way the space
${\tt QS}(\n{P}\s{H})$ can be topologized with the (strong) quotient topology
inherited by ${\tt QS}(\s{H})$. In this way also ${\tt QS}(\s{H})$ can be seen as a principal bundle over ${\tt QS}(\n{P}\s{H})$ with typical fiber $\n{S}^1$.
As a final remark let us observe that the map  $\varpi:{\tt QS}(\s{H})\to\Z_2$ which enters in the group extension
\eqref{eq:ext01} is a continuous homomorphism with respect to the natural discrete topology on $\Z_2$ and the strong topology on ${\tt QS}(\s{H})$.
}\hfill $\blacktriangleleft$
\end{remark}

\begin{remark}[Twisted extensions and cohomology]{\upshape
It is a well known result that the central extensions of a group $\n{G}$ by an abelian group $\n{A}$ are classified by classes in the group-cohomology $H^2(\n{G}, \n{A})$
(see \cite[Chapter 4, Theorem 3.12]{brown-82} or \cite[Proposition 3.4]{tuynman-wiegerinck-87}). 
Also twisted extensions of the type of 
\eqref{eq:ext02} - \eqref{eq:ext03} can be classified by means of a cohomology theory.
Let $\varpi:\n{G}\to\Z_2$ be the homomorphism that defines the twisted extension and consider the (left) group action of $\n{G}$ over $\n{A}$ given by $g\cdot a:=a^{\varpi(g)}$. The group $\n{A}$ endowed with such a $\n{G}$-action is denoted by $\n{A}_\varpi$ and provides a system of twisted coefficients for the group cohomology of $\n{G}$.
Then, the $\varpi$-twisted central extensions of 
$\n{G}$ by $\n{A}$ are classified by the elements of $H^2(\n{G}, \n{A}_\varpi)$. Therefore, the  twisted central extension \eqref{eq:ext02}
corresponds to an element of the cohomology group
$H^2({\tt QS}(\n{P}\s{H}), \n{S}^1_\varpi)$ where the homomorphism $\varpi$ distinguishes between unitary and anti-unitary symmetries.
}\hfill $\blacktriangleleft$
\end{remark}


\subsection{The  Krein space case}
\label{sec:krein_space-Qsim}
The discussion of the Wigner's theorem in Krein spaces presents some technical difficulties with respect to the Hilbert space case discussed in 
Section \ref{sect:hilb_wigner}.
The transition probability between two rays $[\phi],[\varphi]\in \n{P}\s{H}$ 
of a Krein space $(\s{H},\langle\langle\;,\;\rangle\rangle_\eta)$ is naturally defined by mimicking \eqref{eq:prob_hilb}:
\begin{equation}\label{eq:prob_krein}
{\rm Prob}_\eta\big([\phi],[\varphi]\big)\;:=\;\frac{\big|\langle\langle \phi,\varphi\rangle\rangle_\eta\big|^2}{\langle \phi,\phi\rangle\; \langle \varphi,\varphi\rangle}
\;=\;\frac{\big|\langle \phi,\eta\varphi\rangle\big|^2}{\langle \phi,\phi\rangle\; \langle \varphi,\varphi\rangle}\;,
\end{equation}
Let us point out that in both  \eqref{eq:prob_hilb}
and \eqref{eq:prob_krein}
the normalization of the vectors is defined with respect to the Hilbert structure. In fact one
cannot normalize the rays with respect to the indefinite inner product due to the presence of non-zero vectors of zero $\eta$-length. 
Also equation \eqref{eq:prob_krein}  defines a symmetric function
$$
{\rm Prob}_\eta\;:\;\n{P}\s{H}\;\times\; \n{P}\s{H}\;\longrightarrow\; [0,1]\;
$$
and, in analogy with the Hilbert space case, we are tempted to  call  \emph{$\eta$-quantum symmetry}  any invertible transformation of the 
state space $S:\n{P}\s{H}\to \n{P}\s{H}$ preserving the symmetric function ${\rm Prob}_\eta$. However the situation in the framework of Krein spaces  is slightly more complicated than the Hilbert space case due to the fact that  maps $S$ which preserve ${\rm Prob}_\eta$ in general are not defined for all the rays in  $\n{P}\s{H}$. However, it is
reasonable to require that a \virg{good} quantum symmetry $S$, in addition to being invertible, has to be defined  on a set of rays associated to a dense linear manifold $\s{D}\subseteq\s{H}$, and has to map the rays associated to $\s{D}$ onto a set of rays associated to a dense linear manifold $\s{D}'\subseteq\s{H}$. 
The requirement  of working with dense manifolds is suggested by physical interesting examples (\cf Remark \ref{rk:unbound_eta-unit_math}). If one accepts the density condition above as part of the definition of an {$\eta$-quantum symmetry} $S$ then the extended version of the  
Wigner's theorem due to  Bracci,  Morchio and  Strocchi \cite{bracci-morchio-strocchi} states that 
there exists a (possible unbounded) operator $U_S:\s{H}\to\s{H}$ which  implements $S$ and  $U_S$ is either $\eta$-unitary or $\eta$-anti-unitary or $\eta$-pseudo-unitary or $\eta$-pseudo-anti-unitary
(the last two cases can occur only when condition (0) in Definition \ref{def:eta-pseudo-anti-unit-op} holds).

\medskip

However for the aims of this work we need  to consider only the \emph{bounded} version of the  extended  
Wigner's theorem. Let us introduce the group
$$
 {\tt QS}_\eta(\s{H})\;:=\;{\tt U}_\eta(\s{H})\;\sqcup\; {\tt AU}_\eta(\s{H})\sqcup\; {\tt PU}_\eta(\s{H})\sqcup\; {\tt PAU}_\eta(\s{H})
 $$
 of the \emph{(bounded)  linear  $\eta$-quantum symmetries}. This group fits in a group extension
\begin{equation}\label{eq:ext01-eta}
1\;\longrightarrow\;{\tt U}_\eta(\s{H})\;\longrightarrow\;{\tt QS}_\eta(\s{H})\;\stackrel{(\wp,\varpi)}{\longrightarrow}\;\Z_2\times\Z_2\;\longrightarrow\;1\; 
\end{equation}
and the  kernel of the homomorphism $(\wp,\varpi)$ is the group of the $\eta$-unitary operators. More specifically, the maps  $\wp$ and $\varpi$ determine 
the nature of $U\in {\tt QS}_\eta(\s{H})$ according to the following relation:

\begin{equation}\label{eq:eta-general}
\langle\langle\phi,\varphi\rangle\rangle_\eta\;=\;\wp(U)\;\big|\langle\langle U\phi, U\varphi,\rangle\rangle_\eta\big|^{1-\varpi(U)}\;\langle\langle U\phi, U\varphi,\rangle\rangle_\eta^{\varpi(U)}\;,\qquad\quad \forall \phi, \varphi\in\s{H}\;.
\end{equation}
The convention on the signs, which can be extended straightforwardly  to the case of unbounded symmetries, is conveniently 
 summarized in Table 1.3 of Section \ref{sec:introduction}.

\medskip

The generalized Wigner's  theorem proved in \cite{bracci-morchio-strocchi} asserts that 
the set of \emph{(bounded) $\eta$-quantum symmetries} ${\tt QS}_\eta(\n{P}\s{H})$  fits in the  group extension
\begin{equation}\label{eq:ext02-eta}
1\;\longrightarrow\;\n{S}^1\;\stackrel{\imath}{\longrightarrow}\;{\tt QS}_\eta(\s{H})\;\stackrel{\pi}{\longrightarrow}\;{\tt QS}_\eta(\n{P}\s{H})\;\longrightarrow\;1
\end{equation}
where the maps $\imath$ and $\pi$ have the same meaning as in \eqref{eq:ext02}. Again this extension
is \emph{not} central but \emph{twisted} by the homomorphism 
$(\wp,\varpi):{\tt QS}_\eta(\s{H})\to\Z_2\times\Z_2$ according to 
\begin{equation}\label{eq:ext03-eta}
U\;(\lambda\n{1})\;U^{-1}\;=\;\wp(U)\;\lambda^{\varpi(U)}\;\n{1}\;
\end{equation}

\begin{remark}[Choice of the topology]\label{rk:top_strong-eta}{\upshape
By making a parallel with Remark \ref{rk:top_strong-eta} it is natural to ask in which way one can topologize 
${\tt QS}_\eta(\s{H})$ and ${\tt QS}_\eta(\n{P}\s{H})$ in order to
 make the exact sequence \eqref{eq:ext02-eta} topological.
First of all one can notice that the group 
${\tt QS}_\eta(\s{H})$ is made  by four disconnected components (we are tacitly assuming the validity of condition (0) in Definition \ref{def:eta-pseudo-anti-unit-op}) and the three components
${\tt AU}_\eta(\s{H})$,  ${\tt PU}_\eta(\s{H})$, and  ${\tt PAU}_\eta(\s{H})$
are all related to the sub-group $
{\tt U}_\eta(\s{H})
$ by the election of a $\eta$-complex conjugation $C$ as in \eqref{eq:eta-compl_cpnj}
and a $\eta$-reflecting operators $R$ as in \eqref{eq:eta-rev01} such that
 $CR=RC$. This implies that it is enough to topologize the subgroup $
{\tt U}_\eta(\s{H})$ in order to endow the full group 
${\tt AU}_\eta(\s{H})$ with a topology. Once a topology on ${\tt AU}_\eta(\s{H})$ is defined one can induce the quotient topology on ${\tt QS}_\eta(\n{P}\s{H})$ by means of the projection $\pi$. The characterization \eqref{eq:rep-group_eta-unit} implies that  
${\tt U}_\eta(\s{H})$ is an  algebraic subgroup of ${\tt GL}(\s{H})$ 
which is a topological group when endowed with the 
uniform  topology given by the operator norm. As a consequence  ${\tt U}_\eta(\s{H})$ turns out to be 
a  (closed)   Banach-Lie group when
endowed with the topology induced by the operator norm  \cite[Section 23]{upmeier-85}.
Unfortunately, as mentioned in Remark \ref{rk:top_strong-eta}, the uniform topology is generally too strong for many purposes.
For this reason one could be tempted to endow 
${\tt U}_\eta(\s{H})$ with the strong topology. However, the  set ${\tt U}_\eta(\s{H})$ contains elements or arbitrarily large norm (\cf Lemma \ref{lwmma:appB-iso}) and the product of operators is generally not continuous with respect to the strong topology. Thus ${\tt U}_\eta(\s{H})$ fails to be a topological group with respect to the strong topology.
A possible strategy to circumvent these problems is to use the
{compact-open topology} trying to adapt the results from 
\cite[Appendix 1]{atiyah-segal-04} and \cite[Appendix D]{freed-moore-13}.
}\hfill $\blacktriangleleft$
\end{remark}


\subsection{Symmetries of $\eta$-self-adjoint operators}
This Section is devoted to the study of the transformation property of $\eta$-self-adjoint operators under linear $\eta$-quantum symmetries. It is worth remembering that a  densely defined (possible unbounded) linear operator $H$ is $\eta$-self-adjoint if and only if $H=H^\sharp =\eta H^*\eta$. This condition implies  the relation $\s{D}(H)=\eta[\s{D}(H^*)]$ among the domains of $H$ and $H^*$ and assures  that $H$ is a closed operator ($H^*$ is closed by definition of adjoint and the same holds for $H^\sharp$ due to the boundedness and the invertibility of $\eta$).

\medskip

A (possible unbounded) densely defined  linear symmetry $U$ is said to be of type $(\wp,\varpi)$ if it meets equation \eqref{eq:eta-general} for all $\phi, \varphi\in\s{D}(U)$. It follows from Proposition \ref{prob:lin-eta-uni} and Proposition  \ref{prob:lin-eta-anti-uni} that any symmetry $U$ of type $(\wp,\varpi)$ is automatically closable and invertible. Therefore, without loss of generality, one can always consider  $U$ and $U^{-1}$ as closed symmetries of the same type.

\medskip

Although $H$ and $U$ are closed and densely defined operators the product $UHU^{-1}$  may be not well defined in general. As a matter of fact  the product has to be initially defined on the initial domain
$$
\s{D}_0\;:=\;\left\{\varphi\in \s{D}(U^{-1})\ |\ U^{-1}\varphi\in\s{D}(H)\;\;\text{and}\;\; H(U^{-1}\varphi)\in\s{D}(U)\right\}\;.
$$
This domain might be non-dense or even empty. 
Therefore, in order to have a good definition of $UHU^{-1}$ one must assume the density of $\s{D}_0$.

\begin{theorem}\label{theo_inv_eta_self}
Let $H$ be an $\eta$-self-adjoint operator with dense domain $\s{D}(H)$. Let $U$ be  
a linear symmetry $U$ of type $(\wp,\varpi)$ defined on the dense domain $\s{D}(U)$. Assume that the product $UHU^{-1}$ is initially  defined on the dense  domain $\s{D}_0$. Then the operator $UHU^{-1}$ is closable on $\s{D}_0$ and \emph{$\eta$-symmetric}, meaning that the inclusion 
$$
\eta  UHU^{-1}\subseteq (UHU^{-1})^*  \eta\;
$$
holds true.
\end{theorem}
\proof  Since $UHU^{-1}$ is densely defined it is adjointable, namely $(UHU^{-1})^*$
exists and is closed. Moreover one has  $(\eta UHU^{-1})^*=(UHU^{-1})^*\eta$ since $\eta$ is bounded and self-adjoint.
Let  $\phi, \varphi\in \s{D}_0$ and observe that $UHU^{-1}\varphi\in {\s{R}(U)}\subseteq \s{D}(U^{-1})$. A straightforward computation shows that
\begin{equation}\label{eq:compl_01}
\begin{aligned}
\langle\phi,\eta UHU^{-1}\varphi\rangle\;&=\;\langle\langle \phi, UHU^{-1}\varphi\rangle\rangle_\eta\\
&=\;
\wp(U^{-1})\;\big|\langle\langle U^{-1}\phi, HU^{-1}\varphi\rangle\rangle_\eta\big|^{1-\varpi(U^{-1})}\;\langle\langle U^{-1}\phi, HU^{-1}\varphi\rangle\rangle_\eta^{\varpi(U^{-1})}\\
&=\;
\wp(U^{-1})\;\big|\langle\langle HU^{-1}\phi, U^{-1}\varphi\rangle\rangle_\eta\big|^{1-\varpi(U^{-1})}\;\langle\langle HU^{-1}\phi, U^{-1}\varphi\rangle\rangle_\eta^{\varpi(U^{-1})}\\
\end{aligned}
\end{equation}
where  the second equality follows from \eqref{eq:eta-general} and  the third by
the $\eta$-self-adjointness of $H$. By exploiting again \eqref{eq:eta-general} and using  that $\varpi(U^{-1})\varpi(U)=1$  in view of the fact that $U$ and $U^{-1}$ are of the same type one obtains
\begin{equation}\label{eq:compl_02}
\langle\langle HU^{-1}\phi, U^{-1}\varphi\rangle\rangle_\eta^{\varpi(U^{-1})}\;=\;
\wp(U)\;\big|\langle\langle UHU^{-1}\phi, \varphi\rangle\rangle_\eta\big|^{^{\varpi(U^{-1})}-1}\;\langle\langle UHU^{-1}\phi, \varphi\rangle\rangle_\eta\;.
\end{equation}
By inserting \eqref{eq:compl_02} in \eqref{eq:compl_01} one gets
\begin{equation}\label{eq:compl_03}
\begin{aligned}
\langle\langle \phi, UHU^{-1}\varphi\rangle\rangle_\eta\;&=\;
\wp(U^{-1})\;\wp(U)\;\left(\frac{\big|\langle\langle HU^{-1}\phi, U^{-1}\varphi\rangle\rangle_\eta\big|}{\big|\langle\langle UHU^{-1}\phi, \varphi\rangle\rangle_\eta\big|}\right)^{^{1-\varpi(U^{-1})}}\;\langle\langle UHU^{-1}\phi, \varphi\rangle\rangle_\eta\;.
\end{aligned}
\end{equation}
Since $|\langle\langle HU^{-1}\phi, U^{-1}\varphi\rangle\rangle_\eta|=|\langle\langle UHU^{-1}\phi, \varphi\rangle\rangle_\eta|$ and $\wp(U^{-1})\wp(U)=1$ one concludes that
\begin{equation}\label{eq:compl_04}
\begin{aligned}
\langle\langle \phi, UHU^{-1}\varphi\rangle\rangle_\eta\;&=\;
\langle\langle UHU^{-1}\phi, \varphi\rangle\rangle_\eta\;,\qquad\quad \forall\, \phi, \varphi\in \s{D}_0
\end{aligned}
\end{equation}
or equivalently
\begin{equation}\label{eq:compl_05}
\begin{aligned}
\langle\phi,\eta UHU^{-1}\varphi\rangle\;&=\;\langle\eta UHU^{-1}\phi,\varphi\rangle\;,\qquad\quad \forall\, \phi, \varphi\in \s{D}_0
\end{aligned}
\end{equation}
where we used again that $\eta$ is self-adjoint. The last equation implies that 
$$
\eta UHU^{-1}\;\subseteq\; (\eta UHU^{-1})^*\;=\; (UHU^{-1})^*\;\eta
$$
proving that $UHU^{-1}$ is $\eta$-symmetric. Moreover, the first inequality
implies that  $\eta UHU^{-1}$ is symmetric and therefore closable. Since $\eta$ is bounded and invertible one immediately concludes that also $UHU^{-1}$ is closable.\qed

\medskip

The last result deserves a comment. In general the product of two closed operators, even when it is densely defined, can fail to be closable (see \eg \cite[Exercise 2.8.43]{kadison-ringrose-83}). However, this is not the case for $UHU^{-1}$ under the assumption of Theorem \ref{theo_inv_eta_self}. Unfortunately, the conditions of 
Theorem \ref{theo_inv_eta_self} does not seem to be sufficient to prove that the closure $\overline{UHU^{-1}}$ is $\eta$-self-adjoint. For that one needs conditions which assure 
$$
\eta\;\overline{UHU^{-1}}\;=\;\overline{\eta\; UHU^{-1}}\;=\;(UHU^{-1})^*\;\eta
$$
where the first equality is for free due to the boundedness and invertibility of $\eta$, while the second is the difficult part to be proven. A sufficient condition is to prove that
$$
\eta\left[\s{D}\left((UHU^{-1})^*\right)\right]\;\subseteq\;\s{D}\left(\overline{UHU^{-1}}\right)\;
$$
but in general this is difficult to check. A big simplification comes from assuming the boundedness of the linear symmetry $U$.
\begin{corollary}\label{cor1_inv_eta_self}
Let $H$ be an $\eta$-self-adjoint operator with dense domain $\s{D}(A)$. Let $U\in {\tt QS}_\eta(\s{H})$ be  a bounded
linear symmetry $U$ of type $(\wp,\varpi)$. Then the operator $UHU^{-1}$ is $\eta$-self-adjoint.
\end{corollary}
\proof
First of all it is worth noticing that $UHU^{-1}$ is automatically closed and defined on the dense domain $U[\s{D}(H)]$ due to the boundedness and invertibility of $U$.
Moreover, one has
\begin{equation}\label{eq:eta-se-f-cor1}
(UHU^{-1})^*\;=\;(U^{-1})^* H^* U^*\;=\;(U^{-1})^* \eta H \eta U^*
\end{equation}
where  the last equality follows from the  $\eta$-self-adjointness of $H$. By summarizing the results of Lemma \ref{lemma:usef_ident_01} and Lemma \ref{lemma:usef_ident_00} one obtains that each member of $U\in {\tt QS}_\eta(\s{H})$ meets
\begin{equation}\label{eq:eta-se-f-cor2}
\eta U^*\;=\;\wp(U)U^{-1}\eta\;.
\end{equation}
It is interesting to notice that the sign $\varpi(U)$ does not enter in the \eqref{eq:eta-se-f-cor2}. Plugging in the \eqref{eq:eta-se-f-cor2} and the related relation
$$
 (U^{-1})^*\eta\;=\;\wp(U^{-1}) \eta U
$$
into \eqref{eq:eta-se-f-cor1} and exploiting that $\wp(U)\wp(U^{-1})=1$ one finally gets
$$
(UHU^{-1})^*\;=\;(\eta U) H(U^{-1}\eta)\;=\; \eta (UHU^{-1}) \eta\;.
$$
The last equation is equivalent to 
 $UHU^{-1}=(UHU^{-1})^\sharp$ and this proves the claim.
\qed

\medskip

Corollary \ref{cor1_inv_eta_self} holds true, in particular, for bounded $\eta$-self-adjoint operators. The space ${\tt H}_\eta(\s{H})$ of  bounded $\eta$-self-adjoint operators has the structure of a vector space over the field $\R$ of the real number.
Moreover, as consequence of the continuity of the adjoint, the space  ${\tt H}_\eta(\s{H})$ is closed with respect to the operator norm. Said differently,  ${\tt H}_\eta(\s{H})$ is an $\R$-Banach space.
\begin{corollary}\label{cor2_inv_eta_self}
The group ${\tt QS}_\eta(\s{H})$ acts by automorphisms on the $\R$-Banach space ${\tt H}_\eta(\s{H})$
through the \emph{adjoint action} 
$$
{\tt QS}_\eta(\s{H})\;\ni\; U\;\mapsto\;{\rm Ad}_U\;\in\;{\rm Aut}({\tt H}_\eta(\s{H}))$$
given by 
$$
{\rm Ad}_U(H)\;:=\;UHU^{-1}$$ for any $H\in {\tt H}_\eta(\s{H})$.
\end{corollary}
%


\section{$\s{C}$-symmetry for $\eta$-self-adjoint operators}
\label{sec:C-symm}

The notion of $\s{C}$-symmetry for $\eta$-self-adjoint operators dates back to the pioneering works \cite{bender-brody-jones-02,bender-brody-jones-03} in the context of the $\s{PT}$-QM. For a recent review on the subject, see  \cite{bender-07}. From a more mathematical treatment of the notion of $\s{C}$-symmetry  we refer to \cite{kuzhel-09,kuzhel-sudilovskaya-17}
 and references therein.

\subsection{Definition and main properties}

A $\s{C}$-symmetry in a  space with  indefinite metric $\eta$ is an $\eta$-self-adjoint involution $\Xi$ such that the product $\eta \Xi$ is strictly positive. More precisely, following \eg \cite[Definition 3.1]{kuzhel-09} one has:
\begin{definition}[$\s{C}$-symmetry] \label{dfn:C_symmetry}
Let  $(\s{H},\langle\langle\;,\;\rangle\rangle_\eta)$ be the indefinite metric space associated to a fundamental symmetry $\eta^*=\eta=\eta^{-1}$. A linear bounded operator $\Xi:\s{H}\to \s{H}$ which fulfills
\begin{align}\label{eq:C-sym1}
\Xi^2 \;=\; \n{1}\;, \qquad\quad
\eta \Xi \;>\; 0
\end{align}
is called a $\s{C}$-symmetry. The symbol
$$
{\tt C}_\eta(\s{H})\;: =\; \left\{ \Xi \in\bb{B}_{\rm lin}(\s{H})\; \big|\; \Xi\; \text{\upshape meets}\; \eqref{eq:C-sym1} \right\}
$$
denotes the space of the $\s{C}$-symmetries for the space $(\s{H},\langle\langle\;,\;\rangle\rangle_\eta)$.
\end{definition}

\medskip

\noindent 
Note that the positivity condition automatically implies that $\eta \Xi=(\eta \Xi)^*=\Xi^*\eta$
namely 
\begin{equation}\label{eq:Xi_eta-self}
\Xi\;=\;\eta\Xi^*\eta\;=\;\Xi^\sharp
\end{equation}
which means that $\Xi$ is $\eta$-self-adjoint.
Clearly $\Xi_0\:= \eta$ is a (\emph{trivial}) 
$\s{C}$-symmetry showing  that ${\tt C}_\eta(\s{H})$ is not empty.
Let us introduce the space
$$
{{\tt R}}_\eta(\s{H})\;: =\; \left\{ Q \in\bb{B}_{\rm lin}(\s{H})\; \big|\; Q=Q^*\;,\ \ Q\eta+\eta Q=0 \right\}
$$
of bounded self-adjoint operators  anti-commuting with the metric $\eta$. Observe that
$Q\in {{\tt R}}_\eta(\s{H})$ implies 
$$
Q\;=\;-\eta Q^\ast\eta\;=\;-Q^\sharp\;,
$$
namely $Q$ is anti-$\eta$-self-adjoint.
Both the spaces ${{\tt C}}_\eta(\s{H})$ and ${{\tt R}}_\eta(\s{H})$ can be topologized with the uniform topology inherited from the operator norm on $\bb{B}_{\rm lin}(\s{H})$. 
The next result describes the relation between these two spaces.

\begin{lemma}[\cite{kuzhel-sudilovskaya-17, kuzhel-09}] \label{lem:space_of_C_symmetries}
Let  $(\s{H},\langle\langle\;,\;\rangle\rangle_\eta)$ be the indefinite metric space associated to a fundamental symmetry $\eta^*=\eta=\eta^{-1}$. The spaces ${{\tt C}}_\eta(\s{H})$ and ${{\tt R}}_\eta(\s{H})$ are closed with respect to the operator norm topology. Moreover, the maps
\begin{equation}\label{eq:naps_C-Q}
\begin{aligned}
{{\tt C}}_\eta(\s{H})\;\ni\; \Xi\;&\stackrel{Q}{\longmapsto}\;Q(\Xi)\;:=\;\log(\eta\Xi)\;\in\;{{\tt R}}_\eta(\s{H})\\
{{\tt R}}_\eta(\s{H})\;\ni\; Q\;&\stackrel{\Xi}{\longmapsto}\;\Xi(Q)\;:=\;\eta\;\expo{Q}\;\in\;{{\tt C}}_\eta(\s{H})\\
\end{aligned}
\end{equation}
provide a   homeomorphism  ${{\tt C}}_\eta(\s{H})\simeq{{\tt R}}_\eta(\s{H})$.
\end{lemma}
\begin{proof}
Let $\Xi\in{{\tt C}}_\eta(\s{H})$. Since $\eta\Xi$ is strictly positive and bounded the functional calculus 
allows to define the bounded self-adjoint operator 
$Q\equiv Q (\Xi)=\log(\eta\Xi)$. Then
$\expo{Q}=\eta\Xi$ and $\expo{-Q}=(\eta\Xi)^{-1}=\Xi\eta$. Therefore $\expo{-Q}=\eta\expo{Q}\eta=\expo{\eta Q\eta}$ where the last equality is easily justified by the unitarity of $\eta$ and the  boundedness
of $Q$. This proves that $Q\eta+\eta Q=0$, namely $Q\in {{\tt R}}_\eta(\s{H})$. On the other hand let $Q\in {{\tt R}}_\eta(\s{H})$ and define $\Xi\equiv\Xi(Q)=\eta\expo{Q}$. It follows that $\Xi\eta=\eta\expo{Q}\eta=\expo{-Q}$ and in turn $\Xi^2=\n{1}$. Moreover,  $\eta\Xi=\expo{Q}$ is strictly positive by construction, hence $\Xi\in {{\tt C}}_\eta(\s{H})$ . By observing that
$Q(\Xi(Q))=\log(\eta^2\expo{Q})=Q$  and $\Xi(Q(\Xi))=\eta\expo{\log(\eta \Xi)}=\Xi$ it follows that 
${{\tt C}}_\eta(\s{H})$ and ${{\tt R}}_\eta(\s{H})$ are in one-to-one
correspondence. 

Now the topology. Let us start proving that the map   $Q\mapsto \Xi(Q)$ is continuous. As for this, let $\{Q_n\}\subset {{\tt R}}_\eta(\s{H})$ be a sequence converging to $Q$. Given a $\delta>0$ there exists an $N\in\N$ such that
$$
\big|\|Q_n\|\;-\;\|Q\|\big|\;\leqslant\;\|Q_n\;-\;Q\|\;<\;\delta\;,\qquad\quad\forall n> N\;.
$$
Let $r_\delta:=\max\{\|Q_1\|,\ldots,\|Q_N\|, \|Q\|+\delta\}$. It follows that $\|Q\|< r_\delta$ and $\|Q_n\|< r_\delta$ for all $n\in\N$. If $\lambda\in\C$ such that $|\lambda|=r_\delta$ the resolvents $R_n(\lambda):=(Q_n-\lambda\n{1})^{-1}$ and $R(\lambda):=(Q-\lambda\n{1})^{-1}$ are well defined and $\|R_n(\lambda)-R(\lambda)\|\to0$ in view of the identity
$$
R_n(\lambda)\;=\;R(\lambda)\;\big(\n{1}\;-\;(Q_n-Q)\;R(\lambda)\big)^{-1}\;.
$$
Let $C_\delta:=\{\lambda\in\C\;|\; |\lambda|=r_\delta\}$. With the help of the holomorphic functional calculus one can write
$$
\eta\expo{Q_n}\;-\;\eta\expo{Q}\;=\;\eta\;\frac{\ii}{2\pi}\oint_{C_\delta}\expo{\lambda}\big(R_n(\lambda)\;-\;R(\lambda)\big)\;\dd\lambda
$$
which, in view  of the unitarity of $\eta$, leads to the norm estimate
$$
\left\|\Xi(Q_n)\;-\;\Xi(Q)\right\|\;\leqslant\;K_{r_\delta}\;\|R_n(\lambda)-R(\lambda)\|
$$
where the constant $K_{r_\delta}>0$ depends on $r_\delta$. This proves that the sequence $\Xi(Q_n)$ converges in norm towards $\Xi(Q)$, namely the maps $Q\mapsto \Xi(Q)$ is continuous.
In order to prove the continuity of the 
 inverse map $\Xi\mapsto Q(\Xi)$ let us observe first that
$(\eta\Xi)^{-1}=\Xi\eta=\eta(\eta\Xi)\eta$. Since $\eta$ is unitary it follows that
$ \sigma(\eta\Xi)=\sigma((\eta\Xi)^{-1})$ and in view of the \emph{spectral mapping theorem} one has that $\lambda\in \sigma(\eta\Xi)$ if and only if $\lambda^{-1}\in \sigma(\eta\Xi)$. This 
implies that $\|\Xi\|=\|\eta\Xi\|=\max\sigma(\eta\Xi)\geqslant 1$ and $\sigma(\eta\Xi)\subseteq [\|\Xi\|^{-1},\|\Xi\|]$.
Let $\{\Xi_n\}\subset {{\tt C}}_\eta(\s{H})$ be a sequence converging to $\Xi$. Given a $\delta>0$ there exists an $N\in\N$ such that
$$
\big|\|\Xi_n\|\;-\;\|\Xi\|\big|\;\leqslant\;\|\Xi_n\;-\;\Xi\|\;<\;\delta\;,\qquad\quad\forall n> N\;.
$$
Let  $b_\delta:=\max\{\|\Xi_1\|,\ldots,\|\Xi_N\|, \|\Xi\|+\delta\}$ and $I_\delta:=(b_\delta^{-1},b_\delta)$. It follows that $\|\Xi\|< b_\delta$ and $\|\Xi_n\|< b_\delta$ for all $n\in\N$ and in turn $\sigma(\eta\Xi)\subset I_\delta$ and $\sigma(\eta\Xi_n)\subset I_\delta$
for all $n\in\N$. Let $\tilde{C}_\delta\subset \C$ be any closed Jordan curve that crosses the real axis only in 
$b_\delta^{-1}$ and $b_\delta$.
 If $\lambda\in\tilde{C}_\delta$  the resolvents $\tilde{R}_n(\lambda):=(\eta\Xi_n-\lambda\n{1})^{-1}$ and $\tilde{R}(\lambda):=(\eta\Xi-\lambda\n{1})^{-1}$ are well defined and $\|\tilde{R}_n(\lambda)-\tilde{R}(\lambda)\|\to0$ in view of the identity
$$
\tilde{R}_n(\lambda)\;=\;\tilde{R}(\lambda)\;\big(\n{1}\;-\;\eta(\Xi_n-\Xi)\;\tilde{R}(\lambda)\big)^{-1}\;.
$$
With the help of the holomorphic functional calculus one can write
$$
\log(\eta\Xi_n)\;-\;\log(\eta\Xi)\;=\; \frac{\ii}{2\pi}\oint_{\tilde{C}_\delta}\log({\lambda})\big(\tilde{R}_n(\lambda)\;-\;\tilde{R}(\lambda)\big)\;\dd\lambda
$$
where the logarithm is unambiguously defined since the curve $\tilde{C}_\delta$ does not enclose the origin of $\C$. The last equality provides
$$
\left\|Q(\Xi_n)\;-\;Q(\Xi)\right\|\;\leqslant\;\tilde{K}_{b_\delta}\;\|\tilde{R}_n(\lambda)-\tilde{R}(\lambda)\|
$$
where the constant $\tilde{K}_{b_\delta}>0$ depends on $b_\delta$. This proves that the sequence $Q(\Xi_n)$ converges in norm towards $Q(\Xi)$, namely also the maps $\Xi\mapsto Q(\Xi)$ is continuous.
In summary we proved that the spaces ${{\tt C}}_\eta(\s{H})$ and ${{\tt R}}_\eta(\s{H})$ are homeomorphic.

Since the multiplication and the adjoint are continuous operations with respect to the operator norm  it follows that ${{\tt R}}_\eta(\s{H})$ is a \emph{closed} subset of $\bb{B}_{\rm lin}(\s{H})$.
Since ${{\tt C}}_\eta(\s{H})$ is  homeomorphic to ${{\tt R}}_\eta(\s{H})$ it follows that also 
${{\tt C}}_\eta(\s{H})$ is closed.
\end{proof}

\medskip

\noindent
 Lemma \ref{lem:space_of_C_symmetries}  states that any $\s{C}$-symmetry $\Xi$ can be uniquely represented
by  a self-adjoint and  anti-$\eta$-self-adjoint  operator $Q$. Also the set ${{\tt R}}_\eta(\s{H})$ has a \emph{trivial} element $Q_0:=0$. In fact this is the  element associated to $\Xi_0=\eta$ through the homeomorphism
described in  Lemma \ref{lem:space_of_C_symmetries}.

\begin{remark}[Generalized $\s{C}$-symmetries]\label{rk:gen-C-sym}{\upshape
It is possible to generalize Definition \ref{dfn:C_symmetry} by relaxing the boundedness condition for $\Xi$, see \eg \cite[Section 4]{kuzhel-09}. Indeed, it is possible to construct unbounded involutions in infinite dimensional normed linear spaces. For instance given an orthonormal basis $\{\psi_n\}$ of the Hilbert space $\s{H}$ the operator $\Xi$ defined by $\Xi\psi_n:=2n\psi_1-\psi_n$ is unbounded and verifies $\Xi^2=\n{1}$ (\cf  \cite{buckholtz-00}). Moreover, also in a finite dimensional space there are $\s{C}$-symmetries of arbitrarily large norm (\cf Lemma \ref{lwmma:appB-Csymm-anticomm}).
Anyway, also in the unbounded case the representation described in Lemma \ref{lem:space_of_C_symmetries} continues to be valid  \cite[Theorem 6.2.3]{albeverio-kushel-15}.
However, in this work we will not be interested in this kind of generalization.
In fact \cite[Proposition 3.2]{kuzhel-sudilovskaya-17} shows that in the context of gapped $\eta$-self-adjoint operator only the notion of bounded $\s{C}$-symmetry is relevant.
For this reason we will focus our attention on Definition \ref{dfn:C_symmetry}.
}\hfill $\blacktriangleleft$
\end{remark}

The following result will be relevant in the next Section.
\begin{lemma}\label{lemma_reduc_Xi}
Let  $(\s{H},\langle\langle\;,\;\rangle\rangle_\eta)$ be the indefinite metric space associated to a fundamental symmetry $\eta^*=\eta=\eta^{-1}$. Let $\Xi\in {{\tt C}}_\eta(\s{H})$ be a $\s{C}$-symmetry and $Q=Q(\Xi)\in {{\tt R}}_\eta(\s{H})$ the related self-adjoint and  anti-$\eta$-self-adjoint  operator associated to $\Xi$ through the bijection 
described in  Lemma \ref{lem:space_of_C_symmetries}. The operator 
$$
G_\Xi\;:=\;\expo{\frac{1}{2}Q}\;=\;\expo{\frac{1}{2}\log(\eta\Xi)}\;=\;\sqrt{\eta\Xi}
$$
is $\eta$-unitary, self-adjoint and $G_\Xi \Xi G_\Xi^{-1}=\eta$.
\end{lemma}
\proof
The operator $G_\Xi$ is defined via spectral calculus. The representation $G_\Xi=\sqrt{\eta\Xi}$  makes evident that $G_\Xi=G_\Xi^*$
 is a bounded  self-adjoint operator. Moreover, $G_\Xi$ is invertible with inverse $G_\Xi^{-1}=\expo{-\frac{1}{2}Q}=\eta G_\Xi\eta=G_\Xi^\sharp$. The last expression shows that $G_\Xi\in {\tt U}_\eta(\s{H})$ is a bounded
$\eta$-unitary operator. Finally the computation
$$
G_\Xi^{-1} \eta G_\Xi\;=\;\expo{-\frac{1}{2}Q}\eta\expo{\frac{1}{2}Q}\;=\;\eta\expo{Q}\;=\;\Xi
$$
completes the proof.
\qed

\medskip

\noindent
Lemma \ref{lemma_reduc_Xi} states that any $\s{C}$-symmetry $\Xi$ can be reduced to the trivial $\s{C}$-symmetry $\Xi_0=\eta$ by means of a transformation $G_\Xi$ which preserves the geometry of the  indefinite metric space $(\s{H},\langle\langle\;,\;\rangle\rangle_\eta)$.

\begin{remark}{\upshape If one 
renounces to the condition that $\eta$ is an involution as in the original Definition \ref{def:001}, then the second condition in \eqref{eq:C-sym1}
for the definition of a $\s{C}$-symmetry must be replaced by $\frac{\eta}{|\eta|}\Xi>0$. 
With this modification also the results proved in Lemma \ref{lem:space_of_C_symmetries} and Lemma 
\ref{lemma_reduc_Xi} still work.
}\hfill $\blacktriangleleft$
\end{remark}


\subsection{Operators with  $\s{C}$-symmetry}
\label{sec:op-C-sym}
Let us focus on the interplay between the notion of $\eta$-self-adjointness and the existence of $\s{C}$-symmetries. The following definition is borrowed from \cite[Definition 6.3.3]{albeverio-kushel-15}:
\begin{definition}[Operators with  $\s{C}$-symmetry]
Let  $(\s{H},\langle\langle\;,\;\rangle\rangle_\eta)$ be the indefinite metric space associated to a fundamental symmetry $\eta^*=\eta=\eta^{-1}$. A densely defined operator $A$ possesses a  $\s{C}$-symmetry if there exists an element $\Xi\in {{\tt C}}_\eta(\s{H})$ such that
$
\Xi  A=A  \Xi
$.
\end{definition}

\medskip

In this work we are mainly interested in $\eta$-self-adjoint operators with $\s{C}$-symmetry. The following result is known in the literature \cite[Theorem 6.3.4]{albeverio-kushel-15}:
\begin{proposition}\label{prop:H-Xi}
Let  $(\s{H},\langle\langle\;,\;\rangle\rangle_\eta)$ be the indefinite metric space associated to a fundamental symmetry $\eta^*=\eta=\eta^{-1}$. Let $H$ be an $\eta$-self adjoint operator with dense domain $\s{D}(H)$ which admits a $\s{C}$-symmetry $\Xi\in {{\tt C}}_\eta(\s{H})$. Let $G_\Xi$ be the associated $\eta$-unitary operator described in Lemma \ref{lemma_reduc_Xi}. Then the operator
$$
\tilde{H}\;:=\;G_\Xi H G_\Xi^{-1}
$$
is self-adjoint with domain $\s{D}(\tilde{H}):=G_\Xi[\s{D}({H})]$ and commutes  with the fundamental symmetry $\eta$.
\end{proposition}
\proof
Let us start by observing that in the case of an unbounded operator $H$ the condition $\Xi  H=H  \Xi$
and the invertibility of $\Xi$ imply that $\Xi[\s{D}(H)]= \s{D}(H)$. The operator ${H}_\Xi$ is well defined on the dense domain $\s{D}(\tilde{H})$. 
From the the identity  
$\Xi=
G_\Xi^{-1} \eta G_\Xi$ derived in the proof of Lemma \ref{lemma_reduc_Xi} it follows that
$$
G_\Xi^{-1} \eta G_\Xi H\;=\; HG_\Xi^{-1} \eta G_\Xi\;\quad\;\Leftrightarrow\;\quad\;\eta\tilde{H}\;=\;\tilde{H}\eta
$$
namely $\tilde{H}$ commutes with $\eta$. On the other hand 
 $G_\Xi$ is $\eta$-unitary and therefore Corollary \ref{cor1_inv_eta_self} ensures that $\tilde{H}$ is $\eta$-self-adjoint. At level of domains one has that
$$
\s{D}(\tilde{H}^*)\;=\;\eta[\s{D}(\tilde{H})]\;=\;(\eta G_\Xi)[\s{D}({H})]\;=\;(\eta G_\Xi\Xi)[\s{D}({H})]\;=\;G_\Xi[\s{D}({H})]\;=\;\s{D}(\tilde{H})\;.
$$
Moreover, since $\tilde{H}$ is  $\eta$-self-adjoint and commutes with $\eta$ it follows that 
$$
\tilde{H}^*\;=\;\eta\tilde{H}\eta\;=\;\tilde{H}\;,
$$
namely $\tilde{H}$ is self-adjoint. 
\qed

\medskip

Proposition \ref{prop:H-Xi} states that an $\eta$-self-adjoint operator with a $\s{C}$-symmetry is similar to a self-adjoint operator, or  \emph{dynamically stable} by
 using the jargon of Definition \ref{def:dyn_stab}. Also the converse is true:
\begin{theorem}[{\cite[Theorem 6.3.4]{albeverio-kushel-15}}]
\label{theo:albeverio-kushel}
An $\eta$-self-adjoint operator is dynamically stable if and only if it has a $\s{C}$-symmetry.
\end{theorem}

\medskip

\noindent 
The next result shows that dynamically stable $\eta$-self-adjoint operators possess the characteristics for the definition of a quantum dynamics as discussed in Section \ref{sec:introduction_motivation}.
\begin{corollary}[Induced functional calculus]
\label{col:func-calc}
Let  $(\s{H},\langle\langle\;,\;\rangle\rangle_\eta)$ be the indefinite metric space associated to a fundamental symmetry $\eta^*=\eta=\eta^{-1}$. Let $H$ be an $\eta$-self adjoint operator with dense domain $\s{D}(H)$ which admits a $\s{C}$-symmetry $\Xi\in {{\tt C}}_\eta(\s{H})$. 
\begin{itemize}
\item[(1)] The spectrum of $H$ is real, $\sigma(H)\subseteq\R$.
\vspace{1mm}
\item[(2)] There is a $*$-algebra homomorphism $\varphi:L^\infty(\R)\to \bb{B}_{\rm lin}(\s{H})$ defined by
$$
\varphi(f)\;\longmapsto\;f(H)\;:=\;G_\Xi^{-1}f(\tilde{H})G_\Xi\;.
$$
\item[(3)] The map $\R\ni t\mapsto V_t:=\expo{-\ii tH}$
provides a group of $\eta$-unitary operators that solves the Schr\"odinger equation
$$
\ii\frac{\partial}{\partial t}\psi\;=\;H\psi\;,\qquad\quad\psi\in\s{D}({H})\;.
$$
\end{itemize}
\end{corollary}
\proof[{Proof} (sketch of).]
Item (1) follows from the fact that $H$ and $\tilde{H}$ are intertwined by the invertible operator $G_\Xi$. This implies that $\sigma(H)=\sigma(\tilde{H})\subseteq\R$ where the last inclusion follows since $\tilde{H}$ is self-adjoint. Item (2) follows from the functional calculus for the self-adjoint operator $\tilde{H}$. As for item (3) one has that $ V_t=G_\Xi^{-1}\expo{-\ii t\tilde{H}}G_\Xi$ and  $\expo{-\ii t \tilde{H}}$ solves the
Schr\"odinger equation for $\tilde{H}$, namely
$$
\ii\frac{\partial}{\partial t}\left(\expo{-\ii t \tilde{H}}G_\Xi\psi\right)\;=\;\tilde{H}\left(\expo{-\ii t \tilde{H}}G_\Xi\psi\right)\;,\qquad\quad\psi\in\s{D}({H})\;.
$$
A multiplication on the left by $G_\Xi^{-1}$ provides
$$
\ii\frac{\partial}{\partial t}\left(V_t\psi\right)\;=\;H\left(V_t\psi\right)\;,\qquad\quad\psi\in\s{D}({H})\;.
$$
The $\eta$-unitarity of  $V_t$ follows since $G_\Xi$ is $\eta$-unitarity and $\expo{-\ii tH_\Xi}$ commutes with $\eta$.
\qed

\medskip

Let us focus on the  bounded case. 
Given a $\s{C}$-symmetry $\Xi\in {{\tt C}}_\eta(\s{H})$
let 
$$
{\tt H}_{\eta,\Xi}(\s{H})\;:=\;\left\{H\in\ {\tt H}_{\eta}(\s{H})|\  \Xi H = H\Xi\right\}
$$
be the set of bounded $\eta$-self-adjoint  operators that admit $\Xi$ as $\s{C}$-symmetry. More in general 
$$
{\tt CH}_{\eta}(\s{H})\;:=\;\bigsqcup_{\Xi\in {{\tt C}}_\eta(\s{H})}{\tt H}_{\eta,\Xi}(\s{H})
$$
is the set of all $\eta$-self-adjoint  operators which admit a $\s{C}$-symmetry. The case related to the trivial $\s{C}$-symmetry $\Xi_0:=\eta$ is of particular relevance. One has that
$$
{\tt H}_{\eta,0}(\s{H})\;:=\;{\tt H}_{\eta,\Xi_0}(\s{H})\;=\;\left\{H\in\ {\tt H}_{\eta}(\s{H})|\   H = H^*\right\}
$$
is the space of self-adjoint and $\eta$-self-adjoint bounded operators. As immediate consequence of Lemma \ref{lemma_reduc_Xi} and Proposition \ref{prop:H-Xi} one has that:
\begin{corollary}\label{coro_reduc_eta_self_comm}
For any $\Xi\in {{\tt C}}_\eta(\s{H})$ there is an  $\eta$-unitary equivalence
$$
{\tt H}_{\eta,\Xi}(\s{H})\;\stackrel{}{\simeq}\;{\tt H}_{\eta,0}(\s{H})
$$
induced by the invertible map ${\rm Ad}_\Xi:{\tt H}_{\eta,\Xi}(\s{H})\to {\tt H}_{\eta,0}(\s{H})$ given by
$$
{\rm Ad}_\Xi(H)\;:=\;G_\Xi H G_\Xi^{-1}\;=\;\tilde{H}\;.
$$
\end{corollary}

\medskip

The interplay between the $\eta$-quantum symmetries ${\tt QS}_\eta(\s{H})$ described in Section \ref {sec:krein_space-Qsim}  and the  $\s{C}$-symmetries is provided by the following result.

\begin{proposition} \label{prop:reduction_symmetry}
Let  $(\s{H},\langle\langle\;,\;\rangle\rangle_\eta)$ be the indefinite metric space associated to a fundamental symmetry $\eta^*=\eta=\eta^{-1}$, and $U\in {\tt QS}_\eta(\s{H})$ an $\eta$-quantum symmetry of type $(\wp, \varpi)$. 
\begin{itemize}
\item[(1)] The map $U$ induces an invertible transformation ${\rm Ad}_U^\wp: {{\tt C}}_\eta(\s{H})\to {{\tt C}}_\eta(\s{H})$ of the space of $\s{C}$-symmetries in itself given by
$$
{\rm Ad}_U^\wp(\Xi)\;:=\;\wp(U)\; U\Xi U^{-1}\;.
$$
Moreover, if $H\in {\tt H}_{\eta,\Xi}(\s{H})$ then ${\rm Ad}_U(H)\in {\tt H}_{\eta,{\rm Ad}_U^\wp(\Xi)}(\s{H})$.
 \vspace{1mm}
\item[(2)] If ${\rm Ad}_U^\wp(\Xi)=\Xi$ then the  diagram 
$$
\begin{diagram}
{\tt H}_{\eta,\Xi}(\s{H})&\rTo^{{\rm Ad}_\Xi}&       {\tt H}_{\eta,0}(\s{H})\\
\dTo^{{\rm Ad}_U}&&\dTo_{{\rm Ad}_{\tilde{U}}}\\
{\tt H}_{\eta,\Xi}(\s{H})&\rTo_{{\rm Ad}_\Xi}&       {\tt H}_{\eta,0}(\s{H})\\
\end{diagram} 
$$
is commutative, with $\tilde{U}:=G_\Xi UG_\Xi^{-1}$. Moreover,  $\tilde{U}$ is of type $(\wp, \varpi)$.
\end{itemize}
\end{proposition}
\begin{proof}
(1) 
The operator $\Xi' := {\rm Ad}_U^\wp(\Xi)$ is clearly involutive. 
Moreover, for every $\varphi\neq 0$ one has 
$$
\begin{aligned}
\langle\varphi,\eta\Xi'\varphi\rangle\;&=\;\wp(U)\;\langle\langle\varphi, U\Xi U^{-1}\varphi\rangle\rangle_\eta\\
&=\;\wp(U)\wp(U^{-1})\;\left|\langle\langle U^{-1}\varphi, \Xi U^{-1}\varphi\rangle\rangle_\eta\right|^{1-\varpi(U^{-1})}\;  \langle\langle U^{-1}\varphi, \Xi U^{-1}\varphi\rangle\rangle_\eta^{\varpi(U^{-1})}\\
&=\;\langle U^{-1}\varphi, \eta \Xi U^{-1}\varphi\rangle\;>\;0
\end{aligned}
$$
since  $\wp(U)\wp(U^{-1})=1$ and $\eta \Xi>0$. This implies that also $\eta\Xi'>0$ is strictly positive.
Finally, a straightforward computation shows that $[H,\Xi]=0$ if and only if $[{\rm Ad}_U(H),\Xi']=0$ and 
${\rm Ad}_U(H)$ is $\eta$-self-adjoint in view of Corollary \ref{cor1_inv_eta_self}.
(2) The commutativity of the diagram follows from a direct computation. Since $G_\Xi$ is $\eta$-unitary as proved in Lemma \ref{lemma_reduc_Xi}, a direct computation shows that 
the type of 
$\tilde{U}$ only depends on the type  of $U$.
\end{proof}


\subsection{The Hilbert bundle picture}
\label{sec:hilb_pict}
The various notions introduced in Sections \ref{sec:indefinite}, \ref{sec:Q-sym} and \ref{sec:C-symm} can be extended to the case of  \emph{Hilbert bundles}  and complex vector bundles. The motivations of this generalization will be sketched  in Remark \ref{rk:motiv_hilb_bund} below. For the general theory of Hilbert bundles we refer to \cite{dixmier-douady-63,dupre-74} and references therein. For the general of vector bundles we refer  to the two classic monographs
\cite{atiyah-67,husemoller-94}.

\medskip

Let us recall that a Hilbert bundle $\pi:\bb{E}\to X$ is a locally trivial bundle over $X$ such that each fiber $\s{H}_x:=\pi^{-1}(x)$ is a separable Hilbert space with scalar product $\langle\;,\;\rangle_x$ and the projection $\pi$ is an open surjective map. 
When the fibers of the Hilbert bundle $\pi:\bb{E}\to X$ are infinite dimensional then the transition functions
of 
$\bb{E}$ take value on the
structure group   which is the  infinite dimensional unitary group ${\tt U}(\s{H})$ topologized with the strong topology\footnote{In \cite{freed-moore-13}  the structure group is topologized by  the compact-open topology in the sense of Atiyah and Segal. However,  as discussed in Remark \ref{rk:top_strong}, the strong topology should be sufficient for all needs. The  strong topology has been used also in \cite[Chapitre II.10]{dixmier-douady-63,dupre-74}}. For the aims of this work it is sufficient to assume that: 
\begin{assumption}
\label{Ass:reg_X}
The base space $X$ is Hausdorff and compact and is endowed with a regular Borel measure $\mu$. 
\end{assumption}
Under Assumtion \ref{Ass:reg_X}, if  $\s{H}_x\simeq\C^m$ for all $x\in X$ one has that $\pi:\bb{E}\to X$ is (equivalent to) a rank $m$ complex  vector bundle endowed with a Hermitian metric. In this case the structure group is the unitary group ${\tt U}(\C^m)$ with its natural topology. 
\medskip

A morphism between two Hilbert bundles  $\pi:\bb{E}\to X$ and $\pi':\bb{E}'\to X$ is a  continuous map $f:\bb{E}\to \bb{E}'$ which preserves the fiber structure in the sense that $\pi=\pi'\circ f$ and the linear structure of each fiber in the sense that $f|_{\s{H}_x}:\s{H}_x\to\s{H}'_x$ is linear or anti-linear.
Let ${\rm End}(\bb{E})$ be the space of the \emph{endomorphisms} of the Hilbert bundle $\pi:\bb{E}\to X$. This means that   elements $A\in {\rm End}(\bb{E})$ are continuous maps of the total space $\bb{E}$ such that  $\pi=\pi\circ A$  and 
that $A_x:=A|_{\s{H}_x}\in \bb{B}(\s{H}_x)$ is fiberwise a linear or anti-linear bounded  operator. Diagrammatically one has that elements 
$A\in {\rm End}(\bb{E})$ make the  diagrams
$$
\begin{diagram}
\bb{E}&& \rTo^{A} &&\bb{E}\\
&\rdTo_\pi&& \ldTo_{\pi}& \\
&&X &&  \\
\end{diagram}
$$ 
commutative.
The subset of the \emph{fiber-linear} endomorphisms will be denoted with  ${\rm End}_{\rm lin}(\bb{E})$. The adjoint of $A\in {\rm End}(\bb{E})$ is the endomorphism $A^*\in {\rm End}(\bb{E})$ defined fiberwise by $A^*|_{\s{H}_x}=A_x^*$. It is self-adjoint if $A=A^*$ and it is an involution if $A=A^{-1}$. 
\begin{definition}[Krein bundle]\label{def:krein_bundle}
An \emph{indefinite metric structure} on the Hilbert bundle $\pi:\bb{E}\to X$ is given by a self-adjoint involution
$\eta\in{\rm End}_{\rm lin}(\bb{E})$,  
 such that each restriction $\eta_x:=\eta|_{\s{H}_x}$  is   a fundamental symmetry on the fiber  $\s{H}_x$,  \emph{\ie} $\eta^*_x=\eta_x=\eta^{-1}_x$. If for each $x\in X$ the metric $\eta_x$ provides a Krein space structure in the fiber $\s{H}_x$ (in the sense of Definition \ref{def:krein})  then  the pair $(\bb{E},\eta)$ is called a \emph{Krein bundle}. In the  case the fibers of $\bb{E}$ are of finite even dimension then  $(\bb{E},\eta)$ is called a \emph{Krein vector bundle}.
 \end{definition}

\medskip
A morphism between two Krein bundles  $(\bb{E},\eta)$ and $(\bb{E}',\eta')$ is a  Hilbert bundle morphism $f:\bb{E}\to \bb{E}'$ such that $f\circ \eta=\eta' \circ f$.
The notion of $\eta$-self-adjointness (\cf Definition \ref{dfn:eta_self_adjoint}) and $\s{C}$-symmetry (\cf Definition \ref{dfn:C_symmetry}) have a natural extension to the case of Krein bundles.
Let $(\bb{E},\eta)$ be a Krein bundle. An element $H\in{\rm End}_{\rm lin}(\bb{E})$ is $\eta$-self-adjoint if $H=H^\sharp:=\eta\circ H^*\circ \eta$. The subset of the $\eta$-self-adjoint endomorphisms   will be denoted by ${\tt H}_{\eta}(\bb{E})\subset {\rm End}_{\rm lin}(\bb{E})$.
A $\s{C}$-symmetry is an element $\Xi\in{\rm End}_{\rm lin}(\bb{E})$ such that $\Xi=\Xi^{-1}$ and the composition $\eta\circ \Xi$ is strictly positive when restricted to each fiber, \ie
$(\eta\circ \Xi) |_{\s{H}_x}=\eta_x \Xi_x>0$.  The subset of the $\s{C}$-symmetries will be denoted by
${\tt C}_{\eta}(\bb{E})\subset {\rm End}_{\rm lin}(\bb{E})$. In accordance with Definition \ref{dfn:C_symmetry} one says that $H\in{\tt H}_{\eta}(\bb{E})$ admits a $\s{C}$-symmetry if there exists an element $\Xi\in{\tt C}_{\eta}(\bb{E})$ such that $H\circ\Xi=\Xi\circ H$. The subset of the $\eta$-self-adjoint endomorphisms with  $\s{C}$-symmetry $\Xi$ will be denoted with 
${\tt H}_{\eta,\Xi}(\bb{E})\subset {\tt H}_{\eta}(\bb{E})$ while 
${\tt CH}_{\eta}(\bb{E})\subset {\tt H}_{\eta}(\bb{E})$  will denote the subset of the $\s{C}$-symmetric 
$\eta$-self-adjoint endomorphisms. Proposition \ref{prop:H-Xi} and Corollary \ref{col:func-calc} (and also the more general Theorem \eqref{theo:albeverio-kushel}) can be reproved in the setting of Krein bundles with a fiberwise version of the same proofs. In particular, let ${\tt U}_\eta(\bb{E})\subset {\rm End}_{\rm lin}(\bb{E})$ be the subset of the $\eta$-unitary (fiber-linear) endomorphisms, \ie $U\in {\tt U}_\eta(\bb{E})$ if and only if $U^{-1}=U^\sharp$. To any $\s{C}$-symmetry
$\Xi\in{\tt C}_{\eta}(\bb{E})$ it  can be associated a $G_\Xi\in {\tt U}_\eta(\bb{E})$ defined fiberwise by $G_\Xi |_{\s{H}_x}:=\sqrt{\eta_x \Xi_x}$. As in Lemma \ref{lemma_reduc_Xi} one can prove that $G_\Xi=G_\Xi^*$ and $G_\Xi \Xi G_\Xi^{-1}=\eta$. Moreover, a fiberwise adaptation of Corollary \ref{coro_reduc_eta_self_comm} shows that
for any  $H\in{\tt H}_{\eta,\Xi}(\bb{E})$ the transformed endomorphism
 $\tilde{H}:=G_\Xi\circ H\circ G_\Xi^{-1}$ is self-adjoint $\tilde{H}=\tilde{H}^*$ and commute with the metric structure, $\tilde{H}\circ \eta=\eta \circ \tilde{H}$. The set of self-adjoint endomorphisms commuting with $\eta$ will be denoted with ${\tt H}_{\eta,0}(\bb{E})$. In view of the  equivalence induced by $G_\Xi$ one concludes that the classification of the $\s{C}$-symmetric 
$\eta$-self-adjoint endomorphisms is equivalent to the classification of ${\tt H}_{\eta,0}(\bb{E})$. This will be the main task of  Section \ref{sec:freed_moore_K}. Given the  Krein bundle $(\bb{E},\eta)$ and using a local version of the  construction of Note \ref{note:eta-compl} and  Note \ref{note:rqref-reflex}
one can define an $\eta$-reflecting endomorphism $R^*=R^{-1}=-R^\sharp$ and an anti-linear complex structure $C^*=C^\sharp=C=C^{-1}$ such that $R\circ C=C\circ R$. With these endomorphisms and by adapting the prescriptions \eqref{eq:rep-group_eta-ant-unit}, 
\eqref{eq:rep-group_pseud-eta-unit}  and \eqref{eq:rep-group_pseud-eta-ant-unit} one can define the group 
$$
 {\tt QS}_\eta(\bb{E})\;:=\;{\tt U}_\eta(\bb{E})\;\sqcup\; {\tt AU}_\eta(\bb{E})\sqcup\; {\tt PU}_\eta(\bb{E})\sqcup\; {\tt PAU}_\eta(\bb{E})\;\subset\; {\rm End}(\bb{E})\,.
 $$
We refer to  ${\tt QS}_\eta(\bb{E})$ as the group of the \emph{fiber-preserving} $\eta$-quantum symmetries of the Krein bundle $(\bb{E},\eta)$. Let us remark that any $U\in  {\tt QS}_\eta(\bb{E})$ meets the strong condition $\pi=\pi\circ U$.  Indeed, by relaxing the latter condition a more general notion of symmetry can be introduced (\cf Definition \ref{def:cov_endo}).
Anyway,  $U\in  {\tt QS}_\eta(\bb{E})$ is a \emph{good} symmetry for $H\in{\tt H}_{\eta,\Xi}(\bb{E})$
if: (i) $U\circ H=c(U) H\circ U$, and (ii) $U\circ \Xi=\wp(U) \Xi\circ U$.
Here the signs $\wp(U),\varpi(U)\in\{-1,+1\}$ describe the nature of the symmetry $U$ according to the convention of Table 1.3. The sign $c(U)\in\{-1,+1\}$ says if $U$ is a proper or an improper symmetry of $H$ according to the terminology introduced in Section \ref{subsec:metod_CAZ}. 
The  condition (ii) expresses the \emph{compatibility} of the symmetry $U$ with the $\s{C}$-symmetry property of $H$. The fiberwise version of the proof of Theorem \ref{theo_fund_rep_res} shows that $\tilde{U}:=G_\Xi\circ U\circ G_\Xi^{-1}\in {\tt QS}_\eta(\bb{E})$ has the   same nature of $U$ (\ie $\wp(U)=\wp(\tilde{U})$ and $\varpi(U)=\varpi(\tilde{U})$) 
and is \emph{unitary} $\tilde{U}^{-1}=\tilde{U}^*$. In particular the later condition implies $\tilde{U}\circ \eta=\wp(\tilde{U}) \tilde{U} \circ \eta$. Moreover,  $\tilde{U}\circ \tilde{H}=c(\tilde{U}) \tilde{H}\circ \tilde{U}$ with   $c(\tilde{U})=c(U)$. In conclusion we can sum up as follows: 

\smallskip

\noindent 
{\bf Observation A.}
\emph{Once the problem of the classification of the $\s{C}$-symmetric 
$\eta$-self-adjoint endomorphisms is reduced  to the problem of the classification of ${\tt H}_{\eta,0}(\bb{E})$ it turns out that the only relevant  (fiber-preserving) $\eta$-quantum symmetries are induced by \emph{unitary} or \emph{anti-unitary} endomorphisms which \emph{commute} or \emph{anti-commute} with the metric $\eta$.}

\medskip

We will see below that the framework described in Observation A deserves a  further, non-trivial generalization (\cf Observation B). For the moment, let us introduce an extra structure that will be relevant for the purposes of Section \ref{sec:freed_moore_K}: The action of    \emph{Clifford algebras} on Hilbert and Krein bundles.
For a detailed introduction to Clifford algebras we refer to \cite[Chapter III, Section 3]{karoubi-78} or \cite{lee-48}.
Let $C\ell^{r,s}$ be the Clifford algebra
generated over $\R$ by a collection of symbols $e_1,\ldots,e_{r+s}$ subjected to the following relations 
$$
\frac{e_ie_j+e_je_i}{2}\;:=\;
\left\{
\begin{aligned}
&\phantom{++}0&\qquad&\text{if}\quad i\neq j\\
&+\n{1}&\qquad&\text{if}\quad i= j\in \{1,\ldots,r\}\\
&-\n{1}&\qquad&\text{if}\quad i= j\in \{r+1,\ldots,r+s\}\;.\\
\end{aligned}
\right.
$$
 These Clifford algebras have been completely classified.

\begin{definition}[Clifford action on Hilbert and Krein  bundles]
\label{def:cliff_act}
Let $\pi:\bb{E}\to X$ be a  Hilbert  bundle and $C\ell^{r,s}$ a Clifford algebra.
A (unitary) Clifford action on $\bb{E}$ is given by an $\R$-algebra homomorphism $\gamma: C\ell^{r,s}\to {\rm End}_{\rm lin}(\bb{E})$ such that the representatives of the generators $\gamma_i:=\gamma(e_i)$ are unitary, \emph{\ie} $\gamma_i^{-1}=\gamma_i^*$  for all $i=1,\ldots,r+s$. The pair $(\bb{E},\gamma)$ will be called a \emph{Clifford-Hilbert bundle} of type $(r,s)$. In the case of a Krein bundle $(\bb{E},\eta)$ a \emph{compatible} Clifford action
is subjected to the extra condition $\gamma(a)\circ \eta=\eta \circ \gamma(a)$ for all $a\in C\ell^{r,s}$. The triple $(\bb{E},\eta,\gamma)$ given by  the Krein bundle $(\bb{E},\eta)$ and the  compatible Clifford action $\gamma$ will be called a \emph{Clifford-Krein bundle}.
\end{definition}

\medskip

\noindent
A morphism between two Clifford-Hilbert bundles (resp. Clifford-Krein bundles) of the same type $(\bb{E},\gamma)$ and $(\bb{E}',\gamma')$ is a Hilbert bundle (resp. Krein bundles) morphism $f:\bb{E}\to \bb{E}'$ such that $f\circ \gamma(a)=\gamma'(a) \circ f$ for all $a\in C\ell^{r,s}$. 
Definition \ref{def:cliff_act} also works in the vector bundle (\ie finite rank) case.

\medskip

Up to now, all the concepts introduced above are a natural generalization of the theory and the results developed  in Sections \ref{sec:indefinite}, \ref{sec:Q-sym} and \ref{sec:C-symm}. More specifically,  the  whole content of this Section reduces to the content of the previous Sections when the base space of the bundle reduces to a point $X=\{\ast\}$. However, the bundle picture provides more freedom for the construction of \virg{interesting} symmetries. Let us recall that a \emph{section} of a Hilbert bundle $\pi:\bb{E}\to X$ is a map $s:X\to \bb{E}$ such that $(\pi\circ s)(x)=x$ for all $x\in X$. The space of sections is usually denoted with $\Gamma(\bb{E})$. With the help of the measure structure assumed in Assumption \ref{Ass:reg_X}  one can define a Hilbert structure on $\Gamma(\bb{E})$ by means of the scalar product
\begin{equation}
\label{eq:bund_hilb_struct}
\langle s_1, s_2\rangle_{\bb{E}}\;:=\;\int_X\dd\mu(x)\; \langle s_1(x), s_2(x)\rangle_x\;.
\end{equation}
The completion of $\Gamma(\bb{E})$ with respect the topology induced by the Hilbert structure \eqref{eq:bund_hilb_struct} defines the Hilbert space of the $L^2$-sections denoted with $\rr{H}_{\bb{E}}:=L^2(\bb{E},\mu)$. Let $\bb{B}(\rr{H}_{\bb{E}})$ be the space of the linear or anti-linear bounded operators acting on $\rr{H}_{\bb{E}}$  and ${\tt U}(\rr{H}_{\bb{E}})$  and  ${\tt AU}(\rr{H}_{\bb{E}})$ the subsets of the unitary and anti-unitary operators. 
Clearly the fiber-preserving endomorphisms  ${\rm End}(\bb{E})$ can be identified with operators in  $\bb{B}(\rr{H}_{\bb{E}})$. Similarly, one has (up to an identification) the inclusions  ${\tt U}(\bb{E})\subset {\tt U}(\rr{H}_{\bb{E}})$ and ${\tt AU}(\bb{E})\subset {\tt AU}(\rr{H}_{\bb{E}})$. However not all the operator in $\bb{B}(\rr{H}_{\bb{E}})$ come from fiber-preserving endomorphisms. An interesting class of operator is generated by the \emph{covariant endomorphisms} of $\bb{E}$.
\begin{definition}[Covariant endomorphism]
\label{def:cov_endo}
Let $\pi:\bb{E}\to X$ be a Hilbert bundle. A \emph{covariant endomorphism} of $\bb{E}$  is a pair $(A,g_A)$ with $A:\bb{E}\to \bb{E}$ is a continuous map on the total space, $g_A:X\to X$ is a homeomorphism on the base space, such
that the following diagram commutes
$$
\begin{diagram}
\bb{E}&\rTo^{A}&  \bb{E}\\
\dTo^{\pi}&&\dTo_{\pi}\\
X&\rTo_{g_A}&\;\,    X\;.\\
\end{diagram}
$$
The covariant endomorphism $(A,g_A)$ is linear if $A(\lambda p)=\lambda A(p)$ for all $\lambda\in \C$ and $p\in \bb{E}$. It is called anti-linear when $A(\lambda p)=\overline{\lambda} A(p)$.
The space of linear or anti-linear covariant endomorphisms will be denoted with ${\rm Cov}(\bb{E})$
while ${\rm Cov}_{\rm lin}(\bb{E})$ will be used for the subspace of the linear ones.
\end{definition}
\medskip

\noindent
Clearly from Definition \ref{def:cov_endo} one deduces the following inclusions ${\rm End}(\bb{E})\subset {\rm Cov}(\bb{E})$ and ${\rm Cov}_{\rm lin}(\bb{E})\subset {\rm Cov}_{\rm lin}(\bb{E})$ obtained when $g_A={\rm Id}_X$. An element $(A,g_A)\in {\rm Cov}(\bb{E})$ induces a map $\hat{A}:\Gamma(\bb{E})\to \Gamma(\bb{E})$ defined by
$$\hat{A}s\; :=\;A\circ s\circ g_A^{-1}\;.$$
The map $\hat{A}$ is linear or anti-linear according to the nature of $A$ and  by density it defines an element $\hat{A}\in \bb{B}(\rr{H}_{\bb{E}})$. The covariant endomorphism $(A,g_A)$ is called unitary (resp. anti-unitary) if  the associated operator $\hat{A}$ is unitary (resp. anti-unitary). It is an (even or odd) involution if $\hat{A}^2=\varepsilon \n{1}$ with $\varepsilon=\pm 1$. In this case it is necessary that $g_A^{2}={\rm Id}_X$, namely   an \emph{involution} of the base space.

\medskip

In the case of a  Krein bundle $(\bb{E},\eta)$ the map $\eta\in{\rm End}_{\rm lin}(\bb{E})$ induces a Krein structure on the Hilbert space $\rr{H}_{\bb{E}}$ by the indefinite inner product 
\begin{equation}
\langle\langle s_1, s_2\rangle\rangle_{\bb{E},\eta}\;:=\;\int_X\dd\mu(x)\; \langle s_1(x),\eta_x s_2(x)\rangle_x\;.
\end{equation}
This allows to define the full group of the (non necessarily fiber-preserving) $\eta$-quantum symmetries 
$$
 {\tt QS}_\eta(\rr{H}_{\bb{E}})\;:=\;{\tt U}_\eta(\rr{H}_{\bb{E}})\;\sqcup\; {\tt AU}_\eta(\rr{H}_{\bb{E}})\sqcup\; {\tt PU}_\eta(\rr{H}_{\bb{E}})\sqcup\; {\tt PAU}_\eta(\rr{H}_{\bb{E}})\;\subset\; \bb{B}(\rr{H}_{\bb{E}})\;.
 $$
 The main goal of this work is to provide a topological classification of the space ${\tt H}_{\eta,\Xi}(\bb{E})$ of dynamically stable (fiber-preserving) operators under the possible action of compatible $\eta$-quantum symmetries. However the group ${\tt QS}_\eta(\rr{H}_{\bb{E}})$ is too vast and there are physical reasons to restrict the analysis only to \emph{covariant} $\eta$-quantum symmetries 
  ${\tt QS}_\eta(\rr{H}_{\bb{E}})\cap {\rm Cov}(\bb{E})$. 
 Clearly ${\tt QS}_\eta(\bb{E})\subset {\tt QS}_\eta(\rr{H}_{\bb{E}})\cap {\rm Cov}(\bb{E})$ and in this sense we are generalizing the framework described before Observation A allowing also the action of covariant  $\eta$-quantum symmetries in addition to the fiber-preserving symmetry. In view of the equivalence established by  Theorem \ref{theo_fund_rep_res} we can generalize Observation A as follows:

\smallskip

\noindent 
{\bf Observation B.}
\emph{Once the problem of the classification of the $\s{C}$-symmetric 
$\eta$-self-adjoint endomorphisms is reduced  to the problem of the classification of ${\tt H}_{\eta,0}(\bb{E})$ it turns out that the only relevant   $\eta$-quantum symmetries are induced by 
element of ${\tt QS}_\eta(\rr{H}_{\bb{E}})\cap {\rm Cov}(\bb{E})$
 which \emph{commute} or \emph{anti-commute} with the metric $\eta$.}

\medskip

\begin{remark}[Motivations for the bundle point of view: Topological quantum systems]\label{rk:motiv_hilb_bund}
{\upshape 
Hilbert bundles arise in condensed matter physics when one considers electronic systems in a periodic crystal background. 
In this case the invariance under translations allows  to use the \emph{Bloch-Floquet transform} \cite{kuchment-93} which transform the \virg{physical} Hilbert space $L^2(\R^d)$ in the $L^2$-space of an Hilbert bundle $\bb{E}$ over the torus $\n{T}^d\simeq \R^d/\Z^d$ which is usually called the Brillouin zone.  
Under some regularity condition the Hamiltonian $H$ of the system is mapped into an operator $\hat{H}$ associated to a fiber-preserving endomorphism. Translation-invariant \emph{unitary} operators   are mapped into  \emph{fiber-preserving} unitary operators while translation-invariant \emph{anti-unitary} operators are mapped into  \emph{covariant} unitary operators associated with the involution 
$\tau:\n{T}^d\to\n{T}^d$ given by $\tau(k):=-k$. For more details on this construction we refer to \cite[Section 2]{denittis-gomi-14} and reference therein. However, the bundle picture is also typical of more general physical systems, not necessarily related with periodic electronic structure, called Topological Quantum Systems (see \eg \cite[Definition 1.1]{denittis-gomi-15}).
}\hfill $\blacktriangleleft$
\end{remark}


\section{$K$-theory}
\label{sec:freed_moore_K}

The main goal of this Section is to construct a $K$-theory capable of classifying gapped $\eta$-self adjoint operators with a  $\s{C}$-symmetry. In order to do that we will use a re-formulation of the Freed-Moore $K$-theory introduced in \cite{freed-moore-13} in the spirit of the Karoubi's presentation of the $K$-theory 
\cite{karoubi-78}. In the process we will also touch the Fredholm version of the $K$-theory as initially proposed in \cite{atiyah-69}.
Section \ref{sec:karubi_vs_freed} and Section \ref{sec:Fredholm-K} provide a soft overview of the results contained in \cite{gomi-17}. The rest of the material is instead original.


\subsection{Karoubi's formulation of Freed-Moore $K$-theory}\label{sec:karubi_vs_freed}
The twisted equivariant $K$-theory introduced by Freed and Moore in \cite{freed-moore-13}
 unifies the \virg{conventional} twisted equivariant complex $K$-theory \cite{donovan-karoubi-70,rosenberg-89,freed-hopkins-teleman-11}
and  Atiyah's  $KR$-theory \cite{atiyah-66}. The relevance of this  $K$-theory is that it handles complex anti-linear symmetry.
In this Section we provide a re-formulation of the Freed-Moore $K$-theory  using the Karoubi's construction.

\medskip

{A sufficiently general}
 setup for the Freed-Moore $K$-theory, {at least for  applications to condensed matter physics},
  requires the following ingredients:
\begin{assumption}[Standard framework] \label{assumption:freed_moore_K}
Let $X$, $\n{G}$, $\varpi$, $c$, $\tau$ be as follows:
\begin{itemize}
\item[(a)]
$X$ is a compact Hausdorff space;
\vspace{1mm}

\item[(b)]
$\n{G}$ is a finite group acting on $X$ from the left;
\vspace{1mm}

\item[(c)]
$\varpi : \n{G} \to \Z_2$ and $c : \n{G} \to \Z_2$ are homomorphisms;
\vspace{1mm}

\item[(d)] $\tau\in {Z}^2_{\mathrm{group}}(\n{G}; C(X, \n{U}(1))_\varpi)$.
\end{itemize}
\end{assumption}

\medskip
\noindent
Property (b) says that $X$ is a (left) $\n{G}$-space. In (c) 
the symbol  $\Z_2 = \{ \pm 1 \}$ denotes the cyclic group of order $2$. Condition (d)  needs some elucidation. Let $C(X, \n{U}(1))$ be the group of the continuous functions on $X$ with values in the unitary group $\n{U}(1):={\tt U}(\C)$. This group admits a left-action and a right-action of $\n{G}$. The 
left-action of $g \in \n{G}$ on $f \in C(X, \n{U}(1))$ is given by $f \mapsto f^{\varpi(g)}$.
The right-action is given by the pull-back $f \mapsto g^*f$ where $g^*f(x):=f(gx)$. The group $C(X, \n{U}(1))$ endowed with this $\n{G}$-bimodule structure is denoted with $C(X, \n{U}(1))_\varpi$.
 A $2$-cocycle $\tau$ is a map $\tau : \n{G} \times \n{G}  \to C(X, \n{U}(1))_\varpi$ satisfying 
\begin{equation}\label{eq:2cocy-cond}
\tau(g_2, g_3)^{\varpi(g_1)} \; 
 =\; \tau(g_1g_2, g_3) \; \tau(g_1, g_2g_3)^{-1} \; g_3^*\tau(g_1, g_2)\;,\qquad\quad\forall\; g_1,g_2,g_3\in\n{G}
\end{equation}
and the set of the $2$-cocycles is denoted with 
${Z}^2_{\mathrm{group}}(\n{G}; C(X, \n{U}(1))_\varpi)$. Then,  condition  (d) can be rephrased by saying that
$\tau$ is a group $2$-cocycle of $\n{G}$ with values in the  $\n{G}$-bimodule $C(X, \n{U}(1))_\varpi$. Without loss of generality one can assume that the $2$-cocycle $\tau$ is \emph{normalized}, \ie $\tau(e,e)=1$ where $e$ is the unit of $\n{G}$. By combining together
2-cocycle condition \eqref{eq:2cocy-cond} and the {normalization} condition one gets
\begin{equation}\label{eq:constr_cocyc}
\tau(g,e)\;=\;1\;=\;\tau(e,g)\;,\qquad\quad g\in\n{G}\;.
\end{equation}

\medskip 

The following definition is adapted from \cite[Definition 7.23]{freed-moore-13}.

\begin{definition}[Twisted equivariant Hilbert bundle] \label{dfn:twisted_bundle_ungraded}
Let $X$, $\n{G}$, $\varpi$, $c$, $\tau$ be as in Assumption \ref{assumption:freed_moore_K}.
A \emph{$(\varpi, c,\tau)$-twisted $\n{G}$-equivariant (ungraded) Hilbert bundle} on $X$ with $C\ell^{r,s}$-action (or a \emph{twisted bundle} for short) is a Hilbert bundle $\pi:\bb{E}\to X$ equipped with the following data:
\begin{itemize}
\item[(a)] A
\emph{twisted $\n{G}$-action} provided by a  representation $\rho$ of $\n{G}$ on the total space  $\bb{E}$
 which covers the left action of $\n{G}$ on $X$ according to the following diagram
$$
\begin{diagram}
\bb{E}&\rTo^{\rho(g)}&  \bb{E}\\
\dTo^{\pi}&&\dTo_{\pi}\\
X&\rTo_{g}&  X\\
\end{diagram} \qquad\quad \forall\;g\in\n{G}
$$
 and which
 satisfies
\begin{align*}
 \rho(g)^*  &\;=\;  \rho(g)^{-1}&\quad&\text{(isometry)} \\\vspace{1mm}
 \rho(g)\; \ii\n{1} &\;=\; \varpi(g)\ii\; \rho(g)&\quad&\text{(linearity vs. anti-linearity)} \\\vspace{1mm}
\rho(g_1)\rho(g_2) &\;=\; \tau(g_1, g_2) \rho(g_1g_2)&\quad&\text{(projective representation)}
\end{align*}
for all $g,g_1,g_2\in\n{G}$.\vspace{1mm}
\item[(b)] A \emph{(unitary)  action}
$\gamma$ of the Clifford algebra
$C\ell^{r,s}$ 
 (in the sense of Definition \ref{def:cliff_act})
which meets the compatibility condition
$$
\gamma(a) \rho(g) \;=\; c(g)\; \rho(g) \gamma(a)\qquad\quad\text{(Koszul sign rule)}
$$
for all $a\in C\ell^{r,s}$ and $g\in\n{G}$.
\end{itemize}
A homomorphism $f : (\bb{E}, \rho, \gamma) \to (\bb{E}', \rho', \gamma')$ of twisted bundles  is  a complex linear map $f : \bb{E} \to \bb{E}'$ which covers the identity of $X$ and satisfies
\begin{align*}
f \circ \rho(g) &\;=\; \rho'(g) \circ f \\
f \circ \gamma(a) &\;=\; \gamma'(a) \circ f
\end{align*}
for all $a\in C\ell^{r,s}$ and $g\in\n{G}$.
When $f$ is bijective $(\bb{E}, \rho, \gamma)$ and $(\bb{E}', \rho', \gamma')$ are called \emph{isomorphic}. 
\end{definition}

\medskip

\noindent
Definition \ref{dfn:twisted_bundle_ungraded} 
works simultaneously for finite  
and infinite rank bundles. In the finite rank case the underlying  Hilbert bundle reduces 
to nothing more than 
a
finite rank Hermitian vector bundle. In both cases the structures (a) and (b) are defined in the same way. The infinite rank case will be relevant in  Section \ref{subsec:infinite_dimensional_Karoubi}.

\begin{remark}[PUA-representation]\label{rk:PUA-rep}
{\upshape 
Condition (a) in Definition \ref{dfn:twisted_bundle_ungraded} can  be summarized by saying that the Hilbert bundle $\pi:\bb{E}\to X$ is endowed with the action of a (finite) group $\n{G}$ induced by the representation $\rho:\n{G}\to {\tt QS}(\rr{H}_{\bb{E}})\cap {\rm Cov}(\bb{E})$ 
which associates to any $g\in\n{G}$ the covariant quantum symmetry $\rho(g)$. More precisely, according to Definition \ref{def:cov_endo} 
the operator $\rho(g)$ is induced by a covariant endomorphism of the bundle $\bb{E}$ and it is unitary or anti-unitary. The presence of the $2$-cocycle $\tau$ implies that the representation $\rho$ is projective. In summary $\rho$ is a covarianta \emph{PUA-representation}\footnote{PUA stands for projective, unitary, anti-unitary.} in the sense of \cite{parthasarathy-69,thiang-16}.}
\hfill $\blacktriangleleft$
\end{remark}

\medskip

The separation between  PUA-representation and Clifford action in Definition \ref{dfn:twisted_bundle_ungraded} is in some sense artificial although useful for the $K$-theory construction described  below. In fact the Clifford action can always be seen as a factor of an extended PUA-representation or, on the contrary, a PUA-representation can contain a Clifford action.
\begin{proposition}[PUA-representation vs. Clifford action]\label{prop:reduc_clifford}
There is a categorical correspondence between twisted equivariant Hilbert bundles with Clifford action and twisted equivariant Hilbert bundles without Clifford action.
\end{proposition}
\proof[{Proof} (sketch of)]
Let us start with a 
$(\varpi, c,\tau)$-twisted $\n{G}$-equivariant Hilbert bundle $\pi:\bb{E}\to X$ equipped with a
$C\ell^{r,s}$-action. Consider the extended group $\n{G}':=\n{G}\times \Z_2$ with $\Z_2=\{\pm 1\}$ and let $p:\n{G}'\to\n{G}$ be the first factor projection given by $p((g,\epsilon))=g$ for all $g\in\n{G}$ and $\epsilon\in\Z_2$. The action of $\n{G}'$ on $X$ is defined as follows: $g' x=p(g')x$
for all $g':=(g,\epsilon)\in\n{G}$ and $x\in X$. 
The homomorphisms $\varpi$ and $c$ can be extended to homomorphisms $\varpi':\n{G}'\to\Z_2$ and $c':\n{G}'\to\Z_2$ as follows: $\varpi':=\varpi\circ p$ and $c'((g,\pm1))=\pm c(g)$ for all $g\in\n{G}$. Also the 2-cocycle $\tau$ can be extended to a 2-cocycle $\tau'\in {Z}^2_{\mathrm{group}}(\n{G}'; C(X, \n{U}(1))_{\varpi'})$ by following the prescription of Table \ref{tab:rule_tau'}.
%
 \begin{table}[h]
 \centering
 \begin{tabular}{|c||c|c|c|c|c|c|c|c|c|c|}
 \hline
$\tau'$ & $(g_2,+1)$    & $(g_2,-1)$  
\\
\hline
 \hline
\rule[-3mm]{0mm}{9mm}
$(g_1,+1)$ & $\tau(g_1,g_2)$  & $\tau(g_1,g_2)$ \\
\hline
 \rule[-3mm]{0mm}{9mm}
$(g_1,-1)$ & $c(g_1)\;\tau(g_1,g_2)$  & $c(g_2)\;\tau(g_1,g_2)$ 
\\
\hline
 \end{tabular}
 \vspace{2mm}
 \caption{\footnotesize Prescription for the extended 2-cocycle $\tau'$ on $\n{G}\times \Z_2$.}
\label{tab:rule_tau'}
 \end{table}
 %
The map $\rho':\n{G}'\to {\tt QS}(\rr{H}_{\bb{E}})\cap {\rm Cov}(\bb{E})$ defined by 
$\rho'((g,+1)):=\rho(g)$ and $\rho'((g,-1)):=\rho(g)\gamma_{r+s}$ for all $g\in\n{G}$ provides a PUA-representation of 
$\n{G}'$ on the 
underlying  Hilbert bundle $\bb{E}$.
Therefore $\pi:\bb{E}\to X$ can be  equivalently seen as a $(\varpi', c',\tau')$-twisted $\n{G}'$-equivariant Hilbert bundle $\pi:\bb{E}\to X$ equipped with a
$C\ell^{r,s-1}$-action. By repeating this construction $r+s$ times one finally gets a 
twisted Hilbert bundle without an explicit Clifford action but with an extended  PUA-representation which encodes the original Clifford action.
\qed

\medskip 

As in the case of standard untwisted  bundles, one can take the direct sum of twisted bundles, and pull a $\tau$-twisted bundle on $X$ back to a $(\varphi^*\tau)$-twisted bundle on $Y$ by a $\n{G}$-equivariant map $\varphi:Y \to X$. Notice that the so-called homotopy property holds true for twisted bundles, and homotopy equivalent maps induce isomorphic twisted bundles by pullback (up to a canonical isomorphism matching the induced twists). This is because the argument to prove the homotopy property in \cite{atiyah-67} can be applied to twisted bundles, in view of the fact that for any $(\bb{E}, \rho, \gamma)$ and $(\bb{E}', \rho', \gamma')$   the complex vector bundle $\mathrm{Hom}(\bb{E}, \bb{E}') \simeq \bb{E}^* \otimes \bb{E}'$ gives rise to a usual $\n{G}$-equivariant vector bundle with commuting $C\ell^{r,s}$-action whose invariant sections are in bijective correspondence with the homomorphisms of the twisted bundles. 

\medskip

\begin{definition}[Gradation]\label{def_grad}
Let $(\bb{E}, \rho, \gamma)$ be a $(\varpi, c,\tau)$-twisted $\n{G}$-equivariant (ungraded) Hilbert bundle on $X$
 as in Definition \ref{dfn:twisted_bundle_ungraded}. A \textit{gradation} (or \textit{$\Z_2$-grading}) of $(\bb{E}, \rho, \gamma)$ is a $\Gamma\in{\rm End}_{\rm lin}(\bb{E})$
such that:
\begin{itemize}
\item[(a)] $\Gamma$ is a  self-adjoint involution, \ie $\Gamma=\Gamma^*$ and $\Gamma^2=\n{1}$;
\vspace{1mm}
\item[(b)] The relations 
\begin{align*}
\Gamma \rho(g) &\;=\; c(g)\; \rho(g) \Gamma \\
\Gamma\gamma(a) &\;=\; -\gamma(a)\Gamma
\end{align*}
hold for all $a\in C\ell^{r,s}$ and $g\in\n{G}$.
\end{itemize}
\end{definition}

\medskip

\noindent
As explained in Section \ref{sec:hilb_pict}, the condition $\Gamma\in{\rm End}_{\rm lin}(\bb{E})$ means that $\Gamma$ can be seen as a 
continuous  map that associates to each $x\in X$ a linear operator $\Gamma_x$ on the fiber $\bb{E}_x$, such that $\Gamma_x=\Gamma_x^*$ and $\Gamma_x^2=\n{1}_x$. Then a gradation is nothing but a continuous family of Hermitian operators (or Hermitian matrices) squaring to the identity. Notice that the Clifford action $\gamma$  is required to be \emph{odd} with respect to the gradation $\Gamma$.
\begin{remark}[Gapped Hamiltonians and gradations]\label{rk:grad-gap}
{\upshape 
A (non-trivial) \emph{gapped Hamiltonian} is a  self-adjoint 
element $H\in{\rm End}_{\rm lin}(\bb{E})$
which admits a continuous function $\lambda:X\to\R$ such that $H_x-\lambda(x)\;\n{1}_x$ is invertible and compact in each fiber $\s{H}_x:=\pi^{-1}(x)$ and
\begin{equation}\label{eq:gap_cond}
\min\; \sigma(H_x)\;<\;\lambda(x)\;<\;\max\; \sigma(H_x)
\end{equation}
where $\sigma(H_x)\subset\R$ is the spectrum of $H_x$. In this case the self-adjoint operator
$$
\Gamma_H|_{\s{H}_x}\;:=\;\frac{H_x-\lambda(x)\;\n{1}_x}{\big|H_x-\lambda(x)\;\n{1}_x\big|}
$$
is fiberwise well-defined, $\Gamma_H|_{\s{H}_x}\neq \n{1}_x$ in view of \eqref{eq:gap_cond} and 
$(\Gamma_H|_{\s{H}_x})^2=\n{1}_x$. The collection of the $\Gamma_H|_{\bb{E}_x}$ defines a gradation $\Gamma_H:\bb{E}\to \bb{E}$ in the sense of Definition \ref{def_grad} (a).
Moreover,
 condition (b) is satisfied if and only if
$H\rho(g)=c(g)\rho(g)H$ and $H\gamma(a)=-\gamma(a)H$ and this
  requires certain symmetries of the spectra $\sigma(H_x)$. The transformation $H\to\Gamma_H$ is induced by a homotopy as showed in Remark \ref{rk:more_structure}.
}\hfill $\blacktriangleleft$
\end{remark}

\medskip

Two gradations $\Gamma_0$ and $\Gamma_1$ on $(\bb{E}, \rho, \gamma)$ are said to be \emph{homotopic} if there is a gradation $\hat{\Gamma}$ on $p^*(\bb{E}, \rho, \gamma)$ such that $\hat{\Gamma}|_{X \times \{ j \}} = \Gamma_j$ for $j = 0, 1$.  Here 
$p : X \times [0, 1] \to X$ is the equivariant projection from the $\n{G}$-space $X \times [0, 1]$ (with trivial action on the second component) onto the $\n{G}$-space $X$. It is clear from this definition that a homotopy of gradations $\hat{\Gamma}$ amounts to a continuous  family of gradations ${\Gamma}(t)$ on $\bb{E}$ such that ${\Gamma}(0) = \Gamma_0$ and ${\Gamma}(1) = \Gamma_1$.

\medskip

Let us focus on  \textit{triples} $(\bb{E}, \Gamma_0, \Gamma_1)$ given by a  $(\varpi, c,\tau)$-twisted $\n{G}$-equivariant (finit rank) vector bundle $(\bb{E}, \rho, \gamma)$  on $X$ as in Definition \ref{dfn:twisted_bundle_ungraded} and two gradations $\Gamma_0$ and $\Gamma_1$ as in Definition \ref{def_grad}. One has the following two structure:
\begin{itemize}

\item \emph{Isomorphism:}
Two triples $(\bb{E}, \Gamma_0, \Gamma_1)$ and $(\bb{E}', \Gamma'_0, \Gamma'_1)$ are {isomorphic} if there is an isomorphism $f : \bb{E} \to \bb{E}'$ of twisted bundles such that $f \circ \Gamma_i = \Gamma'_i\circ f$ for both $i=0,1$;
\vspace{1mm}
\item \emph{Direct sum:} The triple
$$
\big(\bb{E}, \Gamma_0, \Gamma_1\big)\; \oplus\; \big(\bb{E}', \Gamma'_0, \Gamma'_1\big)\; :=\; \big(\bb{E} \oplus \bb{E}', \Gamma_0 \oplus \Gamma'_0, \Gamma_1 \oplus \Gamma'_1\big)
$$
is called the  direct sum of
 $(\bb{E}, \Gamma_0, \Gamma_1)$ and $(\bb{E}', \Gamma'_0, \Gamma'_1)$.
\end{itemize}

\noindent
Clearly, the direct sum $\bb{E} \oplus \bb{E}'$ inherits the natural twisted structure induced by the direct sum of the PUA-representations $\rho\oplus\rho'$ and by the direct sum of the Clifford actions $\gamma\oplus\gamma'$.

\begin{definition}[$K$-groups - finite rank case]\label{def_K_group}
Let $X$, $\n{G}$, $\varpi$, $c$, $\tau$ be as in Assumption \ref{assumption:freed_moore_K}. Then:

\begin{itemize}
\item
 ${}^\varpi \s{M}^{(\tau, c) + (r, s)}_{\n{G}}(X)$ is the \emph{abelian monoid} of isomorphism classes of triples $(\bb{E}, \Gamma_0, \Gamma_1)$ with addition  given by the direct sum of triples; 
\vspace{1mm}
\item
${}^\varpi \s{Z}^{(\tau, c) + (r, s)}_{\n{G}}(X)$ is the submonoid consisting of isomorphism classes of triples $(\bb{E}, \Gamma_0, \Gamma_1)$ such that $\Gamma_0$ and $\Gamma_1$ are \emph{homotopy equivalent};
\vspace{1mm}
\item
 ${}^\varpi \s{K}^{(\tau, c) + (r, s)}_\n{G}(X) := {}^\varpi \s{M}^{(\tau, c) + (r, s)}_{\n{G}}(X)/{}^\varpi \s{Z}^{(\tau, c) + (r, s)}_{\n{G}}(X)$ is  the \emph{quotient monoid}.
\end{itemize}
\end{definition}

\medskip

\noindent
The following result is crucial and  is based on \cite[Chap. III, Lemma 4.16]{karoubi-78}.

\begin{lemma}
The monoid ${}^\varpi \s{K}^{(\tau, c) + (r, s)}_\n{G}(X)$ is an abelian group in which the additive inverse of $[(\bb{E}, \Gamma_0, \Gamma_1)]$ is given by $[(\bb{E}, \Gamma_1, \Gamma_0)]$.
\end{lemma}
\begin{proof}
It suffices to prove that the gradations $\Gamma_0 \oplus \Gamma_1$ and $\Gamma_1 \oplus \Gamma_0$ on $\bb{E} \oplus \bb{E}$ are homotopy equivalent. This fact follows by observing that 
such a homotopy is realized by the  map 
\begin{align*}
[0,\pi/2]\;\ni\;\theta\;\mapsto\; {\Gamma}(\theta)
\;:=\;
T_\theta\;\left(
\begin{array}{cc}
\Gamma_0 & 0 \\
0 & \Gamma_1
\end{array}
\right)
\;T_\theta \\
\end{align*}
where the matrix $T_\theta$ is defined in
\eqref{eq:T-theta}.
\end{proof}

\medskip

{
The groups ${}^\varpi \s{K}_\n{G}^{(\tau, c) + (r, s)}(X)$ have the following periodicity \cite{gomi-17}.
}

\begin{proposition}[Periodicity] \label{prop:freed_moore_K_property}
Let $X$, $\n{G}$, $\varpi$, $c$, $\tau$ be as in Assumption \ref{assumption:freed_moore_K} and ${}^\varpi \s{K}^{(\tau, c) + (r, s)}_\n{G}(X)$ the abelian group introduced in Definition \ref{def_K_group}. Then  ${}^\varpi \s{K}_\n{G}^{(\tau, c) + (r, s)}(X)$ is subject to the following periodicities:
\begin{align*}
{}^\varpi \s{K}_\n{G}^{(\tau, c) + (r, s)}(X)
&\simeq\;
{}^\varpi \s{K}_\n{G}^{(\tau, c) + (r+1, s+1)}(X) \\
&\simeq\;
{}^\varpi \s{K}_\n{G}^{(\tau, c) + (r + 8, s)}(X)\\
&\simeq\;
{}^\varpi \s{K}_\n{G}^{(\tau, c) + (r, s + 8)}(X)\;.
\end{align*}
In the case  $\varpi\equiv+1$ is trivial, the group $\s{K}_\n{G}^{(\tau, c) + (r, s)}(X) := {}^{+1} \s{K}_\n{G}^{(\tau, c) + (r, s)}(X)$ is subject to the extra reduced periodicities:
\begin{align*}
\s{K}_\n{G}^{(\tau, c) + (r, s)}(X)
&\simeq\;
\s{K}_\n{G}^{(\tau, c) + (r + 2, s)}(X)\\
&\simeq\;
\s{K}_\n{G}^{(\tau, c) + (r, s + 2)}(X)\;.
\end{align*} 
\end{proposition}
\medskip

\noindent
{Some details about the proof of Proposition \ref{prop:freed_moore_K_property} will be provided in Remark \ref{rk:twist_period}.}

\medskip

\begin{remark}[Generalized cohomology]{\upshape 
 Proposition \ref{prop:freed_moore_K_property} shows that the groups ${}^\varpi \s{K}_\n{G}^{(\tau, c) + (r, s )}(X)$ depend only on the difference $s-r$ mod. 8 and this justify the following notation
$$
{}^\varpi \s{K}_\n{G}^{(\tau, c) + n}(X)\;:=\;  {}^\varpi \s{K}_\n{G}^{(\tau, c) + (0, n)}(X)\;\simeq\;{}^\varpi \s{K}_\n{G}^{(\tau, c) + (r, s )}(X)\;,\qquad\quad n=s-r\quad\text{mod.}\quad 8
$$
{The groups ${}^\varpi \s{K}_\n{G}^{(\tau, c) + n}(X)$ form a $\n{G}$-equivariant generalized cohomology theory, so that the homotopy axiom, the excision axiom, etc.\ hold true (see \eg \cite{freed-hopkins-teleman-11,gomi-17}).}
}\hfill $\blacktriangleleft$
\end{remark}


\subsection{Infinite-dimensional Karoubi's formulation}
\label{subsec:infinite_dimensional_Karoubi}
Definition \ref{def_K_group} is modeled on finit rank twisted  vector bundles.
In this Section we will focus on the  case
of twisted  Hilbert bundles $\pi:\bb{E}\to X$ of infinite rank. This means that each fiber of $\bb{E}$ is isomorphic to a separable Hilbert space $\s{H}$ of infinite dimension.
The infinite dimensionality allows the existence of a \emph{universal}  twisted bundle of infinite rank \cite[Section 2.6.]{gomi-17}:

\begin{lemma}[Universal  twisted bundle] \label{lem:locally_universal_bundle}
Let $X$, $\n{G}$, $\varpi$, $c$, $\tau$ be as in Assumption \ref{assumption:freed_moore_K}. There exists a $(\varpi, c,\tau)$-twisted $\n{G}$-equivariant Hilbert bundle $\bb{E}_{\mathrm{univ}}$ on $X$ with $C\ell^{r,s}$-action and a gradation $\Gamma_{\mathrm{univ}}$ which are \emph{locally universal}  in the following sense: For any $\n{G}$-invariant closed subspace $Y \subseteq X$ and any $(\varpi, c, \tau|_Y)$-twisted $\n{G}$-equivariant Hilbert bundle $\pi:\bb{E}\to Y$ with $C\ell^{r,s}$-action and  gradation $\Gamma$ on it, there is an embedding $\bb{E} \to \bb{E}_{\mathrm{univ}}|_Y$ of twisted bundles which preserves the gradations.
\end{lemma}

\medskip

\noindent 
{The universal bundle $\bb{E}_{\mathrm{univ}}$ in Lemma \ref{lem:locally_universal_bundle} has infinite rank, and also the  eigenspaces of $\Gamma_{\mathrm{univ}}$ associated to the eigenvalues $\pm 1$ are infinite dimensional. This allows the   twisted bundle $\pi:\bb{E}\to Y$ to have infinite rank.  
 It is known \cite[Section 2.6.]{gomi-17} that the space of embeddings $\bb{E} \to \bb{E}_{\mathrm{univ}}|_Y$ is contractible with respect to an appropriate topology, and the pair $(\bb{E}_{\mathrm{univ}},\Gamma_{\mathrm{univ}})$ that implements the local universality  in Lemma \ref{lem:locally_universal_bundle} is essentially unique.}

\begin{definition}\label{def:grad_cop_pert}
Let $X$, $\n{G}$, $\varpi$, $c$, $\tau$ be as in Assumption \ref{assumption:freed_moore_K}, and $(\bb{E}_{\mathrm{univ}},\Gamma_{\mathrm{univ}})$ the locally universal pair of Lemma \ref{lem:locally_universal_bundle}. {Let $\bb{K}_x:=\bb{K}(\bb{E}_{\mathrm{univ},x})$ be the algebra of compact operators on the fiber $\bb{E}_{\mathrm{univ},x}$. We denote with
$$
\bb{G}(\bb{E}_{\mathrm{univ}})
\;:=\; \left\{
\Gamma : \bb{E}_{\mathrm{univ}} \to \bb{E}_{\mathrm{univ}}\ |\
\Gamma\; \text{\upshape is gradation and}\ (\Gamma - \Gamma_{\mathrm{univ}})_x\in \bb{K}_x\;\;\;\; \forall x\in X
\right\}
$$
the set of gradations $\Gamma$ on $\bb{E}_{\mathrm{univ}}$ which differ from $\Gamma_{\mathrm{univ}}$ by a compact operator on each fiber.}
\end{definition}

\medskip

\noindent
{In view of Definition \ref{def:grad_cop_pert} any
$\Gamma\in\bb{G}(\bb{E}_{\mathrm{univ}})$ is a \emph{compact perturbation} of the universal gradation $\Gamma_{\mathrm{univ}}$. This implies that also the  eigenspaces of $\Gamma$ associated to the eigenvalues $\pm 1$ are infinite dimensional.}

\medskip

A homotopy on $\bb{G}(\bb{E}_{\mathrm{univ}})$ is a family of gradations which is continuous with respect to the operator norm on the fibers of $\bb{E}_{\mathrm{univ}}$. {Such homotopies introduce an equivalence relation in $\bb{G}(\bb{E}_{\mathrm{univ}})$ denoted with $\sim$ . The following result is proved in \cite[Lemma 4.16]{gomi-17}.}

\begin{proposition}\label{prop:finite_vs_infinite_Karoubi}
Let $X$, $\n{G}$, $\varpi$, $c$, $\tau$ be as in Assumption \ref{assumption:freed_moore_K}. There is a natural bijection
$$
{}^\varpi \s{K}_\n{G}^{(\tau, c) + (r, s)}(X)\;
\simeq\; \bb{G}(\bb{E}_{\mathrm{univ}})/\sim.
$$
\end{proposition}

\medskip

\noindent
{
The bijection in Proposition \ref{prop:finite_vs_infinite_Karoubi} is constructed as follows. Let $(\bb{E}, \Gamma_0, \Gamma_1)$ be a triple consisting of a finite rank twisted vector bundle $\bb{E}$ on $X$ and its gradations $\Gamma_0$ and $\Gamma_1$. The nature of the locally universal structure described in Lemma \ref{lem:locally_universal_bundle}
 allows us to embed the pair $(\bb{E},\Gamma_1)$  into $(\bb{E}_{\mathrm{univ}},\Gamma_{\mathrm{univ}})$. This embedding induces an orthogonal decomposition $\bb{E}_{\mathrm{univ}} \simeq \bb{E} \oplus \bb{E}^\perp$ of twisted bundles which respects the gradations in the sense that $\Gamma_{\mathrm{univ}} = \Gamma_1 \oplus \Gamma^\perp$ where $\Gamma^\perp$ is a gradation on the complement  $\bb{E}^\perp$. Now, from the second gradation $\Gamma_0$ on $\bb{E}$ one has a second gradation $\Gamma_0 \oplus \Gamma^\perp$ on $\bb{E}_{\mathrm{univ}}$, which differs from $\Gamma_{\mathrm{univ}}$ by $(\Gamma_0-\Gamma_1)\oplus0$ which is a compact operator on each fiber of $\bb{E}_{\mathrm{univ}}$. In view of the definition of ${}^\varpi \s{K}^{(\tau, c) + (r, s)}_\n{G}(X) = {}^\varpi \s{M}^{(\tau, c) + (r, s)}_{\n{G}}(X)/{}^\varpi \s{Z}^{(\tau, c) + (r, s)}_{\n{G}}(X)$ and the property of $\bb{E}_{\mathrm{univ}}$, the assignment $(\bb{E}, \Gamma_0, \Gamma_1) \mapsto \Gamma_0 \oplus \Gamma^\perp$ induces a well-defined natural map
$$
\jmath : \ 
{}^\varpi \s{K}^{(\tau, c) + (r, s)}_\n{G}(X)\;
\longrightarrow\; \bb{G}(\bb{E}_{\mathrm{univ}})/\sim
$$
which indeed realizes the bijection in Proposition \ref{prop:finite_vs_infinite_Karoubi}. The proof of the bijectivity of $\jmath$ requires the construction of the inverse to $\jmath$ by a finite dimensional ``approximation'' of the triple $(\bb{E}_{\mathrm{univ}}, \Gamma, \Gamma_{\mathrm{univ}})$ formed by an element $\Gamma \in \bb{G}(\bb{E}_{\mathrm{univ}})$. We refer to \cite[Section 4.3]{gomi-17} for the detail.
}

\medskip

In view of Proposition \ref{prop:finite_vs_infinite_Karoubi} we can generalize Definition \ref{def_K_group} for the infinite rank case:
\begin{definition}[$K$-groups - infinite rank case]\label{def_K_group_inf}
Let $X$, $\n{G}$, $\varpi$, $c$, $\tau$ be as in Assumption \ref{assumption:freed_moore_K}. Then:

\begin{itemize}
\item
 ${}^\varpi \s{M}^{(\tau, c) + (r, s)}_{\n{G}}(X)$ is the \emph{abelian monoid} of isomorphism classes of triples $(\bb{E}, \Gamma_0, \Gamma_1)$ where $\bb{E}$ is a 
 $(\varpi, c,\tau)$-twisted $\n{G}$-equivariant Hilbert bundle over $X$ and the two gradations verify the compactness condition
$$
(\Gamma_0 - \Gamma_{1})_x\;\in\; \bb{K}_x \qquad\quad \forall\; x\in X\;.
$$ 
 Isomorphisms and  addition  are defined as in Section \ref{sec:karubi_vs_freed}; 
\vspace{1mm}
\item
${}^\varpi \s{Z}^{(\tau, c) + (r, s)}_{\n{G}}(X)$ and ${}^\varpi \s{K}^{(\tau, c) + (r, s)}_\n{G}(X)$ are defined as in Definition
\ref{def_K_group}.
\end{itemize}
\end{definition}

\medskip

\noindent
The only difference between Definition \ref{def_K_group} and Definition \ref{def_K_group_inf} is the addition of the compactness condition. Clearly, whenever $\bb{E}$ is a finite rank vector bundle Definition \ref{def_K_group_inf} reduces to Definition \ref{def_K_group}.


\subsection{Fredholm operator formulation of Freed-Moore $K$-theory}
\label{sec:Fredholm-K}
Topological $K$-theory has various formulations. The original formulation of the Freed-Moore $K$-theory  is based on finite rank twisted bundles \cite[Definition 7.33]{freed-moore-13}. However,
 a possible formulation based on Fredholm operators is also mentioned \cite[Remark 7.35]{freed-moore-13}. In this formulation the basic objects are Fredholm families.
\begin{definition}[Self-adjoint Fredholm family]
\label{def:self_adjoint_fredholm_family}
Let $(\bb{E}, \rho, \gamma)$ be a $(\varpi, c,\tau)$-twisted $\n{G}$-equivariant  Hilbert bundle on $X$ with $C\ell^{r,s}$-action, and $\Gamma$ a gradation on $(\bb{E}, \rho, \gamma)$. A self-adjoint Fredholm family  
is a self-adjoint element $A\in{\rm End}_{\rm lin}(\bb{E})$. This means that $A$ restricts to a linear bounded operator  $A_x\:=A|_{\s{H}_x}$   on each fiber Hilbert space $\s{H}_x:=\pi^{-1}(x)$ and
 \begin{itemize}
\item[(a)] $A_x$ is
self-adjoint, $A_x= A_x^*$.
\vspace{1mm}
\end{itemize}
In addition one requires that:
 \begin{itemize}
\item[(b)] $A^2_x-\n{1}_x$ is compact;
\vspace{1mm}
\item[(c)] $\sigma(A_x)\subset[-1,+1]$;
\end{itemize}
Finally the following 
 compatibility relations are required:
 \begin{itemize}
\item[(d)] For all $a\in C\ell^{r,s}$ and $g\in\n{G}$ it holds that
\begin{align*}
A \rho(g) &\;=\; c(g)\; \rho(g) A \\
A \gamma(a) &\;=\; -\gamma(a)A
\\
A \Gamma &\;=\; -\Gamma A\;.
\end{align*}
\end{itemize}
\end{definition}
\medskip

\noindent
Property  (a) in Definition \ref{def:self_adjoint_fredholm_family} says that $A_x$ is a (bounded) self-adjoint operator for each $x\in X$ and this justifies the reality of the spectrum $\sigma(A_x)\subset\R$.
Property (b) says that $A_x$ is an involution (hence an invertible) modulo compact operators.
In view of the Atkinson characterization \cite{atkinson-51} this implies that $A_x$ is a self-adjoint Fredholm operator for each $x\in X$. The following result shows the strong link between self-adjoint Fredholm families and gradations.
\begin{lemma}\label{lemma:family_to_grad}
Let $A$ be a self-adjoint Fredholm family of the $(\varpi, c,\tau)$-twisted $\n{G}$-equivariant  Hilbert bundle $(\bb{E}, \rho, \gamma)$ on $X$, and $\Gamma$ a gradation on $(\bb{E}, \rho, \gamma)$. Then the element
$$
\vartheta(A)\;:=\;-\expo{\pi A\Gamma}\Gamma
$$
provides a second gradation of $(\bb{E}, \rho, \gamma)$ and the difference $\Gamma- \vartheta(A)$ is compact in each fiber.
\end{lemma}
\proof
Since $A$ and $\Gamma$ are by definition elements in ${\rm End}_{\rm lin}(\bb{E})$ then also the product $A\Gamma$ is a fiber preserving linear endomorphism of $\bb{E}$. Fiberwise one has that $\ii(A\Gamma)_x=\ii A_x\Gamma_x$ is a self-adjoint operator in view of the anti-commutativity of $A$ and $\Gamma$. Therefore
$\expo{\pi (A\Gamma)_x}=\expo{-\ii(\ii\pi A_x\Gamma_x)}$ is a well defined unitary operator in each fiber. The element $\vartheta(A)\in {\rm End}_{\rm lin}(\bb{E})$
is then defined fiberwise by $\vartheta(A)_x:=-\expo{\pi (A\Gamma)_x}\Gamma_x$. 
Direct computations show that $\vartheta(A)$ is a self-adjoint involution. Indeed
$$
\vartheta(A)^*\;=\;-\Gamma^*\expo{\pi (A\Gamma)^*}\;=\;-\Gamma\expo{-\pi A\Gamma}\;=\;-\expo{\pi A\Gamma}\Gamma\;=\;\vartheta(A)
$$
and 
$$
\vartheta(A)^2\;=\;\expo{\pi A\Gamma}\Gamma\expo{\pi A\Gamma}\Gamma\;=\;\Gamma\expo{-\pi A\Gamma}\expo{\pi A\Gamma}\Gamma\;=\;\n{1}\;
$$
where the computations can be  understood fiberwise, if necessary. Since $ A\Gamma$ commutes with $\gamma(a)$  and $\rho(g)$ for all
$a\in C\ell^{r,s}$ and $g\in\n{G}$  it follows that the commutation relations between $\vartheta(A)$ and $\gamma(a)$  or $\rho(g)$ are the same commutation relations of $\Gamma$. This proves that $\vartheta(A)$ is a gradation according to Definition \ref{def_grad}.
To finish the proof one has to prove that 
$\Gamma- \vartheta(A)=(\n{1}+\expo{\pi A\Gamma})\Gamma$ is compact in each fiber. For that it is enough to show that $\n{1}_x+\expo{\pi A_x\Gamma_x}$ is a compact operator for each $x\in X$. Observe that the compactness of 
$\n{1}_x-A_x^2$ implies that the spectrum of $A_x^2$ is pure point with  $1$ as the only possible accumulation point. Moreover, the only eigenvalue with possible infinite multiplicity is indeed $1$. From the spectral mapping theorem one concludes that $\sigma(A_x)$ is pure point and the only eigenvalues with  possible infinite multiplicity are $\pm1$. Let us anticipate the following result:
\begin{equation}\label{eq:spect_relat}
\lambda\;\in\;\sigma(A_x)\;\quad\; \Leftrightarrow\;\quad\;\ii\lambda\;\in\;\sigma(A_x\Gamma_x)\;.
\end{equation}
Condition \eqref{eq:spect_relat} implies that 
also $\sigma(A_x\Gamma_x)$ is pure point with $\pm\ii$ as the only possible eigenvalues of infinite multiplicity. As a consequence 
the spectrum of $\n{1}_x+\expo{\pi A_x\Gamma_x}$
is pure point with $0$ as the only possible eigenvalue of infinite multiplicity. This is equivalent to the compactness of  $\n{1}_x+\expo{\pi A_x\Gamma_x}$. Finally, the proof of 
\eqref{eq:spect_relat}. Let $\Pi_{x}^\pm$ be the spectral projections of $\Gamma_x$ related to the eigenvalues $\pm1$ respectively.  Then $\Gamma_x=\Pi_{x}^+-\Pi_{x}^-$. As a consequence of the anti-commutation between $A$ and $\Gamma$ one has that $A_x=A_x^{-,+}+A_x^{+,-}$ where
 $A_x^{\pm,\mp}:=\Pi_{x}^\pm A_x\Pi_{x}^\mp$.
Said differently $A_x$ is completely off-diagonal with respect to the grading induced by  
 $\Gamma_x$. As a consequence $A_x\Gamma_x=A_x^{-,+}-A_x^{+,-}$. Let $\lambda\in{\rm Res}(A_x)$ be a point of the resolvent set of $A_x$. Then the resolvent $R_x(\lambda):=(A_x-\lambda\n{1}_x)^{-1}$ exists
and can be decomposed along $\Pi_{x}^\pm$ as follows: $R_x(\lambda)=R_x^{+,+}+R_x^{-,-}+R_x^{+,-}+R_x^{-,+}$. The conditions $(A_x-\lambda\n{1}_x)R_x(\lambda)=\n{1}_x=R_x(\lambda)(A_x-\lambda\n{1}_x)$ imply
\begin{equation}\label{consR_A}
\begin{aligned}
A_x^{-,+}R_x^{+,+}\;&=\;\lambda R_x^{-,+}\;=\;R_x^{-,-}A_x^{-,+}\\
A_x^{+,-}R_x^{-,-}\;&=\;\lambda R_x^{+,-}\;=\;R_x^{+,+}A_x^{+,-}\\
A_x^{-,+}R_x^{+,-}\;&=\;\lambda R_x^{-,-}\;+\;\Pi^-_x\;=\;R_x^{-,+}A_x^{+,-}\\
A_x^{+,-}R_x^{-,+}\;&=\;\lambda R_x^{+,+}\;+\;\Pi^+_x\;=\;R_x^{+,-}A_x^{-,+}\;.
\end{aligned}
\end{equation}
Consider now the operator $\tilde{R}_x(\lambda):=-\ii(R_x^{+,+}+R_x^{-,-})+R_x^{+,-}-R_x^{-,+}$. By using the relations \eqref{consR_A} a direct computation shows that $(A_x\Gamma_x-\ii\lambda\n{1}_x)\tilde{R}_x(\lambda)=\n{1}_x=\tilde{R}_x(\lambda)(A_x\Gamma_x-\ii\lambda\n{1}_x)$, hence $\ii\lambda\in {\rm Res}(A_x\Gamma_x)$. On the other hand if $\ii\lambda\in {\rm Res}(A_x\Gamma_x)$ then $\lambda\in {\rm Res}(-\ii A_x\Gamma_x)$ and in view of the same argument above this implies $\ii\lambda\in {\rm Res}(-\ii A_x)$ or equivalently $\lambda\in {\rm Res}(- A_x)={\rm Res}( \Gamma_xA_x\Gamma_x)={\rm Res}(A_x)$. In summary we proved that  $\lambda\in {\rm Res}(A_x)$
if and only if $\ii\lambda\in {\rm Res}(A_x\Gamma_x)$ and therefore \eqref{eq:spect_relat} follows.
\qed

\medskip
{As it has been pointed out in Section \ref{subsec:infinite_dimensional_Karoubi}, any finite rank twisted bundle} $(\bb{E}, \rho, \gamma)$ with a gradation $\Gamma$ can be embedded into the locally universal bundle $(\bb{E}_{\mathrm{univ}}, \rho_{\mathrm{univ}}, \gamma_{\mathrm{univ}})$ with its universal gradation $\Gamma_{\mathrm{univ}}$. A Fredholm family $A$ on $\bb{E}$ can be extended to one on $\bb{E}_{\mathrm{univ}}$ in an essentially unique way. Thus, it turns out that it is sufficient to consider Fredholm families on $\bb{E}_{\mathrm{univ}}$. 
{Let $\bb{F}(\bb{E}_{\mathrm{univ}})$ be the  space of Fredholm families on $\bb{E}_{\mathrm{univ}}$.
A notion of a homotopy between Fredholm families can be defined in a natural way \cite[Section 3.1]{gomi-17} and the space $\bb{F}(\bb{E}_{\mathrm{univ}})$ turns out to be endowed with an induced
equivalence relation (still) denoted with $\sim$.}
\begin{proposition}
\label{prop:Karoubi_vs_Fredholm}
There is a natural bijection between $\bb{G}(\bb{E}_{\mathrm{univ}})/\sim$ and $\bb{F}(\bb{E}_{\mathrm{univ}})/\sim$ {which, in view of Proposition \ref{prop:finite_vs_infinite_Karoubi}, implies
$$
{}^\varpi \s{K}_\n{G}^{(\tau, c) + (r, s)}(X)\;
\simeq\; \bb{F}(\bb{E}_{\mathrm{univ}})/\sim.
$$}
\end{proposition}

\medskip

\noindent
{Proposition \ref{prop:Karoubi_vs_Fredholm} is  essentially proved in \cite[Section 4.2 \& Section 4.2]{gomi-17} and the proof  boils down to the approach developed in  \cite{atiyah-singer-69}.
The basic result is provided by  Lemma  \ref{lemma:family_to_grad} which assures that 
 for a given self-adjoint Fredholm family $A \in \bb{F}(\bb{E}_{\mathrm{univ}})$ the operator $\vartheta(A):= \expo{\pi A\Gamma_{\mathrm{univ}}}\Gamma_{\mathrm{univ}}$ is a gradation such that $\vartheta(A)-\Gamma_{\mathrm{univ}}$ is fiberwise compact.
Then $\vartheta$ defines a map $\vartheta:\bb{F}(\bb{E}_{\mathrm{univ}}) \to \bb{G}(\bb{E}_{\mathrm{univ}})$ which descends to a map on the set of homotopy classes. By means a suitable generalization of the argument used in \cite{atiyah-singer-69}, and the  
previous isomorphisms proved in Proposition \ref{prop:finite_vs_infinite_Karoubi},
one can  show that the induced map $\vartheta$ realizes indeed the bijection in Proposition \ref{prop:Karoubi_vs_Fredholm}.}

 \medskip

\begin{remark}[Recalibration of the twist]\label{rk:twist_period}
{\upshape 
{ In \cite[Definition 3.1]{gomi-17} the notion of the Fredholm family is given in terms of
 skew-adjoint operators  and the self-adjoint version given in Definition \ref{def:self_adjoint_fredholm_family} is mentioned in \cite[Remark 4.12]{gomi-17}.
  These Fredholm famlies  lead to the $K$-theories ${}^\varpi \n{K}_{\n{G}}^{(\tau, c) +(r,s)}(X)$ and ${}^\varpi \acute{\n{K}}_{\n{G}}^{(\tau, c) +(r,s)}(X)$, respectively. The former is the $K$-theory mainly studied in \cite{gomi-17} while the latter is nothing but $\bb{F}(\bb{E}_{\mathrm{univ}})/\sim$ in Proposition \ref{prop:Karoubi_vs_Fredholm}. Usually, one has a one-to-one correspondence between self-adjoint operators  and skew-adjoint operators   given by the multiplication with the imaginary unit $\ii$. 
However, in the presence of a non-trivial $\varpi$, this correspondence violates the compatibility relations with the $\rho(g)$ (\cf Definition \ref{dfn:twisted_bundle_ungraded}), and consequently
 does not provide an isomorphism between ${}^\varpi \acute{\n{K}}_{\n{G}}^{(\tau, c) +(r,s)}(X)$ and ${}^\varpi \n{K}_{\n{G}}^{(\tau, c) +(r,s)}(X)$. 
 Another correspondence between self-adjoint operators  and skew-adjoint operators is provided by
  $A \mapsto A\Gamma$  whenever $\Gamma$ is a gradation such that $A\Gamma = -\Gamma A$. 
 The correspondence $A \mapsto A\Gamma$ together with a simultaneous modification of the symmetries $\rho\mapsto  \hat{\rho}$ and $\gamma\mapsto  \hat{\gamma}$ where
 $$
 \hat{\gamma}(a)\;:=\;{\gamma}(a)\Gamma\;,\qquad\quad \hat{\rho}(g)\;:=\;
\left\{
\begin{aligned}
&\rho(g), &\qquad&\text{if}\quad c(g) = +1 \\
&\rho(g)\Gamma  &\qquad&\text{if}\quad c(g)- 1
\end{aligned}\;,
\right.
 $$ 
leads to the isomorphism ${}^\varpi \acute{\n{K}}_{\n{G}}^{(\tau, c) +(r,s)}(X) \to {}^\varpi \n{K}_{\n{G}}^{(\hat{\tau}, c) +(s,r)}(X)$, where the new cocycle $\hat{\tau}$ is related to $\tau$ by
$$
\hat{\tau}(g_1, g_2)
\;:=\;
\left\{
\begin{aligned}
&+\tau(g_1, g_2), &\qquad&\text{if}\quad c(g_1) = 1\; \text{or}\; c(g_2) = 1 \\
&- \tau(g_1, g_2)  &\qquad&\text{if}\quad c(g_1) = c(g_2)=- 1
\end{aligned}
\right.
$$
The cocycles $\tau$ and $\hat{\tau}$ are equivalent if $\varpi$ or $c$ are trivial, so that their difference is essential only if both $\varpi$ and $c$ are non-trivial. Notice also that $\hat{\hat{\tau}} = \tau$. 
The isomorphisms
\begin{equation}\label{eq:sio-K_renew}
{}^\varpi  \s{K}_{\n{G}}^{({\tau}, c) + (r,s)}(X)\;\simeq\; {}^\varpi \acute{\n{K}}_{\n{G}}^{(\tau, c) +(r,s)}(X)\;\simeq\;{}^\varpi \n{K}_{\n{G}}^{(\hat{\tau}, c) + (s,r)}(X)
\end{equation}
combined with the periodicity in \cite[Lemma 3.5]{gomi-17} provides a proof of Proposition \ref{prop:freed_moore_K_property}.
In particular one has that
\begin{equation}\label{eq:for_computation}
{}^\varpi  \s{K}_{\n{G}}^{({\tau}, c) + n}(X)\;\simeq\; {}^\varpi \n{K}_{\n{G}}^{(\hat{\tau}, c) +n}(X)
\end{equation}
where $n:=s-r$ mod. 8. The latter isomorphism will be used in the computation of Section \ref{sec:explicit_comp}.
}
}\hfill $\blacktriangleleft$
\end{remark}

\medskip

\begin{remark}\label{rk:twist_period_2}
{\upshape 
Let ${}^\varpi{\rm Vec}^{(\tau, c) + (r, s)}_{\n{G}}(X)$ be the monoid of isomorphism classes of finite rank $(\varpi, c, \tau)$-twisted $\n{G}$-equivariant vector bundles on $X$ with $C\ell^{r,s}$-action and gradations. There is a forgetful functor which induces a homomorphism $f:{}^\varpi\mathrm{Vec}^{(\tau, c) + (r, s+1)}_{\n{G}}(X) \to {}^\varpi\mathrm{Vec}^{(\tau, c) + (r, s)}_{\n{G}}(X)$. We write ${}^\varpi\mathrm{Triv}^{(\tau, c) + (r, s)}_{\n{G}}(X)$ for the submonoid  given by the image of $f$ in  ${}^\varpi\mathrm{Vect}^{(\tau, c) + (r, s)}_{\n{G}}(X)$  and 
$$
{}^\varpi \n{K}^{(\tau, c) + (r, s)}_{\n{G}}(X)_{\mathrm{fin}} 
\;:= \;{}^\varpi\mathrm{Vect}^{(\tau, c) + (r, s)}_{\n{G}}(X)/
{}^\varpi\mathrm{Triv}^{(\tau, c) + (r, s)}_{\n{G}}(X)
$$ 
for the quotient monoid. In the case that $c$ is trivial and $r = s = 0$, one can identify ${}^\varpi \n{K}^{\tau + (0, 0)}_G(X)_{\mathrm{fin}}$ with the Grothendieck construction applied to the monoid of isomorphism classes of finite rank $(\varpi, \tau)$-twisted $\n{G}$-equivariant (ungraded) vector bundles on $X$ without Clifford action. Any $(\varpi, \tau)$-twisted vector bundle $(\bb{E}, \rho)$ on $X$ with a gradation $\Gamma$ admits the trivial (skew-adjoint) Fredholm family $A \equiv 0$. Thus, by the universality, we have a natural homomorphism $\imath : {}^\varpi \n{K}^{\tau + (0, 0)}_{\n{G}}(X)_{\mathrm{fin}} \to {}^\varpi \n{K}^{\tau + (0, 0)}_{\n{G}}(X)$. It is known \cite{freed-moore-13} that the homomorphism $\imath$ is bijective. This fact justifies the original formulation of the Freed-Moore $K$-theoy based on finite rank twisted vector bundles. We refer the reader to \cite[Section 4.4]{gomi-17} for more details about the finite rank realizability of the $K$-theory.
}\hfill $\blacktriangleleft$
\end{remark}


\subsection{$K$-theory for  $\eta$-self adjoint operators}
\label{sec:eta-K-theory}
The framework described in Section \ref{sec:karubi_vs_freed} 
admits a suitable generalization  adapted for the case of 
 gapped $\eta$-self adjoint operators with a 
 $\s{C}$-symmetry. 
Let us recall two fundamental facts. In view of Corollary
\ref{coro_reduc_eta_self_comm} every $\eta$-selfadjoint operator with a $C$-symmetry can be $\eta$-unitarily reduced to a self-adjoint operator $H=H^*$ commuting with the fundamental symmetry $\eta$.
Moreover, Theorem \ref{theo_fund_rep_res} shows that this  reduction transforms $\eta$-quantum symmetries compatible with the $\s{C}$-symmetry
in unitary or anti-unitary operators which {commute} or {anti-commute} with the metric $\eta$ according to the type of the original symmetry. The bundle version of these facts has been summarized in Observation B in Section \ref{sec:hilb_pict}. Te latter justifies the introduction of the following structures.
\begin{assumption}[Framework for $\eta$-self adjoint systems] \label{assumption:freed_moore_K-eta}
Let $X$, $\n{G}$, $\varpi$, $\wp$, $c$, $\tau$ be as follows:
\begin{itemize}
\item[(a)]
$X$ is a compact Hausdorff space;
\vspace{1mm}

\item[(b)]
$\n{G}$ is a finite group acting on $X$ from the left.
\vspace{1mm}

\item[(c)]
$\varpi : \n{G} \to \Z_2$, $\wp : \n{G} \to \Z_2$ and $c : \n{G} \to \Z_2$ are homomorphisms.
\vspace{1mm}

\item[(d)] $\tau\in {Z}^2_{\mathrm{group}}(\n{G}; C(X, \n{U}(1))_\varpi)$
\end{itemize}
\end{assumption}

\begin{definition}[Twisted equivariant Krein  bundle] \label{dfn:twisted_bundle_ungraded_eta}
Let $X$, $\n{G}$, $\varpi$, $\wp$, $c$, $\tau$ be as in Assumption \ref{assumption:freed_moore_K-eta}.
A \emph{$(\varpi, \wp,c,\tau)$-twisted $\n{G}$-equivariant (ungraded) Krein bundle} on $X$ with $C\ell^{r,s}$-action (or a \emph{twisted Krein bundle} for short) is a Krein vector bundle $\pi:\bb{E}\to X$ with fundamental symmetry $\eta$ (\cf Definition \ref{def:krein_bundle})
equipped with the following data:
\begin{itemize}
\item[(a)] A
\emph{twisted $\n{G}$-action} provided by a  representation $\rho$ of $\n{G}$ on the total space  $\bb{E}$
 which covers the left action of $\n{G}$ on $X$ according to the following diagram
$$
\begin{diagram}
\bb{E}&\rTo^{\rho(g)}&  \bb{E}\\
\dTo^{\pi}&&\dTo_{\pi}\\
X&\rTo_{g}&  X\\
\end{diagram} \qquad\quad \forall\;g\in\n{G}
$$
and which
 satisfies
\begin{align*}
\rho(g)^* &\;=\;  \rho(g)^{-1}&\quad&\text{{(isometry)}} \\\vspace{1mm}
\rho(g) \eta &\;=\;  \wp(g)\eta \rho(g)&\quad&\text{{(proper-isometry vs. pseudo-isometry)}} \\\vspace{1mm}
 \rho(g)\; \ii\n{1} &\;=\; \varpi(g)\ii\; \rho(g)&\quad&\text{(linearity vs. anti-linearity)} \\
 \vspace{1mm}
\rho(g_1)\rho(g_2) &\;=\; \tau(g_1, g_2) \rho(g_1g_2)&\quad&\text{(projective representation)}
\end{align*}
for all $g,g_1,g_2\in\n{G}$.\vspace{1mm}
\item[(b)] 
A \emph{(unitary)  action}
$\gamma$ of the Clifford algebra
$C\ell^{r,s}$ 
 (in the sense of Definition \ref{def:cliff_act})
which meets the structural conditions
$$
\begin{aligned}
{\gamma(a) \eta} \;&=\;  {\eta \gamma(a)}&\quad&\text{(compatibility Krein-Clifford structure )}\\
\vspace{1mm}
\gamma(a) \rho(g) \;&=\; c(g) \rho(g) \gamma(a)&\quad&\text{(Koszul sign rule)}
\end{aligned}
$$
for all $a\in C\ell^{r,s}$ and $g\in\n{G}$.
\end{itemize}
A homomorphism $f : (\bb{E},\eta, \rho, \gamma) \to (\bb{E}',\eta', \rho', \gamma')$ of twisted Krein bundles  is  a complex linear map $f : \bb{E} \to \bb{E}'$ which covers the identity of $X$ and satisfies
\begin{align*}
f \circ \eta &\;=\; \eta' \circ f \\
f \circ \rho(g) &\;=\; \rho'(g) \circ f \\
f \circ \gamma(a) &\;=\; \gamma'(a) \circ f
\end{align*}
for all $a\in C\ell^{r,s}$ and $g\in\n{G}$.
When $f$ is bijective $(\bb{E},\eta, \rho, \gamma)$ and $(\bb{E}',\eta', \rho', \gamma')$ are called \emph{isomorphic}. 
\end{definition}

\begin{remark}[$\eta$-PUA-representation]\label{rk:eta_PUA-rep}
{\upshape 
The only difference between   Assumption \ref{assumption:freed_moore_K} and Assumption \ref{assumption:freed_moore_K-eta} is in item (c). In the latter case  the presence of the map $\wp$ allows to  discriminate between (anti)unitarity and pseudo-(anti)unitarity symmetries.
Condition (a) in Definition \ref{dfn:twisted_bundle_ungraded_eta} says that on the bundle $\pi:\bb{E}\to X$ acts the (finite) group $\n{G}$ through
 the representation $\rho:\n{G}\to {\tt QS}_\eta(\rr{H}_{\bb{E}})\cap {\rm Cov}(\bb{E})$ 
 which associates to any $g\in\n{G}$ the covariant $\eta$-quantum symmetry $\rho(g)$. 
 More precisely any $\rho(g)$ is induced by a covariant endomorphism of the bundle $\bb{E}$ which is unitary or anti-unitary and which {commute} or {anti-commute} with the metric $\eta$.
 The combination of these  properties is modeled out of Observation B in Section \ref{sec:hilb_pict}.
The presence of the $2$-cocycle $\tau$ 
suggests that
 the representation $\rho$ is projective, in general. By adapting the jargon of Remark \ref{rk:PUA-rep}, we can refer to $\rho$ as a covarianta \emph{$\eta$-PUA-representation} of the group $\n{G}$.}
\hfill $\blacktriangleleft$
\end{remark}

\medskip

Gapped $\eta$-self-adjoint Hamiltonians with a $\s{C}$-symmetry are reducible to 
gapped self-adjoint Hamiltonians commuting with $\eta$. At this point the \virg{spectral flattening} described in Remark \ref{rk:grad-gap} can be applied as well.
 This observation   suggests that gapped $\eta$-self-adjoint operators with a $\s{C}$-symmetry can be classified by 
\emph{$\eta$-gradations}: 
\begin{definition}[$\eta$-gradation]\label{def_eta_grad}
Let $(\bb{E}, \eta,\rho, \gamma)$ be a $(\varpi,\wp, c,\tau)$-twisted $\n{G}$-equivariant (ungraded) Krein bundle on $X$
 as in Definition \ref{dfn:twisted_bundle_ungraded_eta}. An \emph{$\eta$-gradation}  of $(\bb{E}, \eta,\rho, \gamma)$ is a $\Gamma\in{\rm End}_{\rm lin}(\bb{E})$
such that:
\begin{itemize}
\item[(a1)] $\Gamma$ is a  self-adjoint involution, \ie $\Gamma=\Gamma^*$ and $\Gamma^2=\n{1}$;
\vspace{1mm}
\item[(a2)] $\Gamma\eta=\eta \Gamma$;
\vspace{1mm}
\item[(b)] The relations 
\begin{align*}
\Gamma \rho(g) &\;=\; c(g)\; \rho(g) \Gamma \\
\Gamma \gamma(a) &\;=\; -\gamma(a)\Gamma
\end{align*}
hold for all $a\in C\ell^{r,s}$ and $g\in\n{G}$.
\end{itemize}
\end{definition}

\medskip

By mimicking the construction of the Freed-Moore $K$-theory discussed in Section \ref{sec:karubi_vs_freed}
one can introduce the notions of isomorphism of twisted Krein bundles and of homotopy equivalence of $\eta$-gradations.  By focusing on {triples} $(\bb{E}, \Gamma_0, \Gamma_1)$ given by a  $(\varpi,\wp, c,\tau)$-twisted $\n{G}$-equivariant  Krein bundle $(\bb{E},\eta, \rho, \gamma)$  on $X$  and two $\eta$-gradations $\Gamma_0$ and $\Gamma_1$, along with the notions of equivalence relation and direct sum   introduced in Section \ref{sec:karubi_vs_freed}, one  defines the following objects:
\begin{definition}[$_\eta K$-groups - finite rank case]\label{def_K_group_eta}
Let $X$, $\n{G}$, $\varpi$, $\wp$, $c$, $\tau$ be as in Assumption \ref{assumption:freed_moore_K-eta}. Then:

\begin{itemize}
\item
 ${}^\varpi_\eta \s{M}^{(\tau, c,\wp) + (r, s)}_{\n{G}}(X)$ is the \emph{abelian monoid} of isomorphism classes of triples $(\bb{E}, \Gamma_0, \Gamma_1)$ with addition  given by the direct sum of triples; 
\vspace{1mm}
\item
${}^\varpi_\eta \s{Z}^{(\tau, c,\wp) + (r, s)}_{\n{G}}(X)$ is the submonoid consisting of isomorphism classes of triples $(\bb{E}, \Gamma_0, \Gamma_1)$ such that $\Gamma_0$ and $\Gamma_1$ are \emph{homotopy equivalent};
\vspace{1mm}
\item
 ${}^\varpi_\eta \s{K}^{(\tau, c,\wp) + (r, s)}_\n{G}(X) := {}^\varpi_\eta \s{M}^{(\tau, c,\wp) + (r, s)}_{\n{G}}(X)/{}^\varpi_\eta \s{Z}^{(\tau, c,\wp) + (r, s)}_{\n{G}}(X)$ is  the \emph{quotient monoid}.
\end{itemize}
\end{definition}

\medskip

\noindent 
The properties of the $K$-theory defined above  will be investigated in the last part of this work by a reduction of the monoids  
${}^\varpi _\eta\s{K}^{(\tau, c,\wp) + (r, s)}_\n{G}(X)$ to the Freed-Moore $K$-theory 
${}^\varpi \s{K}^{(\tau, c) + (r, s)}_\n{G}(X)$ described in Section \ref{sec:karubi_vs_freed}.
Once this reduction will be proved, the extension to the infinite rank case and the relation with the
Fredholm operator formulation will follow straightforwardly  as in Section \ref{subsec:infinite_dimensional_Karoubi} and  Section \ref{sec:Fredholm-K} respectively.


\subsection{Reduction to the original Freed-Moore $K$-theory}
\label{sec_reduc_eta-K-to-K}
Let us start by considering the case of a trivial $\wp\equiv+1$. 
which is the case, for instance, 
 when $\n{G}$ reduces to the  trivial group. In this case any $(\varpi,\wp, c,\tau)$-twisted $(\bb{E},\eta, \rho, \gamma)$  bundle admits a decomposition 
$$
(\bb{E},\eta, \rho, \gamma)\;
=\; (\bb{E}^+,\eta^+\equiv+\n{1}, \rho^+, \gamma^+)
\; \oplus\;
(\bb{E}^-,\eta^-\equiv-\n{1}, \rho^-, \gamma^-)
$$
in which the sectors $(\bb{E}^\pm,\eta^\pm, \rho^\pm, \gamma^\pm)$, generated by projecting on the the eigenspaces  of $\eta$ are $(\varpi, \tau, c)$-twisted bundles in the sense of Definition \ref{dfn:twisted_bundle_ungraded}. Also the $\eta$-gradations have a corresponding decomposition 
and one finds that the $_\eta K$-theory in Definition \ref{def_K_group_eta} splits in the sum of two copies of the Freed-Moore $K$-theory of Definition \ref{def_K_group}, namely
\begin{equation}\label{eq:splitting_trivial}
{}^\varpi_\eta \s{K}^{(\tau, c,\wp\equiv+1) + (r, s)}_\n{G}(X)\;
\simeq\; {}^\varpi \s{K}^{(\tau, c) + (r, s)}_\n{G}(X)\; \oplus\;
{}^\varpi \s{K}^{(\tau, c) + (r, s)}_\n{G}(X)\;.
\end{equation}
The splitting phenomenon described by \eqref{eq:splitting_trivial} is illustrated in the concrete (simple) example discussed in Appendix \ref{app:explicit-2}.
However, a reduction similar to \eqref{eq:splitting_trivial} can be systematically extended to the general case.

\medskip

To reduce the $_\eta K$-theory in Definition \ref{def_K_group_eta} to Freed-Moore$K$-theory of Definition \ref{def_K_group}, let us start with a recipe to produce data in Assumption \ref{assumption:freed_moore_K} from those in Assumption \ref{assumption:freed_moore_K-eta}. 
This recipe partly mimics the strategy of the proof of Proposition \ref{prop:reduc_clifford}.  
In particular, let us recall that given the group $\n{G}$ which acts on the left on the space $X$ one can form the extended group $\n{G}':=\n{G}\times \Z_2$. Also this group acts on $X$ in the following way:
Let $p:\n{G}'\to\n{G}$ be the first factor projection given by $p((g,\epsilon))=g$ for all $g\in\n{G}$ and $\epsilon\in\Z_2$, then  $g' x=p(g')x$
for all $g':=(g,\epsilon)\in\n{G}$ and $x\in X$. 
\begin{definition}[Reduction of data] \label{dfn:dictionary_of_data}
Let 
$X$, $\n{G}$, $\varpi$, $\wp$, $c$, $\tau$ be a set of data as in Assumption \ref{assumption:freed_moore_K-eta}. Consider a new set of data $X'$, $\n{G}'$,  $\varpi'$, $c'$, $\tau'$ defined as follows:

\begin{itemize}
\item
The base space is unchanged: ${X'} = X$;\vspace{1mm}

\item
The new group is $\n{G}':=\n{G}\times \Z_2$ and  acts on ${X'} = X$ via the first factor projection $p:\n{G}'\to\n{G}$;\vspace{1mm}

\item
The new homomorphisms ${\varpi'} : \n{G}' \to \Z_2$ and ${c}' : \n{G}' \to \Z_2$ are given by composing $\phi$ and $c$ with the projection $p$, \ie 
\begin{align*}
{\varpi}' &\;=\; \varpi \circ p\;,\qquad\quad
{c}' \;=\; c \circ p\;;
\end{align*}

\item
The new $2$-cocycle ${\tau}' \in {Z}^2_{\mathrm{group}}(\n{G}; C(X', \n{U}(1))_{\varpi'})$ is given by
$$
{\tau}' \big((g_1, \epsilon_1), (g_2, \epsilon_2)\big)
= \wp(g_2)^{\frac{1-\epsilon_1}{2}} \tau(g_1, g_2)\;
$$
for all $g_1, g_2\in \n{G}$ and   $\epsilon_1,\epsilon_2\in\Z_2=\{\pm 1\}$.
\end{itemize}

The new data  $X'$, $\n{G}'$,  $\varpi'$, $c'$, $\tau'$  fulfill all the requirements of
Assumption \ref{assumption:freed_moore_K}.
\end{definition}

\medskip
The reduction of data described in  Definition \ref{dfn:dictionary_of_data}  essentially consists in 	
eliminating  the map $\wp$ which takes into account the  presence of a Krein structure induced by $\eta$.
As will be clarified in the next result, this reduction  is compensated by
the
factor $\Z_2$ in the extended group $\n{G} \times \Z_2$
that will be  responsible for the specification  of an indefinite inner product.
In order to properly enunciate the Lemma below it is necessary to introduce a bit of terminology. Let
$\rho':\n{G} \times \Z_2\to{\tt QS}(\rr{H}_{\bb{E}})\cap {\rm Cov}(\bb{E})$ be a PUA-representation of a Hilbert bundle $\bb{E}$ over $X$.
Observe that in view of the group law one has that 
$$
\rho'\big((g,\epsilon)\big)\;=\;\rho'\big((g,+1)\big)\rho'\big((e,\epsilon)\big)\;,\qquad\quad g\in\n{G}\;,\;\; \epsilon\in\Z_2
$$
where $e$ is the unit of $\n{G}$. One says that $\rho'$ is a \emph{reducible}  PUA-representation   if 
$\rho' ((e,\epsilon) )\;=\;\pm\n{1}$. If this is this case  one has that $\rho'((g,\epsilon))=\pm\rho' ((g,+1))=:\rho(g)$ where  $\rho$ is a PUA-representation of the group $\n{G}$.
It follows that the PUA-representation $\rho'$ is called \emph{non-reducible} when $\eta':=\rho' ((e,-1) )\neq \pm\n{1}$. Assume that the extended group $\n{G} \times \Z_2$ acts on $X$ only via the factor $\n{G}$. In this case $\eta'$ preserves the fibers of $\bb{E}$, namely $\eta'\in {\rm Cov}(\bb{E})$. Moreover, the group law provides $(\eta')^2=\rho' ((e,+1) )=\n{1}$. Then $\eta'$ is a unitary involution, hence self-adjoint. The \emph{non-reducible} PUA-representation $\rho'$ is called \emph{balanced} if in each fiber of $\bb{E}$ the dimensions of the eigenspaces of $\eta'$ related to the eigenvalues $\pm1$ have  same (possibly infinite) dimension.

\begin{lemma} \label{lem:correspondence_key_to_reduction}
Let $X'$, $\n{G}'$,  $\varpi'$, $c'$, $\tau'$  be the set of data obtained from $X$, $\n{G}$, $\varpi$, $\wp$, $c$, $\tau$ according to the recipe in Definition \ref{dfn:dictionary_of_data}. Then, there is a natural one-to-one correspondence between $(\varpi,\wp, c,\tau)$-twisted $\n{G}$-equivariant  Krein bundles on $X$ with $C\ell^{r,s}$-action and  $(\varpi', c',\tau')$-twisted $\n{G}'$-equivariant  Hilbert bundle on $X$ (endowed with a non-reducible balanced PUA-representation $\rho'$) with $C\ell^{r,s}$-action.
Moreover,
under this correspondence 
also gradations and $\eta$-gradations are in 
one-to-one correspondence.

\end{lemma}
\begin{proof}
($\Rightarrow$)
Let $(\bb{E},\eta, \rho, \gamma)$ be a $(\varpi,\wp, c,\tau)$-twisted $\n{G}$-equivariant  Krein bundles on $X$. Consider the triple $(\bb{E}', \rho', \gamma')$ defined as follows: $\bb{E}':=\bb{E}$, $\gamma':=\gamma$ and
$$
\rho'\big((g, \epsilon)\big)\;:=\;\rho(g)\eta^{\frac{1-\epsilon}{2}}\:=\;
\left\{
\begin{aligned}
\rho(g)&\quad&\text{if}\;\; \epsilon=+1\;\,\\
\rho(g)\eta&\quad&\text{if}\;\; \epsilon=-1\;.
\end{aligned}
\right.
$$ 
A  direct check shows that $(\bb{E}', \rho', \gamma')$ is a $(\varpi', c',\tau')$-twisted $\n{G}'$-equivariant  Hilbert bundle on $X$. Moreover $(\bb{E},\eta, \rho, \gamma)$ and $(\bb{E}', \rho', \gamma')$ have by construction the same Clifford action. If $\Gamma$ is an $\eta$-gradation of $(\bb{E},\eta, \rho, \gamma)$
then the condition $\eta\Gamma=\Gamma\eta$
assures that 
$$
\rho'\big((g, \epsilon)\big)\;=\;\Gamma\rho(g)\eta^{\frac{1-\epsilon}{2}}\;=\;
c(g)\Gamma\eta^{\frac{1-\epsilon}{2}}\;=\;c'((g,\epsilon))\rho'\big((g, \epsilon)\big)\Gamma
$$
and this proves that $\Gamma$ is automatically a gradation of $(\bb{E}', \rho', \gamma')$.\\
($\Leftarrow$) Now the converse implication.
Let $(\bb{E}, \rho', \gamma)$ be a $(\varpi', c',\tau')$-twisted $\n{G}\times \Z_2$-equivariant  Hilbert bundle on $X$ such that 
the action of the factor $\Z_2$ is trivial on $X$ and the maps $(\varpi', c',\tau')$ are related to the maps 
$(\varpi, c,\tau)$ defined on $\n{G}$ according to the prescription in Definition \ref{dfn:dictionary_of_data}. Let $\eta:=\rho'((e, -1))$. Since 
the factor $\Z_2$ acts trivially on $X$ it follows that $\eta$ preserves the fibers of  $\bb{E}$.
Since the PUA-representation  $\rho'$ is non-reducible and balanced the operator $\eta$ endows $\bb{E}$ with a Krein structure.
The PUA-representation  $\rho'$ can be represented as
$$
\rho'\big((g, \epsilon)\big)\;=\;\rho(g)\eta^{\frac{1-\epsilon}{2}}
$$
where $\rho(g):=\rho'((g,+1))$ provides a PUA-representation of $\n{G}$.
 A  direct check shows that $(\bb{E}, \eta, \rho, \gamma)$ is a $(\varpi,\wp, c,\tau)$-twisted $\n{G}$-equivariant  Krein bundles on $X$. If $\Gamma$ is a gradation of $(\bb{E}, \rho', \gamma)$ then 
 $$
 \Gamma \rho'\big((g, \epsilon)\big)\;=\;c'((g,\epsilon))\rho'\big((g, \epsilon)\big)\Gamma\;=\;c(g)\rho(g)\eta^{\frac{1-\epsilon}{2}}\Gamma\;.
 $$
The same computation for the group element $(e, -1)$ provides $\Gamma\eta=\eta\Gamma$ which in turn implies  $\Gamma \rho(g)=c(g)\rho(g)\Gamma$. This proves that $\Gamma$ behaves like an
$\eta$-gradation of $(\bb{E}, \eta, \rho, \gamma)$. 
\end{proof}

\medskip

The correspondence established  in Lemma \ref{lem:correspondence_key_to_reduction} immediately leads to the reduction anticipated  at the beginning of this Section.

\begin{theorem}[$K$-theory reduction]\label{theo:rid_eta_K-K}
For a given data set $\n{G}$, $\varpi$, $\wp$, $c$, $\tau$  as in Assumption \ref{assumption:freed_moore_K-eta} there is  
 a natural isomorphism of $K$-theories
$$
{}^\varpi_\eta \s{K}^{(\tau, c,\wp) + (r, s)}_\n{G}(X) 
\;\simeq\; 
{}^{\varpi'} \s{K}^{(\tau', c') + (r, s)}_\n{G'}(X),
$$
where the group $\n{G}'$, the homomorphisms $ \varpi',c'$ and the 2-cocyle ${\tau}'$  
are obtained from the original data according to the prescription described in 
 Definition \ref{dfn:dictionary_of_data}.
\end{theorem}

\begin{proof}
In view of Lemma \ref{lem:correspondence_key_to_reduction}, one has the identification of monoids
\begin{align*}
{}^\varpi_\eta  \s{M}_\n{G}^{(\tau, c, \wp) + (r, s)}(X)
\;\simeq\; 
{}^{\varpi'} 
\s{M}^{({\tau'}, {c'}) + (r, s)}_{\n{G}'}(X)\;,\qquad\quad {}^\varpi_\eta  \s{Z}_\n{G}^{(\tau, c, \wp) + (r, s)}(X)
\;\simeq\; 
{}^{\varpi'} 
\s{Z}^{({\tau'}, {c'}) + (r, s)}_{\n{G}'}(X)\;.
\end{align*}
This leads to the isomorphism in the theorem.
\end{proof}


\subsection{Explicit computations in some special cases}
\label{sec:explicit_comp}
The aim of this Section is to provide the explicit computation of the groups ${}^\varpi_\eta  \s{K}_\n{G}^{(\tau, c, \wp) + (r, s)}(X)$ in some simple cases making use of the reduction described in Theorem \ref{theo:rid_eta_K-K}. Let us assume that: (i) 
$X=\{\ast\}$ consists of a single point, (ii) $\n{G}=\Z_2$ and (iii) $\varpi \equiv +1$ is the trivial homomorphism. 
Under these three assumptions
 the group cohomology turns out to be trivial, \ie ${H}^2_{\mathrm{group}}(\Z_2; \n{U}(1))= 0$, so one can assume that $\tau\equiv+1$ is also trivial. 
 The three assumptions above describe the situation in which one wants to classify $\s{C}$-symmetric $\eta$-self-adjoint gapped operators defined on a given (possibly finite dimensional) Hilbert space $\s{H}$ on which the group action of $\Z_2$ is implemented   by a single linear operator $\rho(-1)$. 
In the simplest possible case
$\s{H}=\C^2$ and the group-action of $\Z_2$  is trivial, \ie $\rho(-1)=\n{1}$. This situation  is described explicitly in Appendix \ref{app:explicit-2}.
The case of a general Hilbert space $\s{H}$ with a  possible non-trivial group-action but with a  trivial homomorphism $\wp\equiv+1$ has been described through the isomorphism
\eqref{eq:splitting_trivial}. In this Section we focus on the case 
of a non-trivial group-action with $\wp:\Z_2\to\Z_2$ the identity map, \ie $\wp(\pm1)=\pm1$. This means that the action of $\Z_2$ on 
$\s{H}$ is implemented by a linear unitary operator $R:=\rho(-1)$ such that $R^2=\n{1}$ and $R\eta=-\eta R$, namely by an $\eta$-reflecting operator.
 The only remaining choice is about $c$.

\medskip

Under the assumptions above one is left with the computation of the groups
$
{}^{+1}_\eta  \s{K}_{\Z_2}^{(+1, c, {\rm Id}) + (r, s)}(\{\ast\})
$
where the constraints 
$\varpi\equiv+1, \tau\equiv+1$ and $\wp\equiv{\rm Id}$ have been explicitly indicated. By effect of the reduction in Theorem \ref{theo:rid_eta_K-K} one obtains the isomorphisms
$$
{}^{+1}_\eta  \s{K}_{\Z_2}^{(+1, c, {\rm Id}) + (r, s)}(\{\ast\})\;\simeq\;
{}^{+1}  \s{K}_{\Z_2\times\Z_2}^{(\tau', c',) + (r, s)}(\{\ast\})\;\simeq\;
  \s{K}_{\Z_2\times\Z_2}^{(\tau', c') + n}(\{\ast\})\;,\qquad\quad n=s-r\quad\text{mod.}\quad 2
$$
where in the last isomorphism the periodicity and the  notation of Proposition \ref{prop:freed_moore_K_property} have been used. The maps $\tau'$ and $c'$ are defined according to the prescription of Definition \ref{dfn:dictionary_of_data}. More precisely one has that 
$c' : \Z_2 \times \Z_2 \to \Z_2$ is given by 
\begin{equation}\label{eq:def_c'}
{c}'\big((\zeta,\epsilon)\big)\;=\;c(\zeta)\;,\qquad\quad (\zeta,\epsilon)\in\Z_2 \times \Z_2\;.
\end{equation}
The 2-cocycle $\tau':( \Z_2 \times \Z_2)\times( \Z_2 \times \Z_2)\to \n{U}(1)$ is described by the following formula:
$$
{\tau}' \big((\zeta_1, \epsilon_1), (\zeta_2, \epsilon_2)\big)
= \wp(\zeta_2)^{\frac{1-\epsilon_1}{2}} \tau(\zeta_1, \zeta_2)\;=\;\zeta_2^{\frac{1-\epsilon_1}{2}}
\;,\qquad\quad (\zeta_j,\epsilon_j)\in\Z_2 \times \Z_2\;.
$$
This 2-cocycle is non-trivial and its values are summarized in Table \ref{tab:exem-01}.
 \begin{table}[h]
 \centering
 \begin{tabular}{|c||c|c|c|c|c|c|c|c|c|c|}
 \hline
 $\tau'$ & $(+1,+1)$    & $(-1,+1)$ & $(+1,-1)$ & $(-1,-1)$\\
\hline
 \hline
 \rule[-3mm]{0mm}{9mm}
$(+1,+1)$  & $+1$ & $+1$ & $+1$ &$+1$
\\
\hline
\rule[-3mm]{0mm}{9mm}
$(-1,+1)$  & $+1$ & $+1$ & $+1$ &$+1$ \\
\hline
 \rule[-3mm]{0mm}{9mm}
$(+1,-1)$   & $+1$ & $-1$ & $+1$ &$-1$
\\
\hline
 \rule[-3mm]{0mm}{9mm}
$(-1,-1)$ & $+1$ & $-1$ & $+1$ &$-1$\\
\hline
  \end{tabular}\vspace{2mm}
 \caption{\footnotesize Values of the 2-cocyle $\tau'$.
}
\label{tab:exem-01}
 \end{table}

\noindent
Interestingly, the presence of the Krein metric $\eta$ translates in the effects of the non-trivial twist $\tau'$
in the computation of the Freed-Moore $K$-theory $\s{K}_{\Z_2\times\Z_2}^{(\tau', c') + n}(\{\ast\})$.

\medskip

In view of Remark \ref{rk:twist_period} the  Freed-Moore  $K$-groups can be expressed in terms of the twisted  $K$-theory for skew-adjoint Fredholm operators. More precisely one has that
$$
\s{K}_{\Z_2\times\Z_2}^{(\tau', c') + n}(\{\ast\})\;\simeq\;
 \n{K}_{\Z_2\times\Z_2}^{({\tau}', c') + n}(\{\ast\})
$$
where one has  tacitly used that $\hat{\tau'}=\tau'$.
Let us compute the latter groups in the two cases given by the different choices of $c$. Definition \eqref{eq:def_c'} shows that $c'$ is trivial if and only if also $c$ is trivial.
\begin{lemma}\label{lemma:comptation_explic}
Regardless of the choice of $c'$ it holds true that
$$
\n{K}_{\Z_2\times\Z_2}^{({\tau}', c') - n}(\{\ast\})\;\simeq\;K_{\C}^{-n}(\{\ast\})
$$
where 
$$
K_{\C}^{-n}(\{\ast\})\;\simeq\;
\left\{
\begin{aligned}
&\Z &\qquad&\text{if}\quad n\; \text{even} \\
&0  &\qquad&\text{if}\quad n\; \text{odd}
\end{aligned}\;
\right.
$$ 
is the usual complex $K$-theory of the single-point space $\{\ast\}$.
\end{lemma}
\proof[{Proof} (sketch of)]
$(c'\equiv +1).$
Assume that $c'$ is trivial. As a consequence
$\n{K}_{\Z_2\times\Z_2}^{({\tau}', +1) - n}(\{\ast\})$ is a usual twisted $K$-theory group.
Then one can apply the known {decomposition}
$$
\n{K}_{\Z_2\times\Z_2}^{({\tau}', +1) - n}(\{\ast\})
\;\simeq\; \s{R}_{\tau'}(\Z_2 \times \Z_2)\; \otimes\; K_{\C}^{-n}(\{\ast\})\;
$$
where   $\s{R}_{\tau'}(\Z_2 \times \Z_2)$ is the free abelian group generated by the $\tau'$-projective complex representations of $\Z_2 \times \Z_2$. 
Such projective representations are unique up to isomorphisms, since they are in one-to-one correspondence with ungraded complex representations of the complexified Clifford algebra 
$\C\ell^{2}:=C\ell^{0,2}\otimes\C$.
 Hence $\s{R}_{{\tau'}}(\Z_2 \times \Z_2) \simeq \Z$ and one finally gets
$$
\n{K}_{\Z_2\times\Z_2}^{({\tau}', +1) - n}(\{\ast\})\;\simeq\; K_{\C}^{-n}(\{\ast\})\;.
$$

\noindent
$(c'\not\equiv +1).$
If $c'$ is  non-trivial then the PUA-representation of $\Z_2\times\Z_2$ is generated by  $R:=\rho((-1,+1))$ which is odd and by $\eta:=\rho((+1,-1))$ which is even.  Let us introduce a different pair of generators
$\tilde{\gamma}_1:=R$ and $\tilde{\gamma}_2:=R\eta$ which are both odd. In view of the twisting $\tau'$ one has that $R\eta=-\eta R$ and this implies that
$$
\tilde{\gamma}_1^2\;=\;\n{1}\;,\qquad\tilde{\gamma}_2^2\;=\;-\n{1}\;,\qquad\tilde{\gamma}_1\tilde{\gamma}_2\;=\;-\tilde{\gamma}_2\tilde{\gamma}_1\;.
$$
Therefore a $({\tau'}, {c'})$-twisted representation of $\Z_2 \times \Z_2$ with $C\ell^{r,s}$-action amounts to a $\Z_2$-graded representation of the complexified Clifford algebra 
$\C\ell^{r+s+2}:=C\ell^{r+1,s+1}\otimes\C$.
 The $K$-theory $\n{K}_{\Z_2\times\Z_2}^{({\tau}', c') - (s-r)}(\{\ast\})$ is the quotient of the free abelian monoid generated by $\Z_2$-graded $\C\ell^{r+s+2}$-modules by the submonoid generated by those which can be extended to $\Z_2$-graded $\C\ell^{r+s+3}$-modules. This group is exactly $K_{\C}^{r + s + 2}(\{\ast\})$. Thus, by using the Bott periodicity one gets
$$
\n{K}_{\Z_2\times\Z_2}^{({\tau}', c') - n}(\{\ast\})
\;:=\;\n{K}_{\Z_2\times\Z_2}^{({\tau}', c') - (s-r)}(\{\ast\})
\;\simeq\;K_{\C}^{r + s + 2}(\{\ast\})
\;\simeq\;K_{\C}^{(r-s) + 2(s + 1)}(\{\ast\})
\;\simeq\;K_{\C}^{-n}(\{\ast\})\;.
$$
The proof is completed.
\qed

\medskip

In conclusion, under the simplified hypotheses (i), (ii) and (iii) above, and using the various isomorphisms and the computation in  Lemma \ref{lemma:comptation_explic} one obtains that
\begin{equation}\label{eq:isomerf-01}
{}^{+1}_\eta  \s{K}_{\Z_2}^{(+1, c, {\rm Id}) + (r, s)}(\{\ast\})\;\simeq\;K_{\C}^{-(s-r)}(\{\ast\})\;,\qquad\quad \wp\equiv {\rm Id}\;.
\end{equation}
This result  must be compared with the splitting \eqref{eq:splitting_trivial} which provides
$$
{}^{+1}_\eta  \s{K}_{\Z_2}^{(+1, c, +1) + (r, s)}(\{\ast\})\;\simeq\; \s{K}^{(+1, c) + (r, s)}_{\Z_2}(\{\ast\})\; \oplus\;
 \s{K}^{(+1, c) + (r, s)}_{\Z_2}(\{\ast\})
 $$
in the case of a trivial homomorphism $\wp$.
Since $\s{K}^{(+1, c) + (r, s)}_{\Z_2}(\{\ast\})\simeq \n{K}_{\Z_2}^{(+1, c) - (s-r)}(\{\ast\})$
and
$$
\n{K}_{\Z_2}^{(+1, c) - (s-r)}(\{\ast\})\;\simeq\;
K_{\C}^{-(s-r)-\delta(c)}(\{\ast\})
$$
where
$$
\delta(c)\;:=\;\left\{
\begin{aligned}
&0 &\qquad&\text{if}\quad c\equiv+1 \\
&1  &\qquad&\text{if}\quad c\equiv{\rm Id}\;,
\end{aligned}\;
\right.
$$
one concludes that
$$
{}^{+1}_\eta  \s{K}_{\Z_2}^{(+1, c, +1) + (r, s)}(\{\ast\})\;\simeq\;K_{\C}^{-(s-r)-\delta(c)}(\{\ast\})\;\oplus\; K_{\C}^{-(s-r)-\delta(c)}(\{\ast\})\;,\qquad\quad \wp\equiv +1\;
$$
which is valid when $\wp$ acts trivially.

\appendix

\section{Maxwell-type operators and metamaterials}
\label{sec:maxwell-meta}
This Section is devoted to the presentation of a particular family of $\eta$-self-adjoint models 
that is finding use in the recent physical literature. The toy model $\rr{m}_w$ described in equation \eqref{eq:mod2} is a particularly simple representative of this family.

\medskip

Let $M_0$ be a (possibly unbounded) self-adjoint operator on the complex Hilbert space $\s{H}$ with dense
domain $\s{D}_0$ and spectrum $\sigma(M_0)\subseteq\R$.
 In many physical situations
it can be useful to endow $\s{H}$ with a complex structure given by the anti-unitary involution $C$. In this case it is natural to require that $M_0$ respect the given  complex structure:
\begin{assumption}\label{ass:maxw1}
The complex structure preserves the domain of $M_0$, namely $C[\s{D}_0]\subseteq \s{D}_0$, and 
$M_0$ is \emph{even} ($CM_0C=M_0$) or \emph{odd} ($CM_0C=-M_0$)
with respect to $C$.
\end{assumption}

\medskip

\noindent
Assumption \ref{ass:maxw1} has consequences on the form of the spectrum of $M_0$. Let us focus on the odd case:

\begin{proposition}
Under the odd case of Assumption \ref{ass:maxw1} the spectrum of $M_0$ is symmetric around zero, namely $\omega\in\sigma(M_0)$ if and only if $-\omega\in\sigma(M_0)$.
\end{proposition}
\proof
Let $\omega\in\rho(M_0)=\C\setminus \sigma(M_0)$ be a point in the resolvent set. Then
$$
R_0(\omega)\;:=\;\frac{1}{M_0\;-\;\omega\n{1}}
$$
exists as a bounded operator. Since $CR_0(\omega)C=-R_0(-\omega)$ one concludes that also $-\omega\in\rho(M_0)$. The proof is concluded by observing that $\sigma(M_0)=\C\setminus\rho(M_0)$.
\qed

\medskip

\begin{definition}[Metamaterial weight]\label{def:meta1}
A  \emph{metamaterial weight} on $\s{H}$ is a bounded self-adjoint operator $W$ that satisfies $
0\;\notin\sigma(W)$. The weight is called \emph{real} if $CWC=W$. 
\end{definition}

\begin{proposition}
Let $W$ be a metamaterial weight. Then there is a pair of bounded self-adjoint operators  $(W_+,W_-)$ and non negative constants $0\leqslant a_\pm<+\infty$ such that:
\begin{itemize}
\item[(a)] $W=W_++W_-$ and $W_\pm W_\mp=0$;
\vspace{1mm}
\item[(b)] $
0\;< \;\pm W_\pm\;\leqslant\; a_\pm\n{1}
$.
\end{itemize}
\end{proposition}
\proof
Since the spectrum of a self-adjoint operator is closed there are positive  constants $0<b_\pm\leqslant a_\pm<+\infty$ such that
$$
\sigma(W)\;\subseteq\; [-a_-,-b_-]\;\cup\;[b_+,a_+]\;.
$$
By spectral calculus one can define the spectral projections $\Pi_-:=\chi_{ [-a_-,-b_-]}(W)$ and
$\Pi_+:=\chi_{ [b_+,a_+]}(W)$. Then the claim follows by setting $W_\pm:=\Pi_\pm W\Pi_\pm$. 
\qed

\medskip

\noindent
From Definition \ref{def:meta1} it follows that $W^{-1}$ exists as a bounded self-adjoint operator. Therefore one can associate with $W$ the operator
\begin{equation}\label{eq:max-typ001}
\eta_W\;:=\;\frac{W}{|W|}\;=\;\frac{W_+}{|W|}+\frac{W_-}{|W|} \;.
\end{equation}
It turns out that $\eta_W$ has the properties of a fundamental symmetry according to Remark \ref{rk:fund_sym}.

\medskip

Let us focus the attention on the family of operators of the following type:
\begin{definition}[Metamaterial Maxwell-type operator]\label{def:meta_max-typ001}
Let $M_0$ be a self-adjoint operator with dense domain $\s{D}_0$ which meets Assumption \ref{ass:maxw1}. Let $W$ be a metamaterial weight in the sense of Definition \ref{def:meta1}.
The operator 
\begin{equation}\label{eq:max-typ1}
M\;:=\;W\; M_0\;=\;\eta_W\; \tilde{M}\qquad\text{with}\qquad \tilde{M}\;:=\;|W|\;M_0
\end{equation}
is called a \emph{metamaterial Maxwell-type operator}.
\end{definition}

\medskip

\noindent
From its very definition  it follows that $M$ is densely defined with domain
$$
\s{D}(M)\;=\;\s{D}(M_0)\;=\;\s{D}_0\;.
$$

\begin{remark}[Maxwell-type operator]\label{rk_max_eqat}{\upshape 
The general notion of  Maxwell-type operator has been discussed in \cite[Section 6]{denittis-lein-18}. The name comes from the fact that the Maxwell's equations for the propagation of the light in a medium can be recast in the form \eqref{eq:max-typ1} \cite{wilcox-66,kato-67,birman-solomyak-87} (see also \cite{denittis-lein-14a,denittis-lein-14b,denittis-lein-18}). In this specific case the \emph{free} operator, acting on the Hilbert space $L^2(\R^3,\C^6)$, is
\begin{equation}\label{eq:max-typ1_M0}
M_0\;:=\;\left(\begin{array}{cc}0 & +\ii\nabla^\times \\ -\ii\nabla^\times & 0\end{array}\right)
\end{equation}
where $\nabla^\times$ denotes the \emph{curl} operator. The material weight is described by the  matrix-valued function
$$
W(x)\;:=\;\left(\begin{array}{cc}\varepsilon^{-1}(x) & 0 \\ 0 & \mu^{-1}(x) \end{array}\right)
$$
where $\varepsilon(x)$ is the electric permittivity tensor and  $\mu(x)$ is the magnetic permeability tensor. In the \virg{standard} definition of Maxwell-type operator the operator $W$ is required to be positive (\cf \cite[Definition 6.1]{denittis-lein-18}). On the contrary, Definition \ref{def:meta1} allows $W$ to have also a negative part. For instance this is the case which models \emph{single-negative metamaterials} \cite{fredkin-ron-02,alu-engheta-03,jiang-chen-li-zhang-zi-zhu-04} in which the electric permittivity tensor is positive $\varepsilon>0$ and the magnetic permeability tensor is negative 
 $\mu<0$ (or viceversa). This justifies the use of the expression \emph{metameterial} in Definition \ref{def:meta1} and Definition \ref{def:meta_max-typ001}. Finally, let us notice that the operator $\rr{m}_w$ in \eqref{eq:mod2} is an example
of metamaterial Maxwell-type operator.
}\hfill $\blacktriangleleft$
\end{remark}

Let $\s{H}_*$ be the Hilbert space obtained  by endowing $\s{H}$ with the weighted scalar product 
\begin{equation}\label{eq:wigh_prod}
\langle\phi,\varphi\rangle_*\;:=\;\langle\phi,|W|^{-1}\varphi\rangle\;.
\end{equation}
The equation
\begin{equation}
\langle\phi,\eta_W\varphi\rangle_*\;=\;\langle\phi,W|W|^{-2}\varphi\rangle\;=\;\langle \eta_W\phi,|W|^{-1}\varphi\rangle\;=\;\langle \eta_W\phi,\varphi\rangle_*
\end{equation}
shows that $\eta_W$ is self-adjoint, hence a fundamental symmetries, also in the weighted Hilbert space $\s{H}_*$. Therefore $\eta_W$ induces a Krein space  structure
$(\s{H}_*,\langle\langle\;,\;\rangle\rangle_{\eta_W})$  by means of the indefinite inner product
$$
\langle\langle \phi, \varphi \rangle\rangle_{\eta_W}\;: =\; \langle \phi, \eta_W \varphi \rangle_*\; =\; \langle \phi, W^{-1} \varphi \rangle\;,\qquad\quad \forall \phi, \varphi\in\s{H}_*\;.
$$

\begin{proposition}
Let $M=WM_0$ be a metamaterial Maxwell-type operator in the sense of Definition \ref{def:meta_max-typ001}. Then:
\begin{itemize}
\item[(1)] The 
operator $\tilde{M}:=|W|M_0$ is self-adjoint on the weighted Hilbert space $\s{H}_*$ with dense domain
$\s{D}_0$.
\vspace{1mm}
\item[(2)] The operator ${M}=\eta_W\tilde{M}$ is $\eta_W$-self-adjoint on the Krein space $(\s{H}_*,\langle\langle\;,\;\rangle\rangle_{\eta_W})$.
\end{itemize}
\end{proposition}
\proof
Item (1) is a particular case 
of \cite[Proposition 6.2]{denittis-lein-18}. Item
(2) follows from the self-adjointness of $\tilde{M}$ and $\eta_W$   with respect to  the  Hilbert structure of $\s{H}_*$. This implies  that the adjoint of ${M}$ is $\tilde{M}\eta_W=\eta_W M \eta_W$. The last relation is equivalent to the $\eta_W$-self-adjointness of  $M$. 
\qed

\medskip

The last result is preparatory for the following natural question: \emph{When a metamaterial Maxwell-type operator admits a $\s{C}$-symmetry $\Xi$?} A set of sufficient conditions,  borrowed from \cite[Proposition]{kuzhel-sudilovskaya-17},  is described in the following: 
\begin{theorem}\label{theo_metamat}
Let $M=WM_0$ be a metamaterial Maxwell-type operator in the sense of Definition \ref{def:meta_max-typ001}. Assume that: (a) $M_0$ is a non negative operator on $\s{H}$ and $0\notin{\rm Ker}(M_0)$; (b)
$\rho(M)\neq\emptyset$;  (c) $0$ and $\infty$ are not singular critical points of $M$.
 Then, the $\eta_W$-self-adjoint  operator   $M$ admits the 
 $\s{C}$-symmetry $\Xi:=\n{E}_M(\R_+)-\n{E}_M(\R_-)$ where $\n{E}_M$ denotes the projection-valued spectral measure of $M$.
\end{theorem}
\proof
Condition (a) is equivalent to $\langle\langle \varphi, M\varphi \rangle\rangle_{\eta_W}\geqslant 0$ for all 
$\varphi\in\s{H}_*$ meaning  that $M$ is $\eta_W$-non negative. Together with condition (b) this is enough to prove that $\sigma(M)\subseteq\R$ \cite[Chapter 2, Theorem 3.27]{azizov-iokhvidov-67}. 
Moreover, conditions (a) and (b) imply that $M$ is \emph{definitizable} in the sense of \cite[Section I.3]{langer-82}. Moreover the only possible critical points  of $M$ are $0$ and $\infty$ \cite[Section II.6]{langer-82}. The second part of condition (a) along with the invertibility of $W$ assures that $0$ is not an eigenvalue of $M$. If $0$ and $\infty$ are not singular critical points then \cite[Theorem 5.7]{langer-82}
assures the existence of a spectral functional calculus for $M$ and the spectral projections  $\n{E}_M(\R_+)$ and $\n{E}_M(\R_-)$ are well defined. The operator $\Xi:=\n{E}_M(\R_+)-\n{E}_M(\R_-)$ verifies all the properties of a  $\s{C}$-symmetry for $M$.
\qed

\medskip

Theorem \ref{theo_metamat} provides a set of  criterions for
 a metamaterial Maxwell-type operator to describe a dynamically stable system in the sense of Definition \ref{def:dyn_stab}. However condition (a) is quite unpleasant since it excludes the case of the Maxwell equations described in Remark \ref{rk_max_eqat}. Indeed, the operator $M_0$ described by \eqref{eq:max-typ1_M0} has spectrum $\sigma(M_0)=\R$ and does not meet condition (a). However, it is odd in the sense of Assumption \ref{ass:maxw1}. This leads to the following:

\medskip

\noindent 
{\bf Open  problem.}
\emph{Is it possible to prove that a   metamaterial Maxwell-type operator $M=WM_0$ admits a  $\s{C}$-symmetry $\Xi$ under the odd case of Assumption \ref{ass:maxw1}? If so, what other conditions are necessary?}


\section{Some explicit computation in dimension two}
\label{app:explicit-2}
In this Section we focus on the simplest (finite dimensional) Krein space given by the Hilbert space $\C^2$ endowed with the metric operator
\begin{equation}\label{eq:appB_sandard_grad}
\eta\;:=\;\left(\begin{array}{cc}1 & 0 \\0 & -1\end{array}\right)\;.
\end{equation}
The canonical basis $v_1:=(1,0)$ and $v_2:=(0,1)$ diagonalizes $\eta$ and the usual complex conjugation $C$ is compatible with $\eta$ in view of the fact that $Cv_j=v_j$ for $j=1,2$. The $\eta$-reflecting symmetry is given by
$$
R\;:=\;\left(\begin{array}{cc}0 & 1 \\1 & 0\end{array}\right)\;.
$$
The following relations hold true: $C\eta=\eta C$, $R\eta=\eta R$ and $R\eta=-\eta R$. 
In order to determine
$$
 {\tt QS}_\eta(\C^2)\;:=\;{\tt U}_\eta(\C^2)\;\sqcup\; {\tt AU}_\eta(\C^2)\sqcup\; {\tt PU}_\eta(\C^2)\sqcup\; {\tt PAU}_\eta(\C^2)
 $$
it is enough to compute ${\tt U}_\eta(\C^2)$ in view of the relations \eqref{eq:rep-group_eta-ant-unit},
\eqref{eq:rep-group_pseud-eta-unit}  and \eqref{eq:rep-group_pseud-eta-ant-unit}.
\begin{lemma}
\label{lwmma:appB-iso}
Any element of ${\tt U}_\eta(\C^2)$ can be parametrized as
$$
U(r, \alpha,\beta, \delta)\;:=\;\left(\begin{array}{cc}\expo{\ii\alpha}\sqrt{1+r^2} & \expo{\ii(\beta+\delta)}r \\\expo{\ii(\alpha-\delta)}r & \expo{\ii\beta}\sqrt{1+r^2}\end{array}\right)\;,\qquad r\in[0,+\infty)\;,\quad \alpha,\beta,\delta\in[-\pi,+\pi].
$$
Moreover, $\|U(r, \alpha,\beta, \delta)\|=r+\sqrt{1+r^2}$.
\end{lemma}
\proof
Given a complex matrix
$$
U\;:=\;\left(\begin{array}{cc}a & b \\c & d\end{array}\right)\;,\qquad\quad a,b,c,d\in\C
$$
the condition $U^*\eta U=\eta$, which ensures $U\in {\tt U}_\eta(\C^2)$, reads as
$$
\left(\begin{array}{cc}|a|^2-|c|^2 & \bar{a}b- \bar{c}d\\
a\bar{b}- c\bar{d} & |b|^2-|d|^2\end{array}\right)\;=\;\left(\begin{array}{cc}1 & 0 \\0 & -1\end{array}\right)\;.
$$
From $\bar{a}b= \bar{c}d$ one gets $|{a}|^2|b|^2=|c|^2|d|^2$. By inserting $|a|^2=|c|^2+1$ and 
$|d|^2=|b|^2+1$ in the previous equation one finally gets $|b|=|c|=:r\geqslant0$. This implies that $|a|=|d|=\sqrt{1+r^2}$ and so there are two angles $\alpha$ and $\beta$ such that
$$
a\;:=\;\expo{\ii\alpha}\sqrt{1+r^2}\;,\quad\text{and}\qquad d\;:=\;\expo{\ii\beta}\sqrt{1+r^2}\;.
$$
Let $b=\expo{\ii\delta_b}r$ and $c=\expo{\ii\delta_c}r$. The condition $\bar{a}b= \bar{c}d$ implies 
that $\delta_b+\delta_c=\alpha+\beta$. The last condition is satisfied by  $\delta_b:=\beta+\delta$ and  $\delta_c:=\alpha-\delta$ for some angle $\delta$. This completes the proof of the parametrization of $U(r, \alpha,\beta, \delta)$. From the parametrization it follows that
$$
U(r, \alpha,\beta, \delta)\;=\;H(r, \delta)\; D(\alpha,\beta)
$$
where
$$
H(r, \delta)\;:=\;\left(\begin{array}{cc}\sqrt{1+r^2} & \expo{+\ii\delta}r \\\expo{-\ii\delta}r & \sqrt{1+r^2}\end{array}\right)\;,\quad\text{and}\qquad 
D(\alpha,\beta)\;:=\;\left(\begin{array}{cc}\expo{\ii\alpha} & 0 \\0 & \expo{\ii\beta}\end{array}\right)\;.
$$
The matrix $H(r, \delta)$ is self-adjoint while $D(\alpha,\beta)$ is unitary. This implies that
$$
\|U(r, \alpha,\beta, \delta)\|\;=\;\|H(r, \delta) D(\alpha,\beta)\|\;=\;\|H(r, \delta)\|\;=\;r+\sqrt{1+r^2}
$$
where $r+\sqrt{1+r^2}$ is the maximum eigenvalue of $H(r, \delta)$.
\qed

\medskip

In the next result the $\R$-Banach space ${\tt H}_\eta(\n{C}^2)$ of the $\eta$-self-adjoint operators is described. 

\begin{lemma}
\label{lwmma:appB-eta-self}
Any element of ${\tt H}_\eta(\n{C}^2)$ can be parametrized as
$$
H(x_1,x_2, y, z)\;:=\;\left(\begin{array}{cc}x_1 & y+\ii z \\-y+\ii z & x_2\end{array}\right)\;,\qquad\quad x_1,x_2,y,z\in\R\;.
$$
\end{lemma}
\proof
Given a complex matrix
\begin{equation}\label{eq:appB_repH_ETASELF}
H\;:=\;\left(\begin{array}{cc}a & b \\c & d\end{array}\right)\;,\qquad\quad a,b,c,d\in\C
\end{equation}
the condition $\eta H^*\eta =H$, which ensures $H\in {\tt H}_\eta(\n{C}^2)$, reads as
$$
\left(\begin{array}{cc}\bar{a} & -\bar{c} \\-\bar{b} & \bar{d}\end{array}\right)=\left(\begin{array}{cc}a & b \\c & d\end{array}\right)\;.
$$
The three relations $\bar{a}=a,\bar{d}=d$ and $\bar{b}=-c$ can be solved in terms of four real parameters $x_1,x_2,y,z\in\R$ as follows
$$
a\;:=\;x_1\;\quad\; d\;:=\;x_2\;\quad\;\quad\; b\;:=\;y+\ii z\;\quad \; c\;:=\;-y+\ii z\;
$$
This completes the proof.
\qed

\medskip

Also the space $
{\tt C}_\eta(\C^2)$
 of the $\s{C}$-symmetries and the space 
${{\tt R}}_\eta(\C^2)$ of the self-adjoint operators  anti-commuting with the metric $\eta$  can be easily described.

\begin{lemma}
\label{lwmma:appB-Csymm-anticomm}
Any element of ${\tt C}_\eta(\C^2)$ can be parametrized as
\begin{equation}\label{eq:appB_repXi}
\Xi(r, \theta)\;:=\;\left(\begin{array}{cc}\sqrt{1+r^2} & r\expo{\ii\theta} \\-r\expo{-\ii\theta} & -\sqrt{1+r^2}\end{array}\right)\;,\qquad r\in[0,+\infty)\;,\quad \theta\in[-\pi,+\pi]\;.
\end{equation}
Moreover, $\|\Xi(r, \theta)\|=r+\sqrt{1+r^2}$.
Similarly, any element of ${\tt R}_\eta(\C^2)$ can be parametrized as
\begin{equation}\label{eq:appB_repQ}
Q(r, \theta)\;:=\;\left(\begin{array}{cc}0 & r\expo{\ii\theta} \\r\expo{-\ii\theta} & 0\end{array}\right)\;,\qquad r\in[0,+\infty)\;,\quad \theta\in[-\pi,+\pi]\;
\end{equation}
and $\|Q(r, \theta)\|=r$.
\end{lemma}
\proof
Let $\Xi\in {\tt C}_\eta(\C^2)$. In view of \eqref{eq:Xi_eta-self} one has that $\Xi$ is $\eta$-self-adjoint and so it can be represented as
$$
\Xi\;:=\;\left(\begin{array}{cc}x_1 & y+\ii z \\-y+\ii z & x_2\end{array}\right)\;,\qquad\quad x_1,x_2,y,z\in\R\;.
$$
according to Lemma \ref{lwmma:appB-eta-self}.
The involutive property $\Xi^2=\n{1}$ 
reads as
$$
\left(\begin{array}{cc}x_1^2-y^2-z^2 & (x_1+x_2)(y+\ii z)\\
(x_1+x_2)(-y+\ii z) & x_2^2-y^2-z^2\end{array}\right)\;=\;\left(\begin{array}{cc}1 & 0 \\0 & 1\end{array}\right)\;.
$$
The relations $x_1^2-y^2-z^2=1=x_2^2-y^2-z^2$ imply that $|x_1|=|x_2|\neq0$. 
Let us set $y+\ii z:=r\expo{\ii\theta}$ and 
$-y+\ii z:=-r\expo{-\ii\theta}$.
The conditions $(x_1+x_2)(y+\ii z)=0=(x_1+x_2)(-y+\ii z)$ imply $|(x_1+x_2)r|=0$. In the case $r=0$ the matrix $\Xi$ is diagonal and $|x_1|=1=|x_2|$. The condition $\eta\Xi>0$ implies $x_1>0$ and $-x_2>0$. As a result the condition $r=0$ implies $\Xi=\eta$. Let us assume now that $r\neq0$. This implies that $x_2=-x_1$ and $x_1=\sqrt{1+r^2}$.  The matrix $\eta\Xi$ reads
$$
\eta\Xi\;:=\;\left(\begin{array}{cc}\sqrt{1+r^2} & r\expo{\ii\theta} \\r\expo{-\ii\theta} & \sqrt{1+r^2}\end{array}\right)\;.
$$
 The eigenvalues of the self-adjoint matrix
$\eta\Xi$ are  $\sqrt{1+r^2}\pm r>0$ showing that 
$\eta\Xi$ is positive. Moreover,
the unitarity of $\eta$ implies $\|\Xi\|=\|\eta\Xi\|=r+\sqrt{1+r^2}$. 

Since the elements of ${\tt R}_\eta(\C^2)$ are self-adjoint one has
$$
Q\;:=\;\left(\begin{array}{cc}x_1 & r\expo{\ii\theta} \\r\expo{-\ii\theta} & x_2\end{array}\right) 
$$
with $x_1,x_2\in\R$. The condition $\eta Q=-Q\eta$ forces $x_1=x_2=0$. The eigenvalues of $Q$ are $\pm r$ and this shows that $\|Q\|=r$.
\qed

\medskip

Starting from \eqref{eq:appB_repQ} one can compute
$$
\eta\expo{Q}\;=\;\left(\begin{array}{cc}\cosh(r) & \expo{\ii\theta}\sinh(r) \\ -\expo{-\ii\theta} \sinh(r)& -\cosh(r)\end{array}\right)\;=\;\left(\begin{array}{cc}\sqrt{1+\sinh(r)^2} & \expo{\ii\theta}\sinh(r) \\ -\expo{-\ii\theta} \sinh(r)& -\sqrt{1+\sinh(r)^2}\end{array}\right)
$$
where $\cosh(r) :=\frac{1}{2}(\expo{r}+\expo{-r})$ and $\sinh(r) :=\frac{1}{2}(\expo{r}-\expo{-r})$
are the  hyperbolic functions that satisfy $\cosh(r)^2-\sinh(r)^2=1$. The change of variable
$r':= \sinh(r)$ and the comparison with \eqref{eq:appB_repXi} shows that $\eta\expo{Q}\in {\tt C}_\eta(\C^2)$ in accordance with the general result proved in Lemma \ref{lem:space_of_C_symmetries}. Moreover, as a consequence of Lemma \ref{lwmma:appB-Csymm-anticomm} one has that ${\tt C}_\eta(\C^2)\simeq\C\simeq {\tt R}_\eta(\C^2)$.

\medskip

We are now in position to study the set of $\eta$-self adjoint operators which admit a 
$\s{C}$-symmetry. 
\begin{lemma}
\label{lemma:appB_H_eta_Xi}
Let $\Xi\in{\tt C}_\eta(\C^2)$ be a $\s{C}$-symmetry and $H\in{\tt H}_\eta(\n{C}^2)$ an $\eta$-self adjoint operator. The condition $H\Xi=\Xi H$ holds true if and only if $H$ is of the form
\begin{equation}\label{eq:appB_repH_Xi_sym}
H\;=\; u\;\n{1}\;+\; v\;\Xi\;,\qquad\quad u,v\in\R\;.
\end{equation}
\end{lemma}
\proof
Clearly operators of the form \eqref{eq:appB_repH_Xi_sym} are $\eta$-self-adjoint since $\n{1}$ and $\Xi$ are  $\eta$-self-adjoint  and commute with $\Xi$. To prove the converse let us use the parametrization 
\eqref{eq:appB_repH_ETASELF} and \eqref{eq:appB_repXi} for $H$ and $\Xi$, respectively.
A straightforward computation shows that
$$
H\Xi\;=\;\left(\begin{array}{cc}x_1\sqrt{1+r^2}-(y+\ii z)r\expo{-\ii\theta} & x_1r\expo{\ii\theta}- (y+\ii z)\sqrt{1+r^2}\\
(-y+\ii z)\sqrt{1+r^2}-x_2r\expo{-\ii\theta} & -x_2\sqrt{1+r^2}+(-y+\ii z)r\expo{\ii\theta}\end{array}\right)
$$
and 
$$
\Xi H\;=\;\left(\begin{array}{cc}x_1\sqrt{1+r^2}+(-y+\ii z)r\expo{\ii\theta} & x_2r\expo{\ii\theta}+ (y+\ii z)\sqrt{1+r^2}\\
-(-y+\ii z)\sqrt{1+r^2}-x_1r\expo{-\ii\theta} & -x_2\sqrt{1+r^2}-(y+\ii z)r\expo{-\ii\theta}\end{array}\right)\;.
$$
The commutation relations between $H$ and $\Xi$ imply the following conditions
$$
(y+\ii z)\;=\;(x_1-x_2)\frac{r}{2\sqrt{1+r^2}}\expo{\ii\theta}\;,\qquad
(-y+\ii z)\;=\;-(y+\ii z)\expo{-\ii2\theta}\;=\;(x_2-x_1)\frac{r}{2\sqrt{1+r^2}}\expo{-\ii\theta}\;.
$$
Let us define $u:=\frac{x_1+x_2}{2}$ and  $v:=\frac{x_1-x_2}{2\sqrt{1+r^2}}$. With this notation one has that $(y+\ii z)=vr\expo{\ii\theta}$, $x_1=u+v\sqrt{1+r^2}$ and $x_2=u-v\sqrt{1+r^2}$. After plugging in these expressions in \eqref{eq:appB_repH_ETASELF} one obtains the representation \eqref{eq:appB_repH_Xi_sym}.
\qed

\medskip

Let us compare the  $\R$-Banach space ${\tt H}_\eta(\n{C}^2)$ of the $\eta$-self-adjoint operators
with its subspace ${\tt H}_{\eta,\Xi}(\n{C}^2)$ of the $\eta$-self-adjoint operators commuting with the 
$\s{C}$-symmetry $\Xi$.
Lemma  \ref{lwmma:appB-eta-self} shows that ${\tt H}_\eta(\n{C}^2)\simeq\R^4$ while Lemma \ref{lemma:appB_H_eta_Xi} proves that ${\tt H}_{\eta,\Xi}(\n{C}^2)\simeq\R^2$.
The (total) space of all $\eta$-self-adjoint  operators which admit a $\s{C}$-symmetry is given by
 $$
{\tt CH}_{\eta}(\n{C}^2)\;:=\;\bigsqcup_{\Xi\in {{\tt C}}_\eta(\n{C}^2)}{\tt H}_{\eta,\Xi}(\n{C}^2)\;.
$$
This is a real vector bundle over ${\tt C}_\eta(\n{C}^2)\simeq\C$ with typical fiber $\R^2$. This bundle is trivial since $\C$ is contractible.

\medskip

The metric operator $\eta$ has the property of a  $\s{C}$-symmetry. It turns out that ${\tt H}_{\eta,0}(\n{C}^2):={\tt H}_{\eta,\eta}(\n{C}^2)$ is the space of the  self-adjoint and  $\eta$-self-adjoint  operators.
As a consequence of Lemma \eqref{lemma:appB_H_eta_Xi} one has that:
\begin{corollary}
Any  $H\in {\tt H}_{\eta,0}(\n{C}^2)$  can be uniquely represented as
\begin{equation}\label{eq:appB_repH_Xi_sym_self_eta_self}
H\;=\; \left(\begin{array}{cc}u+v & 0\\
0 & u-v\end{array}\right)
,\qquad\quad u,v\in\R\;.
\end{equation}
\end{corollary}

\medskip

Let $H=H(u,v)$ be the element of ${\tt H}_{\eta,\Xi}(\n{C}^2)$ given by the parametrization \eqref{eq:appB_repH_Xi_sym}. Then the following facts are easily verifiable:
\begin{itemize}
\item[(1)] ${\rm det}\,H(u,v)=u^2-v^2$ and ${\rm Tr}\, H(u,v)=2u$;
\vspace{1mm}
\item[(2)] The  eigenvalues of $H(u,v)$ are $\lambda_\pm:=u\pm v$ and the related eigenvectors are
$$
\psi_+\;:=\;\left(\begin{array}{c}r\expo{\ii\theta} \\1-\sqrt{1+r^2}\end{array}\right)\;,\qquad\quad
\psi_-\;:=\;\left(\begin{array}{c}-r\expo{\ii\theta} \\1+\sqrt{1+r^2}\end{array}\right)\;
$$
where $r$ and $\theta$ are the parameters which parametrize the $\s{C}$-symmetry $\Xi$ according to \eqref{eq:appB_repXi}.
Moreover $\Xi\psi_\pm=\pm\psi_\pm$.
\end{itemize}

\medskip

An element $H\in {\tt H}_{\eta,\Xi}(\n{C}^2)$ is called \emph{gapped} (in zero) if it is invertible. The subset of the gapped element of ${\tt H}_{\eta,\Xi}(\n{C}^2)$ is denoted with
$$
{\tt H}_{\eta,\Xi}^G(\n{C}^2)\;:=\;\left\{H\in{\tt H}_{\eta,\Xi}(\n{C}^2)\;|\; {\rm det}\,H\neq0 \right\}.
$$
Since  $H\in {\tt H}_{\eta,\Xi}^G(\n{C}^2)$ implies that $\lambda_\pm\neq 0$ it follows that  one can decompose ${\tt H}_{\eta,\Xi}^G(\n{C}^2)$ into four disjoint components according to the signs of $\lambda_+$ and $\lambda_-$. Let introduce the subset
$$
\begin{aligned}
{\tt H}_{\eta,\Xi}^{(+,+)}(\n{C}^2)\;:&=\;\left\{H\in{\tt H}_{\eta,\Xi}(\n{C}^2)\;|\; \lambda_+>0\;, \lambda_->0 \right\}\;\simeq\;\left\{(u,v)\in\R^2\;|\; u>|v|\geqslant 0 \right\}\\
{\tt H}_{\eta,\Xi}^{(+,-)}(\n{C}^2)\;:&=\;
\left\{H\in{\tt H}_{\eta,\Xi}(\n{C}^2)\;|\; \lambda_+>0\;, \lambda_-<0 \right\}\;\simeq\;\left\{(u,v)\in\R^2\;|\; v>|u|\geqslant 0 \right\}\\
{\tt H}_{\eta,\Xi}^{(-,+)}(\n{C}^2)\;:&=\;
\left\{H\in{\tt H}_{\eta,\Xi}(\n{C}^2)\;|\; \lambda_+<0\;, \lambda_->0 \right\}\;\simeq\;\left\{(u,v)\in\R^2\;|\; v>-|u|\leqslant 0 \right\}\\
{\tt H}_{\eta,\Xi}^{(-,-)}(\n{C}^2)\;:&=\;
\left\{H\in{\tt H}_{\eta,\Xi}(\n{C}^2)\;|\; \lambda_+<0\;, \lambda_-<0 \right\}\;\simeq\;\left\{(u,v)\in\R^2\;|\; u<-|v|\leqslant 0 \right\}\\
\end{aligned}
$$
where the last isomorphisms are obtained in terms of the parametrization \eqref{eq:appB_repH_Xi_sym}
and the related description of the eigenvalues $\lambda_\pm=u\pm v$. It results that these four components are connected and contractible.  
 Summing up, we proved that:
 \begin{theorem}
 \label{theo:classif_nontriv_gap}
 The space ${\tt H}_{\eta,\Xi}^G(\n{C}^2)$ is the disjoint union of four connected and contractible component
 $$
{\tt H}_{\eta,\Xi}^G(\n{C}^2)\;=\;{\tt H}_{\eta,\Xi}^{(+,+)}(\n{C}^2)\;\sqcup\;{\tt H}_{\eta,\Xi}^{(+,-)}(\n{C}^2)\;\sqcup\;{\tt H}_{\eta,\Xi}^{(-,+)}(\n{C}^2)\;\sqcup\;{\tt H}_{\eta,\Xi}^{(-,-)}(\n{C}^2)\;.
 $$  
 Then ${\tt H}_{\eta,\Xi}^G(\n{C}^2)$ has the topology of a space of four points. The representative of each component can be chosen as follow: 
 $$
 {\tt H}_{\eta,\Xi}^{(\pm,\pm)}(\n{C}^2)\;\simeq\;\{\pm\n{1}\}\;,\qquad\quad {\tt H}_{\eta,\Xi}^{(\pm,\mp)}(\n{C}^2)\;\simeq\;\{\pm\Xi\}\;.
 $$
 \end{theorem}
\medskip

The two components ${\tt H}_{\eta,\Xi}^{(\pm,\pm)}(\n{C}^2)$ are called \emph{trivial} since describe 
gapped operators which can be always deformed to a multiple of the identity without violating the gap condition.  As a consequence the two components ${\tt H}_{\eta,\Xi}^{(\pm,\mp)}(\n{C}^2)$ are called \emph{non trivial}.

\begin{remark}[Comparison with the self-adjoint case]
\label{rk:more_structure}
{\upshape 
Let $H$ be a self-adjoint matrix on $\C^2$. Let $\lambda_1$ and $\lambda_2$ be the eigenvalues of $H$
and assume the \emph{non-trivial} zero gap condition  $\lambda_1 \lambda_2<0$. Let $P_1$ and $P_2$ be the spectral projections associated with the eigenvectors $\psi_1$ and $\psi_2$. Consider the continuous path
$$
H_t\;:=\;\sum_{j=1}^2\left[(1-t)\; \lambda_j+ t\;{\rm sgn}(\lambda_j)\right]\;P_j\;, \qquad \quad t\in[0,1]\;
$$
where ${\rm sgn}(\lambda_j):=\frac{\lambda_j}{|\lambda_j|}$.
Evidently $H_0=H$ and $\Gamma:=H_1$ verifies $\Gamma^2=\n{1}$, namely it is a gradation. With respect the basis $\{\psi_1,\psi_2\}$ the matrix $\Gamma$ reads
\begin{equation}
\Gamma\;=\; \left(\begin{array}{cc}{\rm sgn}(\lambda_1) & 0\\
0 & {\rm sgn}(\lambda_2)\end{array}\right)\;.
\end{equation}
Consider the unitari matrix
\begin{equation}
\label{eq:T-theta}
T_\theta\;:=\; \left(\begin{array}{cc}\cos\theta & \sin\theta\\
\sin\theta & -\cos\theta\end{array}\right)\;.
\end{equation}
A direct computation shows that $T_\theta=T_\theta^*=T_\theta^{-1}$. Moreover, one has that
$$
\Gamma_{\theta}\;:=\; T_\theta \Gamma T_\theta\;:=\;\left(\begin{array}{cc}{\rm sgn}(\lambda_1)\cos^2\theta+ {\rm sgn}(\lambda_2)\sin^2\theta& [{\rm sgn}(\lambda_1)-{\rm sgn}(\lambda_2)]\cos \theta \sin \theta \\
{[{\rm sgn}(\lambda_1)-{\rm sgn}(\lambda_2)]}\cos\theta \sin \theta  & 
{\rm sgn}(\lambda_1)\sin^2\theta+ {\rm sgn}(\lambda_2)\cos^2\theta
\end{array}\right)\;. 
$$
If one consider the  continuous path of $\Gamma_{\theta}$ with $\theta\in[0,\pi/2]$ one obtains that $\Gamma_{0}=\Gamma$ while 
\begin{equation}
\Gamma_{\frac{\pi}{2}}\;=\; \left(\begin{array}{cc}{\rm sgn}(\lambda_2) & 0\\
0 & {\rm sgn}(\lambda_1)\end{array}\right)\;.
\end{equation}
This computation shows that the \emph{non-trivial} component of the self-adjoint matrix gapped in zero is contractible to a single point represented by the standard gradation  $\eta$. On the contrary, Theorem  \eqref{theo:classif_nontriv_gap} proves that the set ${\tt H}_{\eta,0}$ has two non-trivial components ${\tt H}_{\eta,0}^{(\pm,\mp)}(\n{C}^2)$ represented by $\pm\eta$ respectively. In this case it is not possible to use the matrix $T_\theta$ to intertwine between $\eta$ and $-\eta$ since $\eta T_\theta\neq T_\theta \eta$ and the commutation with $\eta$ is a defining property for ${\tt H}_{\eta,0}$. This example shows that the indefinite metric structure induced by $\eta$ induces more structure with respect to the usual self-adjoint case. The difference between the self-adjoint case and the $\eta$-self-adjoint case
is a manifestation of the splitting phenomenon
\eqref{eq:splitting_trivial}.
}\hfill $\blacktriangleleft$
\end{remark}


\medskip
\medskip

\end{document}